\documentclass{article}

\usepackage[preprint]{neurips}

\usepackage[utf8]{inputenc}
\usepackage[T1]{fontenc}
\usepackage[bookmarks=true,colorlinks=true,linkcolor=cyan,citecolor=magenta]{hyperref}
\usepackage{etoc}
\usepackage{url}
\usepackage{booktabs}
\usepackage{amsmath, amssymb, amsfonts}
\usepackage{amsthm}
\usepackage{mathtools}
\usepackage{mathrsfs}
\usepackage{nicefrac}
\usepackage{microtype}
\usepackage{xcolor}
\usepackage{graphicx}
\usepackage{subcaption}
\usepackage{caption}
\usepackage{bm}
\usepackage{multirow}
\usepackage{upgreek}
\usepackage{paralist}
\usepackage{footnote}
\usepackage{verbatim}
\usepackage{comment}
\usepackage{stackrel}
\usepackage{multicol}
\usepackage{tcolorbox}
\usepackage{epstopdf}
\usepackage{algorithm}
\usepackage{algorithmic}

\usepackage[capitalize,noabbrev]{cleveref}

\theoremstyle{plain}
\newtheorem{theorem}{Theorem}[section]
\newtheorem{proposition}[theorem]{Proposition}
\newtheorem{lemma}[theorem]{Lemma}
\newtheorem{corollary}[theorem]{Corollary}
\theoremstyle{definition}
\newtheorem{definition}[theorem]{Definition}
\newtheorem{assumption}[theorem]{Assumption}
\theoremstyle{remark}
\newtheorem{remark}[theorem]{Remark}

\usepackage[textsize=tiny]{todonotes}

\usepackage{fontawesome5}

\definecolor{codebg}{RGB}{248,248,248}
\definecolor{codecolor}{RGB}{140,142,144}

\usepackage[compact]{titlesec}
\titlespacing{\section}{0pt}{12pt plus 2pt minus 2pt}{0pt plus 2pt minus 2pt}
\titlespacing{\subsection}{0 pt}{8pt plus 2pt minus 0pt}{0pt plus 2pt minus 2pt}
\title{HJ-Gauss: A Monte-Carlo HJ Reachability Scheme}

\author{%
	Lekan Molu \\
	Amazon IRG \\
	\And 
	Venkatraman Renganathan \\
	Cranfield University \\ 
	\And 
	Namhoon Cho \\
	Seoul National University \\
}

\definecolor{light-blue}{rgb}{0.30,0.35,1}
\definecolor{light-green}{rgb}{0.20,0.49,.85}
\definecolor{purple}{rgb}{0.70,0.69,.2}

\newcounter{mnote}

\newcommand{\ie}{i.e.,\ }
\newcommand{\eg}{e.g.\ }
\newcommand{\cf}{cf.\ }

\newcommand{\normop}[1]{|#1|}
\newcommand{\eqldefect}{\varepsilon_{\mathrm{QL}}}
\newcommand{\tubeband}{\mc{N}_\eta}

\newcommand{\mc}[1]{\mathcal{#1}}
\newcommand{\bb}[1]{\mathbb{#1}}

\newcommand{\shnote}[1]%
{\textcolor{magenta}{SH: #1}}
\newcommand{\lmnote}[1]%
{\textcolor{orange}{LM: #1}}

\newcommand{\Note}[1]{}
\renewcommand{\Note}[1]{\hl{[#1]}}  

\def\pde{PDE\ }
\def\sdf{\bm{\ell}}

\def\reline{\bb{R}}
\def\targetset{\mathcal{L}}
\def\traj{\xi}
\def\ren{\bb{R}^n}

\def\payoff{\Phi}
\def\state{\bm{x}}

\def\statey{\bm{y}}

\def\term{\sdf}

\def\pursuer{\bm{P}}
\def\evader{\bm{E}}

\def\valueparam{\bm{\omega}}
\def\control{\bm{u}}
\def\disturb{\bm{w}}
\def\switchcurve{\bm{\gamma}}
\def\valuefunc{\bm{v}}

\def\hamfunc{\bm{H}}

\def\uppervalue{\bm{v}^+}
\def\lowervalue{\bm{v}^-}
\def\upperham{\bm{H}^+}
\def\lowerham{\bm{H}^-}

\def\openset{\Omega}

\def\interface{\Gamma}
\def\flock{F}

\def\group{\mc{C}}
\def\subgroup{\mc{S}}
\def\climb{u_z}

\newcommand{\bc}{\bm{c}}
\newcommand{\bR}{\bm{R}}
\newcommand{\bx}{\bm{x}}
\newcommand{\by}{\bm{y}}

\newcommand{\bs}{\bm{s}}
\newcommand{\bv}{\valuefunc}
\newcommand{\bp}{\bm{p}}

\begin{document}

\maketitle

\begin{center}
	\vspace{-2.8em}
	\href{https://scriptedonachip.com/downloads/Papers/hjgauss_talk.html}{%
		\tcbox[
		on line,
		colback=codebg,
		colframe=codebg,
		coltext=codecolor,
		boxrule=0.4pt,
		arc=4pt,
		boxsep=2pt,
		left=4pt, right=4pt, top=2pt, bottom=2pt
		]{\fontfamily{fvm}\selectfont\small\faGlobe\enspace Slides}}
	\href{https://github.com/robotsorcerer/levelsetpy/tree/main/monte_carlo}{
		\tcbox[
		on line,
		colback=codebg,
		colframe=codebg,
		coltext=codecolor,
		boxrule=0.4pt,
		arc=4pt,
		boxsep=2pt,
		left=4pt, right=4pt, top=2pt, bottom=2pt
		]{\fontfamily{fvm}\selectfont\small\faGithub\enspace Code}}
	\href{mailto:ogunmolu@amazon.com}{%
		\tcbox[
		on line,
		colback=codebg,
		colframe=codebg,
		coltext=codecolor,
		boxrule=0.4pt,
		arc=4pt,
		boxsep=2pt,
		left=4pt, right=4pt, top=2pt, bottom=2pt
		]{\fontfamily{fvm}\selectfont\small\faEnvelope\enspace Email}}
	\vspace{-0.5em}\\
\end{center}

\begin{abstract}
	 Backward reachable sets or tubes (BRS/Ts), evaluated with grid-based level-set methods over viscous Hamilton-Jacobi (HJ) partial differential equations (PDEs), furnish principled reachability certificates for learning-enabled control problems. However, these methods incur an $O(M^n)$ memory cost, where $M$ is the number of grid points for every $n$-state dimension. This prohibitive storage cost has precluded their applications in many high-dimensional physical phenomena. Towards scalable reachability analysis, we  propose a frozen-coefficient Picard iterative Gaussian sampling scheme that reduces this exponential memory footprint to a linear one:  with a Cole-Hopf-type transformation, the HJ PDE effectively reduces to a sequence of linear heat equations, whose  values are then iteratively  recovered via  Gaussian heat-kernel expectations (using the Feynman-Kac formula). In this sentiment, Monte Carlo roll-outs over Gaussian densities  ultimately recover the (approximate) HJ value and its spatial gradient. Ours is a storage- and discretization-free  algorithm whose memory footprint scales (for  $N$ \textit{i.i.d} samples) as $N\cdot n$ under  a polynomial sampling power law in $N$; and provide a conditional linear convergence analysis to the \textit{consistent} viscosity solution. Furthermore, we provide a quasilinearization residual per iteration between frozen coefficient samples for this approximation scheme. Our theoretical machinery is  rigorously validated on  rocket launch and Dubins pursuit-evasion (P-E) games over  Holm-Bonferroni-informed experiments: for a $45D$-multi-rocket launch P-E game study, we find an \textit{almost zero}  Picard residual floor. Certifying the safety of \textbf{$100,000$} European starlings (\textit{sturnus vulgaris}) in murmurations  over their collective value functions reveals vacuole nucleation, cordon formation, and flock splitting on the resulting BRT zero-levelset phase topology as predator attacks evolve, demonstrating the scalability of our approach to large-scale multi-agent systems.
\end{abstract}


\etocdepthtag.toc{main}

\section{Introduction}
\label{sec:intro}

\noindent
A central challenge in learning-enabled control of complex high-dimensional systems is certifying that a learned controller, neural policy, or planning algorithm  generates evidence that it or any of its components satisfies all specified requirements including functional and allocated baselines in a \emph{verification} sense~\citep{DAUVerif}. As systems scale, algorithmic reliability  in reinforcement learning~\citep{SafeRL}, model-based~\citep{berkenkamp2017safe}, and other continuous, high-dimensional systems  necessitate scalable \emph{reachability} algorithms for system verification. Reachability concerns evaluating the \textit{decidability} of a dynamical systems' evolution of trajectories throughout a phase space. Decidable reachable systems are those where one can compute all states that can be reached from an initial condition in \textit{a finite number of steps}. Backward reachable (\textit{similarly, }reach-avoid) tubes \ie BRTs (BRATs), are the set of all states from which a system is guaranteed to reach (similarly avoid) a target region regardless of disturbances that may affect its behavioral evolution; when the controller must robustly counter a worst-case disturbance, then we must resolve its \emph{robustly controlled BRT (BRAT)}~\citep{Mitchell2020} \ie RCBRT/RCBRAT. Computing these amount to solving the Hamilton-Jacobi-Isaacs (HJI)~\citep{Isaacs1965} partial differential equation (PDE), under two competing inputs~\citep{Mitchell2005, LygerosReachability}.


\begin{figure}[tb]
	\centering
	\includegraphics[width=\columnwidth]{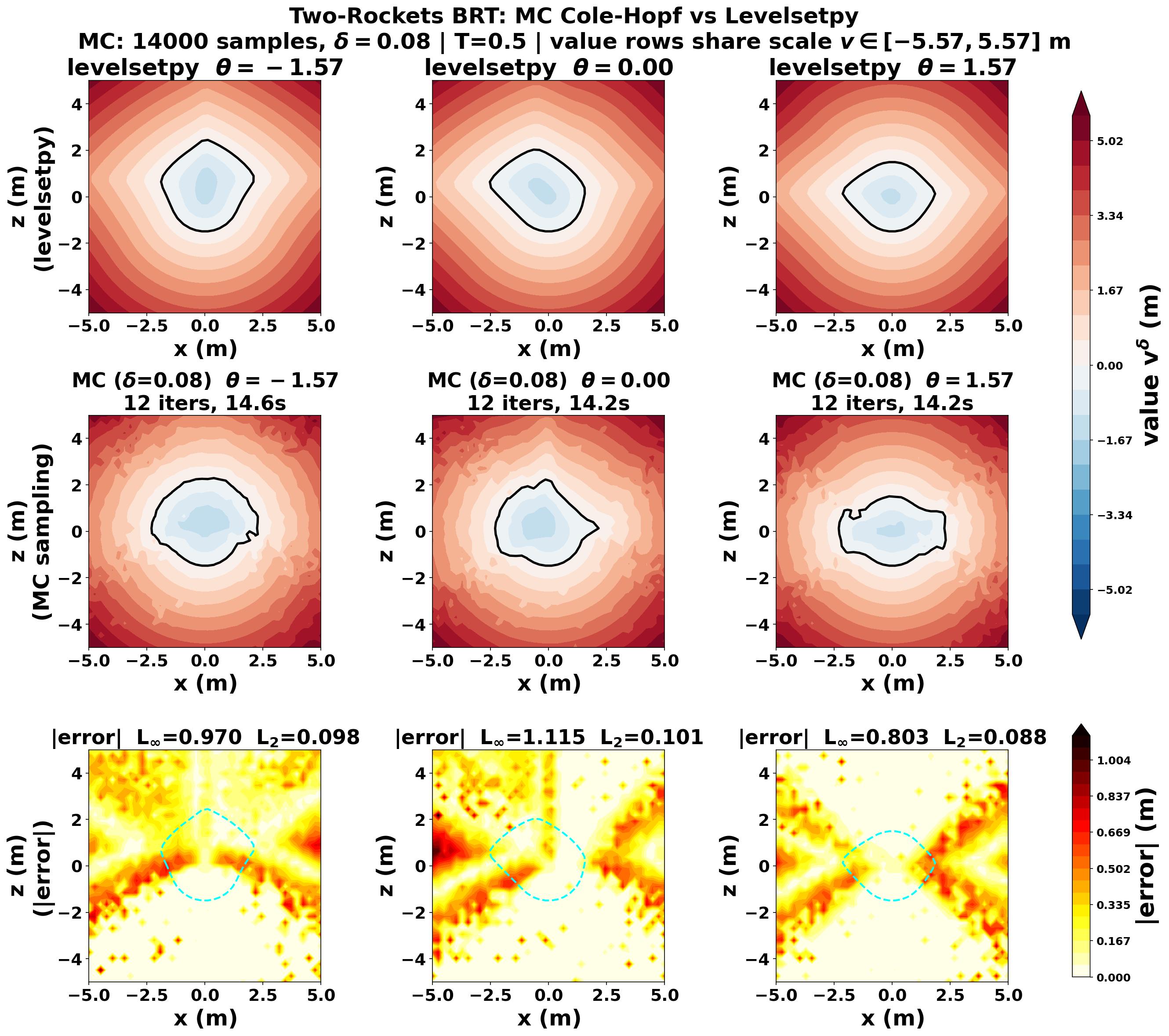}
	\caption{2D BRT slices for specific relative heading angles $\theta=\{-90^\circ, 0, 90^\circ \}$ of two rockets on a plane as derived in \citep{MoluxLevPyCDC}. The top row depicts the classical grid-based resolved levelsets using the classical grid-based solver in \citep{MoluxLevPyCDC} vs our Monte Carlo solver  (middle row). The top two rows plot the value function $\valuefunc^\delta$ (m) on a \emph{single shared, zero-centered} diverging color scale (common colorbar at right) so the two solvers are directly comparable; the blue-cyan-red contours signify  up to $20$ levelsets of divergence from the zero levelset (black contours), the field being negative inside the capture region (radius $1.5$) and positive outside. The bottom row is a heat map of the pointwise errors at each relative orientation $\theta$ (its own colorbar). Notice that errors are largest near the boundary of the usable part where the absolute gradient of the value function  $|D\valuefunc^\delta|$ is largest. This figure shows a single representative Monte Carlo seed for visual clarity.}
	\label{fig:rockets_slices}
\end{figure}

\begin{tcolorbox}[title=\textbf{Central Contribution: Monte Carlo Trajectory Optimization of Quasi-Linear BRTs}, colback=yellow!3, colframe=magenta!70, fonttitle=\small, fontlower=\small] 
	\small
	A \emph{backward reachable tube} (BRT) is the set of initial states from which a controlled system is guaranteed to reach a target set within a given time horizon, for all adversarial inputs. Given system dynamics $\dot{\state}=f(t;\state,\bm{u},\bm{w})$ with controller $\bm{u}$ and disturbance $\bm{w}$, the BRT is the zero sublevel set of a value function $\valuefunc^\delta$ that solves the viscous HJI PDE~\eqref{eq:HJ-RCBRT-Visc}. Classical grid-based solvers~\citep{Mitchell2005, LevelSetTOMS} discretize the state space on a grid with $M$ points per dimension; the resulting memory cost is $O(M^n)$, which grows exponentially in the state dimension $n$. For $n=6$ with $M=100$, this already requires $10^{12}$ grid cells. The proposed method replaces the grid with $N$ Monte Carlo samples of state space trajectories, reducing spatial memory to $O(N\cdot n)$.
\end{tcolorbox}

\noindent The major drawback for realizing the HJ reachability of complex systems is computational: classical HJ solvers scale as $O(M^n)$ in memory, where $M$ is the number of grid evaluation points per dimension and $n$ is the state dimension. While GPU-accelerated implementations~\citep{MoluxLevPyCDC, LevelSetTOMS} reduce wall-clock time, they retain the same exponential memory footprint. Self-supervised, PINN-based approaches such as DeepReach~\citep{DeepReach} train a neural network to minimize the HJ PDE residual directly via a physics-informed loss, without requiring a reference grid solution. However, the method scales  moderately in dimensions (reportedly 9D) and accuracy degrades in the limit of higher-dimensional data. We introduce a fundamentally different approach that replaces grid storage with a sampling scheme so that the memory cost is $O(N \cdot n)$ for $N$ Gaussian samples with a grid/discretization-free storage.

\noindent\textbf{Relation to grid-free and sampling-based HJ solvers.}
Three research families border ours: \emph{Hopf- and Lax-Oleinik-formula methods}~\citep{DarbonOsher2016, ChowDarbonOsherYin2017, KirchnerHopf2018} solve the HJ equation with grids, evaluating the value pointwise in convex-form; however, the Hopf formula demands a convex (or state-independent) Hamiltonian, whereas the state-dependent, nonconvex Hamiltonians of reachability fall outside that class. This paper introduces a quasi-linearization  scheme that trades their exactness for this generality at the quantified price of a defect bound (\cref{app:errors}). \emph{System decomposition}~\citep{ChenHerbertDecomp2018} methods remove dimension by splitting self-contained subsystems and is exact when the coupling structure permits, whereas our introduced sampler is indifferent to coupling structure and thus complements decomposition. \emph{Stochastic PDE control}~\citep{Kappen2005PI, TheodorouPI2010} is the stochastic sibling of our transformation: a similar logarithmic substitution linearizes the stochastic HJB equation under a noise-control duality, with the diffusion supplied by process noise; here, the viscosity $\delta$ is an analysis parameter of a deterministic worst-case game, and the frozen coefficient $\bc$ plays the duality's role. ~\citep{kharroubi2013numerical} leverage time-discretization to transform the nonlinear HJB to a form that admits Monte Carlo evaluations of a resulting backward stochastic differial equation via least squares regression. Finally, \emph{sample-based stochastic reachability}~\citep{SummersLygeros2010, LesserOishiErwin2013} certifies probabilistic reach-avoid for stochastic systems; we certify adversarial (worst-case) reachability, a different guarantee.

\noindent
We consider first-order nonlinear scalar HJ PDEs of the form
%
	\begin{align}
		\label{eq:ivp}
		\valuefunc_t + \hamfunc(t; \state, D \valuefunc) = 0 \text{   in   } \openset \times (0, T], \quad
		\valuefunc(0; \state) = \bm{g}(0; \state) \text{   on   } \openset \times \{t=0\}
		\tag{HJ-IVP}
	\end{align}
%
where the state $\state$ belongs to the open set $\openset \subseteq \ren$; $\valuefunc_t$ denotes the partial time derivative of the solution, $\valuefunc(t; \state)$; and $\hamfunc(t; \state, D \valuefunc)$ is the Hamiltonian, continuously defined on $\openset$, with $D \valuefunc$ representing the spatial gradient. The terminal value $\bm{v}(0; \state)$ is bounded and uniformly continuous (BUC) function. Note that $\bm{v}(0;\state)$ is the \emph{datum} of the initial-value problem \eqref{eq:ivp}, posed at $t=0$; in backward-reachability, it encodes the target set, and in physical time $t_{\mathrm{phys}}=T-t$ it is the \emph{terminal} cost. Numerically, it is realized by a signed-distance cost $\sdf(\state)$ whose zero sublevel set is the target set, so that $\bm{g}(0;\state)\equiv\sdf(\state)$; we reserve $\bm{g}$ for the abstract initial datum and $\sdf$ for its signed-distance instantiation.

Fixing a viscosity parameter $\delta>0$,~\citep{CrandallTwoApprox} introduced the \emph{vanishing viscosity} solution, $\valuefunc^\delta$, of \eqref{eq:ivp} \ie,
%
	\begin{align}
		\label{eq:ivp-viscous}
		\valuefunc_t^\delta + \hamfunc(t; \state, D \valuefunc^\delta) = \dfrac{\delta}{2} \Delta \valuefunc^\delta \text{ in } \openset \times (0, T]; \quad
		\valuefunc^\delta(0; \state) = \bm{g}(0; \state) \text{ on } \openset \times \{t=0\}
		\tag{HJ-Visc}
	\end{align}
%
such that \eqref{eq:ivp-viscous} satisfies the uniform convergence bound
\begin{align}
	\label{eq:visc_conv}
	\sup_{t \in (0, T]} \sup_{\state\in \ren} \left| \valuefunc(t; \state) - \valuefunc^\delta(t; \state)\right| \le k\sqrt{\delta}
\end{align}
for a constant $k>0$. \textit{Key insight: \eqref{eq:ivp-viscous} can be reduced, via a generalized Cole-Hopf-type transformation, to a sequence of linear heat equations admitting explicit Gaussian heat-kernel solutions.} This yields a locally linearized PDE, \ie an iterative (Picard) approximation in which the nonlinear coefficient is frozen at each step and updated after the linear solution is found. The exact reduction holds for the quadratic case $\hamfunc = \tfrac{1}{2}|p|^2$, where the frozen coefficient is constant and the residual vanishes identically.


\noindent\textbf{Connections to ML.}
Beyond reachability, the proposed scheme connects to several active ML research directions.
\begin{inparaenum}[(i)]
	\item \emph{Safe RL and policy verification}: certifying that a learned policy satisfies safety constraints requires computing BRTs for the closed-loop system; the memory bottleneck of grid-based solvers is the primary obstacle to safe RL in high-dimensional spaces~\citep{SafeRL, berkenkamp2017safe}.
	\item \emph{Model-based RL}: planning algorithms that learn a dynamics model and then verify safety via HJ reachability are limited to low-dimensional state spaces~\citep{DeepReach}.
	\item \emph{Diffusion-based methods}: the Gaussian heat-kernel representation derived here is structurally similar to the nonconvex optimizers for retrieving global minima of differentiable objectives~\citep{Chaudhari2018, HJMAD}, which may be extended to Stein score-based diffusion generative models.
\end{inparaenum}

\noindent\textbf{Contributions.}
Our contributions are as follows.
\begin{inparaenum}[(i)]
	\item A generalized Cole-Hopf-type transformation decouples the nonlinear viscous HJ equation into linear heat equations with Gaussian heat-kernel solutions; this enables a frozen-coefficient quasi-linearization.
	\item A sampling-based algorithm (\autoref{alg:quasi_lin}) with $O(N \cdot n)$ memory complexity that scales with sample count rather than discretization dimension, offers a viable alternative to impractical grid-based $O(M^n)$ approaches in high dimensions.
	\item Theoretical guarantees:
	\begin{inparaenum}[(a)]
		\item a finite-sample concentration bound of $O(N^{-1/2})$ for the Monte Carlo value estimator, independent of state dimensions (Theorem~\ref{thm:mc_complexity}),
		\item conditional linear convergence of the Picard iteration with explicit contraction constant  (Theorem~\ref{thm:algorithm_convergence}),
		\item a Duhamel bound on the quasi-linearization defect separating the iteration's fixed point from the viscous solution (Theorem~\ref{thm:defect}), and
		\item a conservative safety certificate (Corollary~\ref{cor:conservative_cert}).
	\end{inparaenum}
	\item Validation on two-player pursuit-evasion games achieving $L^2_{\text{rel}}$ errors of $0.03$–$0.20$ against grid-based benchmarks~\citep{LevelSetPy, MoluxLevPyCDC, LevelSetTOMS, MitchellLSToolbox, Mitchell2005};
	\item Demonstration on 45-dimensional multi-agent problems where grid-based discretization are currently computationally prohibitive;
	\item Safety certification of starling murmurations at the $100,000$-agent scale via decoupled per-flock value functions (\autoref{app:starlings}).
\end{inparaenum}

The rest of this work is structured as follows. The uninitiated in HJ verification theory may consult the background in~\autoref{app:back_unabridged}. Section~\ref{sec:methods} derives the quasi-linearization of \eqref{eq:ivp-viscous}, states the algorithm, and establishes the theoretical guarantees. Section~\ref{sec:results} presents numerical results and  Section ~\ref{sec:conclude} concludes the paper. All proofs appear in \autoref{app:proofs}. In \autoref{app:errors}, we establish a rigorous analysis of the error bounds, convergence rates, and robustness properties of the quasi-linearized Cole-Hopf transformation scheme. Lastly, additional numerical experiments are provided in \autoref{app:results}. 

\section{Transformation of the HJ PDEs} 
\label{sec:methods}
We now transform the nonlinear viscous HJ PDE \eqref{eq:ivp-viscous} into a linearized form and extract its solution via Gaussian heat-kernel expectations (all proofs are in \autoref{app:proofs}). We succinctly state definitions that aid the construction of our results (with details in \autoref{app:back_unabridged}). 

\begin{savenotes}
	\begin{algorithm}[tb!]
		\caption{Quasi-Linearization Algorithm   Cole-Hopf
			\label{alg:quasi_lin}}
		\begin{algorithmic}
			\STATE Fix: $\epsilon > 0$, $\bv^{(0)}(t;\bx) = g(\bx)$, \;
			$\bc^{(0)} = \frac{2\,\hamfunc(t;\bx,D\bm{g})}{{\delta\,|D\bm{g}|^2}}$.
			\STATE \textbf{For $k = 0,1,2,\ldots$:}
			\begin{enumerate}
				\STATE Freeze $\bc^{(k)}$ at the current iterate.
				\STATE Solve the heat equation
				$\valueparam_t = \frac{\delta}{2}\Delta\valueparam$
				with initial data
				$\valueparam^{(k)}(0;\bx) = e^{-c^{(k)} g(\bx)}$.
				\STATE Recover $\bv^{(k+1)} = -(1/\bc^{(k)})\log\valueparam^{(k+1)}$.
				\STATE Update $D\bv^{(k+1)}$ and $\bc^{(k+1)} = \frac{2\hamfunc(t;\bx,D\bv^{(k+1)})}{\delta|D\bv^{(k+1)}|^2}$.
				\STATE Check convergence: $\|\bv^{(k+1)} - \bv^{(k)}\|/\|\bv^{(k)}\| < \varepsilon$.
			\end{enumerate}
		\end{algorithmic}
	\end{algorithm}
\end{savenotes}
	We adopt the backward-reachability viscosity-solution convention,
	\begin{align}
	\label{eq:sign_convention}
	\valuefunc_t + \hamfunc(t; \state, D\valuefunc) = 0, \qquad 
	\hamfunc(t; \state, p) = \max_{\bm{u}} \min_{\bm{w}} \langle p, f(t; \state, \bm{u}, \bm{w}) \rangle.
	\end{align}
	throughout, with Hamiltonian $\hamfunc$ defined via the zero-sum differential game formulation of \autoref{app:back_unabridged}.
	%
	%
	\emph{Time convention}: $t$ denotes the \emph{backward horizon} (time-to-go); the target data is posed at $t = 0$ and the tube is grown over $t \in (0,T], \, T>0$. 

\begin{definition}[The Target Set]
	The invariant set $\targetset_0$ obtained at $T$ is target set,
	\begin{align}
		\mathcal{L}_0(T) = \{ \state \in {\bb{R}}^n \,|\, \valuefunc(0; \state) \le 0 \}
		\tag{Target-Set}
		\label{eq:target_set}
	\end{align}
	that is ``\textit{robustly controlled}" over the ``\textit{distance-to-target-set}" cost $\bm{g}(0; \state)$. 
\end{definition}

\begin{definition}[The RCBRT]
	The \textit{robustly controlled backward reachable \textbf{tube}} (RCBRT)~\citep{Mitchell2020} on $(0, T]$ 
	is the closure of the open set,
	\begin{align}
		\mathcal{L}([-T, 0], \mathcal{L}_0) = \{\state \in \ren \,| \, \exists \, \beta \in \mathcal{B}(t)  \forall \, \bm{u} \in \mathcal{U}(t), \,
		  \exists \,
		\bar{\tau} \in [-T, 0], \bm{\xi}(\bar{\tau}) \in  \mathcal{L}_0 \}.
		\tag{Target-Tube}
		\label{eq:rcbrt}
	\end{align}
	with measurable functions, $\bm{u}: [t, T]  \rightarrow \mathcal{U}, \bm{w}: [t, T] \rightarrow \mathcal{W}$, where $\mathcal{U} \in \bb{R}^m$ and $\mathcal{W} \subset \reline^p$ are compact sets.
\end{definition}

\begin{definition}[The HJI-RCBRT]
	The value function of the RCBRT is the viscosity solution of
	\begin{align}
		\valuefunc_t(t; \state) + \min\{0, \hamfunc\left(t; \state, D \valuefunc(t; \state)\right)\} = 0, \quad
		\valuefunc(0; \state) = \bm{g}(0; \state).
		\tag{HJI-RCBRT}
		\label{eq:HJ-RCBRT}
	\end{align} 
\end{definition}

\begin{definition}[The Zero LevelSet]
	
	\begin{subequations}
		\begin{align}
			\state(t)\in \mathcal{L}(\cdot) \implies \valuefunc(t; \state) \le 0, \,\,
			\valuefunc(t; \state) \le 0 \implies \state(t) \in \mathcal{L}(\cdot). 
			\tag{Zero Levelset}
			\label{eq:reachables} 
		\end{align}
	\end{subequations}
\end{definition}

\subsection{Exact Reduction and Quasi-Linearization}
\label{sec:exact_vs_approx}

Let us first draw the distinction that governs the interpretation of our contribution.

\noindent\textbf{Quasi-linearization.}
For a general Hamiltonian $\hamfunc(t;\state,p)$, define the spatially-varying coefficient,
\begin{align}
	\label{eq:coeff_frozen}
	\bc(t; \state) = \frac{2}{\delta}\cdot\hamfunc^\delta/|D\valuefunc^\delta|^2,
\end{align}
where $\hamfunc^\delta := \hamfunc(t;\state,D\valuefunc^\delta)$. With $\bc$ spatially varying, $\valueparam^\delta := \exp(-\bc\valuefunc^\delta)$ applied to \eqref{eq:ivp-viscous} yields a heat equation and the residual $\bR = \bR_{\mathrm{alg}} + \bR_{\mathrm{der}}$, made up of an algebraic ($\bR_{\mathrm{alg}}$) and derivative component ($\bR_{\mathrm{der}}$). The \emph{algebraic} part is eliminated  by the choice \eqref{eq:coeff_frozen}, leaving $\bR_{\mathrm{der}}$ composed of $\bc_t$, $D\bc$, and $\Delta\bc$. 
Ours is a \emph{Picard quasi-linearization} scheme \ie a successive quadratic matching: at iteration $k$, the coefficient $\bc^{(k)}$ is \emph{frozen} at the current iterate, the linear heat equation is solved \emph{exactly} for $\valueparam^{(k+1)}$. Equivalently, the viscous HJ equation with the locally-matched quadratic surrogate Hamiltonian $\tilde{\hamfunc}^{(k)} = \tfrac{\delta}{2}\bc^{(k)}|p|^2$ recovers  $\valuefunc^{(k+1)}$, and $D\valuefunc^{(k+1)}$ is updated  to recompute $\bc^{(k+1)}$. \emph{The iteration's limit is the fixed point of this surrogate solve map, \emph{not}, in general, the viscous solution itself.}

\noindent\textbf{Exact case ($\hamfunc = \tfrac{1}{2}|p|^2$).}
When the Hamiltonian is purely quadratic in the co-state, setting $\bc = 1/\delta$ (a constant) and $\valueparam^\delta := \exp(-\valuefunc^\delta/\delta)$ is an \emph{exact} Cole-Hopf transformation: $\valueparam^\delta$ satisfies the homogeneous heat equation $\valueparam^\delta_t = (\delta/2)\Delta\valueparam^\delta$ with no residual. 

Theorem~\ref{thm:defect} (elucidated in {\autoref{app:errors}) bounds the gap between the two by the Duhamel norm of the discarded residual, and the bound degrades where $\bc$ varies rapidly. The seeming circularity in \eqref{eq:coeff_frozen} is resolved by the frozen-coefficient interpretation and this is stated in Remark~\ref{rem:exact_vs_approx}.

\subsection{HJ Linearization}
\label{sec:hj_lin}

Henceforth, we replace $\hamfunc(t; \state, D \valuefunc^\delta)$ with $\hamfunc^\delta$ for brevity and we write $\bm{g}(\state)$ for the terminal cost $\bm{g}(0;\state)$. With $\bc^{(k)}$ frozen, the transformation,
\begin{align}
	\label{eq:cole-hopf}
	\valueparam^\delta := \exp(-\bc \valuefunc^\delta),
	\tag{Linear-Trans}
\end{align}
reduces \eqref{eq:ivp-viscous} to a heat equation under the frozen-coefficient approximation introduced in~Proposition~\ref{prop:heat_transformation} below. \textit{For general nonlinear Hamiltonians, this transformation induces a residual term and should be interpreted as a quasi-linearization}, \ie the approximation is exact only when $\hamfunc^\delta$ is quadratic in $p$ (see \S\ref{sec:exact_vs_approx}).
\begin{proposition}[Heat-equation solve under frozen coefficient]
	\label{prop:heat_transformation}
	With $\bc^{(k)}$ frozen at the current iterate, the linear step of Algorithm~\ref{alg:quasi_lin} defines the transformed variable $\valueparam^\delta$ as the solution of the heat initial-value problem below. When $\hamfunc=\tfrac{1}{2}|p|^2$ and $\bc=1/\delta$,  $\valueparam^\delta = \exp(-\bc\valuefunc^\delta)$ is exact; for general Hamiltonians it discards the derivative residual of Lemma~\ref{lem:exact_residual} in \autoref{app:proofs}. Thus $\valueparam^\delta$ satisfies,
	\begin{align}
	\label{eq:Heat-Equation}
	\valueparam_t^\delta - \dfrac{\delta}{2} \Delta \valueparam^\delta = 0
	\text{ in } \ren \times (0, T], \quad \valueparam^\delta = \exp\left(-\bc \bm{g}(\bm{y})\right) \text{ on } \ren \times \{t=0\}.
	\tag{Linear-HJ}
	\end{align}
	The solution to \eqref{eq:Heat-Equation} admits the following unique, explicit Green's convolution representation,
	\begin{align}
	\label{eq:heat_kernel}
	\valueparam^\delta(t; \state) &= \frac{1}{(\sqrt{2 \pi \delta t})^n} \int_{\ren} \exp \left(-\frac{|\state - \bm{y}|^2}{2 \delta t}\right)
	\exp\left(-\bc\bm{g}(\bm{y})\right) d\bm{y}, \quad \state \in \ren,\; t \in (0,T],
	\end{align}
	which, via the Feynman-Kac formula, can equivalently be written as a Gaussian expectation,
	\begin{align}
	\label{eq:heat_kernel_expectation}
	\valueparam^\delta(t; \state) = \bb{E}_{\statey \sim \mc{N}(\state, \delta tI_n)} \left[\exp\left(-\bc\bm{g}(\bm{y})\right)\right].
	\end{align}
\end{proposition}

\begin{lemma}[Smoothed HJ Solution]
	\label{lemma:hj_solution}
	Under the frozen-coefficient interpretation of \eqref{eq:coeff_frozen} (or exactly when $\hamfunc=\tfrac{1}{2}|p|^2$), the solution to the smoothed Hamilton-Jacobi equation \eqref{eq:ivp-viscous} is,
	\begin{align}
		\label{eq:ivp-viscous-solution}
		\valuefunc^\delta(t; \state) &= -\frac{1}{\bc^{(k)}}\log \bigg\{\dfrac{1}{(\sqrt{2\pi\delta t})^n}\int_{\ren} \exp \left(-\frac{|\state-\bm{y}|^2}{2\delta t}\right)
		\exp\left(-\bc^{(k)} \bm{g}(\bm{y})\right) d\by\bigg\},
	\end{align}
	in $\ren \times (0, T]$, where $\bc^{(k)}$ is the frozen coefficient at iteration $k$. This can equivalently be written as,
	\begin{align}
	\label{eq:hj_sol}
	\valuefunc^\delta(t; \state) &= -\frac{1}{\bc^{(k)}} \cdot \log \bb{E}_{\statey \sim \mc{N}(\state, \delta tI_n)} \left[\exp\left(-\bc^{(k)} \bm{g}(\bm{y})\right)\right].
	\end{align}
\end{lemma}
\begin{remark}[Admissible data on $\ren$]
	\label{rem:admissible_data}
	The classical whole-space theory requires the datum of the heat problem to lie in $C(\ren) \cap L^\infty(\ren)$~\citep{EvansPDEBook}. Equations \eqref{eq:heat_kernel}-\eqref{eq:heat_kernel_expectation} are the fundamental-solution formula for the Cauchy problem on $\ren$, and the free-space kernel carries unit mass over $\ren$ alone. The datum of \eqref{eq:Heat-Equation} is $\valueparam^\delta(0;\cdot) = \exp(-\bc\bm g)$ and not $\bm g$ itself, which matters because the signed-distance instantiation $\sdf$ is unbounded above on $\ren$. For a bounded target set, $\bm g$ is bounded below, say $\bm g \ge g_{\min}$, and $\bm g(\bm y) \to +\infty$ as $|\bm y| \to \infty$, hence
	\begin{align}
		\label{eq:datum_bound}
		0 < \exp(-\bc\bm g(\bm y)) \le \exp(-\bc g_{\min}) \quad \text{for all } \bm y \in \ren, \qquad \exp(-\bc\bm g(\bm y)) \to 0 \text{ as } |\bm y| \to \infty.
	\end{align}
	The transformed datum is therefore continuous, strictly positive, bounded and decaying on $\ren$, so \eqref{eq:heat_kernel} is the unique bounded solution of \eqref{eq:Heat-Equation}, the integral converges absolutely, and \eqref{eq:heat_kernel_expectation} is a genuine Gaussian expectation. Consequently the estimator that draws $\statey \sim \mc{N}(\state, \delta t I_n)$ over all of $\ren$ in \eqref{eq:hj_sol} and \eqref{eq:spatial_grad} is unbiased for the quantity inside the logarithm. Only the lower bound $\bm g \ge g_{\min}$ is needed here, \ie no upper bound on $\bm g$ is used.
\end{remark}
%

\begin{corollary}[Spatial gradients of the HJ solution]
	\label{lemma:grad_hj_sol}
	Under the same frozen-coefficient assumption as Lemma~\ref{lemma:hj_solution}, the spatial gradient of the value function $\valuefunc^\delta(t;\state)$ admits the form,
	\begin{align}
	\label{eq:spatial_grad}
	D\valuefunc^\delta &= \dfrac{1}{t \cdot \delta \cdot \bc^{(k)}}
	\left(\state - \dfrac{\bb{E}_{\statey \sim\mc{N}(\state, \delta tI_n)}\left[\statey \cdot \exp\!\left(-\bc^{(k)} \bm{g}(\bm{y})\right)\right]}{\bb{E}_{\statey \sim\mc{N}(\state, \delta tI_n)}\left[\exp\!\left(-\bc^{(k)} \bm{g}(\bm{y})\right)\right]}\right),
	\end{align}
	where $\statey \sim \mc{N}(\state, \delta tI_n)$, \ie, with mean $\state$ and covariance $\delta t I_n$.
\end{corollary}
\begin{remark}[Exact reduction vs.\ quasi-linearization]
	\label{rem:exact_vs_approx}
	Lemmas~\ref{lemma:hj_solution} and Corollary~\ref{lemma:grad_hj_sol} are exact when $\hamfunc = \tfrac{1}{2}|p|^2$, in which case $\bc = 1/\delta$ is a constant and the residual $\bR(t;\state)$ in the linearized equation vanishes identically (see \autoref{app:proofs}). For general Hamiltonians, $\bc^{(k)}$ depends on the unknown $D\valuefunc^\delta$; we resolve this by Algorithm~\ref{alg:quasi_lin}, whereupon $\bc^{(k)}$ is formed from the previous iterate, frozen, and updated after each linear solve. The formulas in both lemmas apply at each iteration with $\bc^{(k)}$ in place of $\bc$.
\end{remark}
\begin{remark}[Evaluation states vs.\ Monte Carlo samples]
	\label{rem:eval_vs_mc_samples}
	Note that $N$ is the i.i.d.\ Gaussian draws $\statey_1,\ldots,\statey_N \sim \mc{N}(\state,\delta t I_n)$ used \emph{independently at each} evaluation state $\state$ to form the Monte Carlo estimators in \eqref{eq:hj_sol}, \eqref{eq:spatial_grad}, and the concentration bound of Theorem~\ref{thm:mc_complexity} below. The $N$ draws are resampled fresh at every evaluation state and every Picard iteration of Algorithm~\ref{alg:quasi_lin}; they are not the $M$ evaluation states, $\state_1,\ldots,\state_M \in \openset$ where $\valuefunc(t; \state)$ is queried; carrying no randomness, they may be laid out on a uniform grid or scattered arbitrarily, since the \emph{solve} is grid-free. It requires no state space discretization storage or marching cubes. Hence, computational cost of the total randomness consumed per iteration is $M \times N$ draws. The zero sublevel set $\mc L_0(T)$ of \eqref{eq:HJ-RCBRT-Visc} is recovered by evaluating $\valuefunc^\delta$ at $M$ states ~(\cref{alg:quasi_lin}) and isocontouring the resulting array at level zero\footnote{This is a deterministic post-processing step with marching squares/cubes that consumes no further samples.}. 
\end{remark}

%
%
%

\subsection{Relations to Backward Reachable Sets/Tubes}
Introducing the viscosity therm, $\delta$, we may write \eqref{eq:HJ-RCBRT} as 
\begin{align}
	\valuefunc_t^\delta(t; \state) + \min\{0, \hamfunc^\delta\} = \dfrac{\delta}{2}\Delta \valuefunc^\delta, \quad \valuefunc^\delta(0; \state) = \bm{g}(0; \state),
	\tag{HJI-RCBRT-Visc}
	\label{eq:HJ-RCBRT-Visc}
\end{align}
 where $\delta$ smooths $\valuefunc(t; \state)$.  The invariant (target) set $\targetset_0$ obtained at $T$ is 
\begin{align}
	\mathcal{L}_0(T) = \{ \state \in \openset \,|\, \valuefunc^\delta(0; \state) \le 0 \}.
\end{align}
Similarly, the viscous version of the \textit{reach-avoid BRT} \ie \eqref{eq:HJI-RCBRAT} can be represented as
\begin{align}
	\min\{\valuefunc_t^\delta(t; \state) + \hamfunc^\delta - \dfrac{\delta}{2}\Delta \valuefunc^\delta, \bm{g}(t; \state) - \sdf(t; \state)\} \le 0,  \,\, \valuefunc^\delta(0; \state) = \bm{g}(0; \state).
	\tag{HJI-RCBRAT-Visc}
	\label{eq:HJI-RCBRAT-Visc}
\end{align}

\subsection{Weighted Importance Sampling of Reachable Sets}
While we may not recover the exact value function, the mean and variance of its derivatives can be evaluated with the kernel expectations of Lemma \ref{lemma:hj_solution}, and Corollary \ref{lemma:grad_hj_sol}. Previously, self-supervised physics-informed neural network (PINN) approaches such as DeepReach~\citep{DeepReach} train neural networks to minimize the HJ PDE residual directly without requiring external supervision or a reference grid solution. However, neural network-based approximation methods encounter inherent scaling limitations; DeepReach scales only to moderate dimensions (reported to 9D/10D in~\citep{DeepReach}), and approximation accuracy degrades as dimension increases. \textit{Our sampling-based scheme addresses this bottleneck by decoupling memory cost from dimensionality. }

The estimators of Lemma \ref{lemma:hj_solution}  and Corollary~\ref{lemma:grad_hj_sol} is a ratio of an exponential-weight expectation; in high dimensions,  $\bc\,\bm g$ is large over the kernel's support so that the denominator $\bb{E}[e^{-\bc\bm g}]$ is carried by a rare event where the weights collapse onto a few samples, and the ratio's variance explodes. The remedy is to move the samples to critical mass of the weight. We \emph{tilt} the proposal: in place of $\statey \sim \mc{N}(\state, \sigma^2 I_n)$, $\sigma^2 \triangleq \delta t$, we draw $\statey \sim q_\theta \triangleq \mc{N}(\state + \theta, \sigma^2 I_n)$ and reweight by the exact density ratio $w_\theta(\statey) = \exp\left((|\theta|^2 - 2\langle\statey - \state, \theta\rangle)/(2\sigma^2)\right)$, which leaves every expectation unbiased. The zero-variance proposal is proportional to the integrand $\varphi_{\state}(\statey)e^{-\bc\bm g(\statey)}$ itself; its (first-order) Laplace  Gaussian addresses the shift,
\begin{align}
	\label{eq:tilt_shift}
	\theta^{\star} = -\sigma^2\, \bc^{(k)}\, D\valuefunc^{(k)}(t;\state),
\end{align}
so that samples are pushed one preconditioned gradient step along the descent direction of the running value iterate. The previous Picard iterate hands the sampler the drift it needs, at no extra cost. The resulting self-normalized estimators is monitored by the effective sample size, $\mathrm{ESS} \triangleq \left(\sum_i \tilde w_i\right)^2/\sum_i \tilde w_i^2$ with $\tilde w_i \triangleq w_{\theta}(\statey_i)e^{-\bc\bm g(\statey_i)}$: an $\mathrm{ESS}$ near $N$ certifies the expectations; an $\mathrm{ESS}$ collapse flags exactly the weight degeneracy described above. 

The tilted counterparts of Lemma \ref{lemma:grad_hj_sol} is stated as Corollary~\ref{cor:tilted_estimators} in \autoref{app:proofs}.
The quasi-linearization, essentially a Picard fixed-point iterative scheme, is stated in \autoref{alg:quasi_lin}. The heat-kernel expectation is estimated via Monte Carlo, while the resulting log-sum-exp estimator for the value function enjoys the sample guarantee given in \S \ref{sec:mc-conv}.

\subsection{Sampling Complexity and Convergence Guarantees}
\label{sec:mc-conv}

In this sub-section, we analyze the sample complexity of the scheme with coefficient $\bc$ frozen and then generate a conditional linear-convergence guarantee as a fixed-point iteration on a finite collection of evaluation states.

\begin{theorem}[Finite-sample concentration of the frozen-coefficient value estimator]
\label{thm:mc_complexity}

Fix phase $(t;\state) \in (0,T] \times \openset$ and a frozen coefficient $\bc > 0$ (we treat the case $\bc < 0$ and $\bc = 0$ later  in Remark~\ref{rem:coeff_sign}).
Let $\sigma \triangleq \sqrt{\delta t}, 
\bm{\zeta} \sim \mc{N}(\state, \delta tI_n)$, 
and assume that the terminal cost is bounded on the sampling support, i.e. there exist constants $g_{\min} \le g_{\max}$ such that $
 g_{\min} \le \bm g(\bm{\zeta}) \le g_{\max}$ for almost every $\bm{\zeta} \in \ren$.
 
Let $Z \triangleq \exp(-\bc \, \bm g(\bm{\zeta})),
\mu \triangleq \bb E[Z],
\valuefunc_{\bc}(t;\state) \triangleq -\frac{1}{\bc}\log \mu$.
Given i.i.d. samples $\bm{\zeta}_1,\ldots,\bm{\zeta}_N \sim \mc{N}(\state, \delta tI_n)$, let
$ Z_i \triangleq \exp(-\bc \, \bm g(\bm{\zeta}_i))$,
$\bar{Z}_N \triangleq \frac{1}{N}\sum_{i=1}^N Z_i$, and  
$\hat{\valuefunc}_{\bc,N}(t;\state) \triangleq -\frac{1}{\bc}\log \bar Z_N$.
Let $\alpha \triangleq e^{-\bc g_{\max}}, \, \beta \triangleq e^{-\bc g_{\min}}.$
Then $\alpha \le Z_i \le \beta$ almost surely, and for all $\varepsilon > 0$,
\begin{align}
\label{eq:mc-concentration}
\bb P&\left(\left|\hat{\valuefunc}_{\bc,N}(t;\state) - \valuefunc_{\bc}(t;\state)\right| \ge \varepsilon\right)
\le  2\exp\left(-\frac{2N\mu^2(1-\exp(-\bc\varepsilon))^2}{(\beta-\alpha)^2}\right).
\end{align}
\end{theorem}
\begin{remark}
	\label{rem:which_bound}
	Since the sampling measure is the full Gaussian on $\ren$ (Remark~\ref{rem:admissible_data}), the two sides of $g_{\min} \le \bm g \le g_{\max}$ play different roles. The tail bound \eqref{eq:mc-concentration} uses only $\alpha \le Z_i \le \beta$, and for a bounded target set $\bm g \ge g_{\min}$ always holds, so $g_{\max} = +\infty$ with $\alpha = 0$ is admissible there and leaves the finite range $\beta - \alpha = e^{-\bc g_{\min}}$. 
\end{remark}
\begin{remark}[Tightness via Bernstein and Jensen bounds]
\label{rem:tightness-bernstein-jensen}
The Hoeffding bound in \eqref{eq:mc-concentration} is asymptotically tight but ignores variance. When $\mathrm{Var}(Z) \ll (\beta-\alpha)^2/4$, a Bernstein tail bound yields strictly tighter concentration. 
\end{remark}

Next, we introduce a corollary that  lower-bounds $\mu \ge \alpha$; it  needs a finite $g_{\max}$ over $\openset$ carrying the sampling mass. This is exactly the constant we do not control uniformly in $\delta$ (see  \autoref{sec:limitations}). Corollary~\ref{cor:mc_complexity} uses the loose lower bound $\mu \ge \alpha = e^{-\bc g_{\max}}$; when a sharper estimate (e.g., via Jensen, $\mu \ge e^{-\bc\bb{E}[\bm{g}(\bm\zeta)]}$) is available, the required sample size can be substantially reduced.

\begin{corollary}[Explicit sample size independent of $\mu$]
	\label{cor:mc_complexity}
	In particular, since $\mu \ge \alpha$, it is sufficient to choose
	\begin{align}
		\label{eq:mc-cor-N}
		N \ge \frac{(\beta-\alpha)^2}{2\alpha^2(1-\exp(-\bc\varepsilon))^2}\log\left(\frac{2}{\alpha}\right)
	\end{align}
	to guarantee $
	\bb P\left(\left|\widehat{\valuefunc}_{\bc,N}(t;\state) - \valuefunc_{\bc}(t;\state)\right| \ge \varepsilon\right) \le \alpha$.
\end{corollary}
\begin{remark}[The bound in \eqref{eq:mc-cor-N} is exponential in $1/\delta$]
	\label{rem:sample-bound-exponential}
	With $\bc=1/\delta$ (the exact quadratic case), $\beta/\alpha = \exp\!\big((g_{\max}-g_{\min})/\delta\big)$, so the right-hand side of \eqref{eq:mc-cor-N} grows as $\exp\!\big(2(g_{\max}-g_{\min})/\delta\big)$ \ie exponentially in $1/\delta$; while the Crandall-Lions viscosity error is only $O(\sqrt\delta)$: the accuracy gain from shrinking $\delta$ and the sample cost of Corollary~\ref{cor:mc_complexity} pull in opposite directions, and this bound does not certify a favorable trade-off $\delta$. As Remark~\ref{rem:tightness-bernstein-jensen} above notes, the Hoeffding bound is loose; \autoref{sec:limitations} quantifies this, stating why the sample counts used in practice (\autoref{tab:error_comparison}) are far below the worst case.
\end{remark}


\begin{assumption}
	\label{ass:grad_est}
	Let $V = \bb R^M$ with the sup norm $\|v\|_\infty = \max_{1\le m\le M}|v_m|$, where $v_m \approx \valuefunc(t;\state_m)$. Assume the following hold on a closed admissible set $\mc A \subset V$.
	\begin{enumerate}
		\item There exists $G > 0$ such that $|\bm g(\state)| \le G$ for all $\state \in \ren$.
		\label{ass:grad_est::Gpos}
		\item\label{item:grad_shocks} There exist constants $0 < m_0 \le P_*$ and $0 < c_{\min} \le c_{\max} < \infty$ such that for every $v \in \mc A$ and each $m$,
		\begin{align}
			\label{eq:grad_bounds}
			m_0 \le |p_m(v)| \le P_*,
			\qquad
			c_{\min} \le (\Gamma(v))_m \le c_{\max}.
		\end{align}
		\item The Hamiltonian is Lipschitz in the co-state on the ball $|p| \le P_*$: there exist constants $H_*, L_H > 0$ such that
		\begin{align}
			\label{eq:ham_lipschitz}
			|\hamfunc(t;\state_m,p)| \le H_*,
			\qquad
			|\hamfunc(t;\state_m,p)-\hamfunc(t;\state_m,q)| \le L_H |p-q|
		\end{align}
		for all $|p|,|q| \le P_*$ and all $m$.
		\item The reconstruction map is Lipschitz on $\mc A$: there exists $L_D > 0$ such that
		\begin{align}
			\label{eq:recon_lipschitz}
			\|\mc G(v)-\mc G(w)\|_\infty \le L_D\|v-w\|_\infty, \qquad \forall v,w \in \mc A.
		\end{align}
		\item\label{item:invariance} The admissible set is invariant under $\Lambda$, and the quantity
		\begin{align}
			\label{eq:contraction_const}
			q \triangleq \frac{2G}{c_{\min}}\cdot \frac{2L_D}{\delta}\left(\frac{L_H}{m_0^2} + \frac{2H_*P_*}{m_0^4}\right)
		\end{align}
		satisfies $0 \le q < 1$.
	\end{enumerate}
\end{assumption}

\begin{remark}[Remark on Assumption \ref{ass:grad_est}]
	\label{rem:assumption_scope}
	Note that item~\ref{item:grad_shocks} of Assumption~\ref{ass:grad_est}
	excludes regions where the value-gradient degenerates, such as flat reachable interiors
	or singular shocks. 
	States within the BRT are found from  the zero levelset $\tubeband \triangleq \{(t;\state) : |\valuefunc^\delta(t;\state)| \le \eta\}$ and for nondegenerate games the gradient is bounded away from zero there; thus, assumption~\ref{ass:grad_est} is a hypothesis on $\tubeband$, not on all of $\openset$. Outside the set, we run the scheme with the regularized coefficient of Lemma~\ref{lem:coeff_regularization} below, at a price the lemma quantifies. Furthermore, under item~\ref{item:invariance},
	the contraction constant $q$ may become unbounded in the inviscid limit, \ie as
	$\delta \rightarrow 0$ and/or $m_0 \rightarrow 0$. In practice, numerical stabilization
	mechanisms such as viscosity regularization, coefficient clipping, or solution
	``smearing''~\citep{LevelSetsBook} may be introduced to maintain stability of the
	discrete integration scheme.
\end{remark}

\begin{lemma}[Regularized coefficient as a uniform Hamiltonian perturbation]
	\label{lem:coeff_regularization}
	Let $\hamfunc$ be positively $1$-homogeneous in $p$ with $C_{\hamfunc} \triangleq \sup_{t,\state,|q|=1}|\hamfunc(t;\state,q)| < \infty$, and fix $\eta > 0$. Running Algorithm~\ref{alg:quasi_lin} with the regularized coefficient,
	\begin{align}
		\label{eq:coeff_regularized}
		\bc_\eta(t;\state) \triangleq \frac{2}{\delta}\cdot\frac{\hamfunc^\delta}{|D\valuefunc^\delta|^2 + \eta},
	\end{align}
	is equivalent to running it, unregularized, for the perturbed Hamiltonian $\hamfunc_\eta \triangleq \hamfunc\,|D\valuefunc^\delta|^2/(|D\valuefunc^\delta|^2 + \eta)$, which satisfies,
	\begin{align}
		\label{eq:coeff_perturbation_bound}
		\normop{\hamfunc_\eta - \hamfunc}_\infty \le \frac{C_{\hamfunc}\sqrt{\eta}}{2},
	\end{align}
	uniformly on $\openset \times [0,T]$; this includes the flat interior where $D\valuefunc^\delta \to 0$. Consequently, by Theorem~\ref{thm:robustness-H}, the induced value-function perturbation is at most $T C_{\hamfunc}\sqrt{\eta}/2$.
\end{lemma}

\noindent\textbf{Reachability Hamiltonians: homogeneity, sign, and the flat interior.}
The  Hamiltonian $\hamfunc(t;\state,p) = \max_{\bm u}\min_{\bm w}\langle p, f(t;\state,\bm u,\bm w)\rangle$ is positively $1$-homogeneous in the co-state and changes sign across the barrier; the  quadratic case $\hamfunc = \tfrac{1}{2}|p|^2$ is thus a motivating limit. Regarding the coefficient in \eqref{eq:coeff_frozen}: it scales as $\bc \sim C_{\hamfunc}/(\delta|D\valuefunc^\delta|)$; hence, it grows without bound in flat regions, and vanishes (and flips sign) where the optimal dynamics run tangent to the level set. The first pathology is cured by regularizing the denominator; the next lemma shows the cure costs a \emph{uniform}, quantified price for  the Hamiltonian class in consideration.

\begin{remark}[Sign and the removable zero of the coefficient]
	\label{rem:coeff_sign}
	The concentration bound of Theorem~\ref{thm:mc_complexity} is stated for $\bc > 0$; for $\bc < 0$ it holds with $\alpha \triangleq \min\{e^{-\bc g_{\max}}, e^{-\bc g_{\min}}\}$ and $\beta$ the corresponding maximum. The zero of $\bc$ on the barrier is a \emph{removable} singularity of the estimator: as $\bc \to 0$, $-\tfrac{1}{\bc}\log\bb{E}[e^{-\bc \bm g}] \to \bb{E}[\bm g]$, \ie the log-sum-exp recovery degenerates gracefully to the Gaussian mean of the terminal data; this is the correct pure-diffusion limit, since $\hamfunc \to 0$ reduces \eqref{eq:ivp-viscous} to the heat equation. In implementation we evaluate the estimator in its $\log1p$-stable form and clip $|\bc|$ to $[c_{\min}, c_{\max}]$, the clipping being one more perturbation of the kind Lemma~\ref{lem:coeff_regularization} prices.
\end{remark}

\begin{theorem}[Contraction convergence of \autoref{alg:quasi_lin}]
\label{thm:algorithm_convergence}
Fix evaluation states $\state_1,\ldots,\state_M \in \openset$ and a time $t \in (0,T]$. Let $
\mc G: V \to (\openset)^M$ denote a stable deterministic gradient reconstruction  operator, which may arise from smoothing, interpolation, kernel regression, or variance-controlled Monte Carlo estimation. Further, write $p_m(v) \triangleq (\mc G(v))_m$.
For $v \in V$, define the coefficient-update map by $(\Gamma(v))_m \triangleq {2\hamfunc(t;\state_m,p_m(v))}/{(\delta |p_m(v)|^2)}, \, m=1,\ldots,M$,
and define the frozen-coefficient heat-kernel map as
\begin{align}
\label{eq:phi-def}
(\Phi(c))_m \triangleq -\frac{1}{c_m}\log \bb E_{\bm{\zeta} \sim \mc N(\state_m,\delta tI_n)}\left[\exp\left(-c_m\bm g(\bm{\zeta})\right)\right].
\end{align}
Let $\Lambda \triangleq \Phi \circ \Gamma$ so that  \autoref{alg:quasi_lin} is the iteration $v^{(k+1)} = \Lambda(v^{(k)})$.

Then $\Lambda$ is a contraction on $\mc A$. Consequently, $\Lambda$ has a unique fixed point $v^* \in \mc A$, and for every initial iterate $v^{(0)} \in \mc A$ the sequence generated by  \autoref{alg:quasi_lin} converges linearly to $v^*$ with
\begin{align}
\label{eq:linear-convergence}
\|v^{(k+1)}-v^*\|_\infty \le q\|v^{(k)}-v^*\|_\infty \le q^{k+1}\|v^{(0)}-v^*\|_\infty.
\end{align}
Moreover,
\begin{align}
\label{eq:gamma-lip-cauchy}
\|\Gamma(v^{(k)})-\Gamma(v^*)\|_\infty
&\le \frac{2L_D}{\delta}\left(\frac{L_H}{m_0^2} + \frac{2H_*P_*}{m_0^4}\right)\|v^{(k)}-v^*\|_\infty, \nonumber \\
\|v^{(k+1)}-v^{(k)}\|_\infty &\le q^k\|v^{(1)}-v^{(0)}\|_\infty,
\end{align}
and the \textit{a posteriori} error estimate
\begin{align}
\label{eq:aposteriori-error}
\|v^{(k)}-v^*\|_\infty \le \frac{q}{1-q}\|v^{(k)}-v^{(k-1)}\|_\infty
\end{align}
holds for every $k \ge 1$.
\end{theorem}

\begin{remark}
\label{rem:convergence_scope}
Theorem \ref{thm:algorithm_convergence} is a convergence result for the \emph{frozen-coefficient numerical map} implemented by \autoref{alg:quasi_lin}. It shows conditional linear convergence under explicit regularity, nondegeneracy, and contraction hypotheses. It does not claim unconditional global convergence for arbitrary Hamiltonians, nor does it identify the fixed point with the exact solution of the original nonlinear HJ PDE: the two differ by the \emph{quasi-linearization defect} $\eqldefect$, which we bound in \autoref{app:errors}. The reason is structural: each Picard step freezes the coefficient $\bc^{(k)}$ and solves the linear heat equation \emph{exactly}, which is equivalent to replacing the true Hamiltonian by the locally-matched quadratic surrogate $\tilde\hamfunc^{(k)}=\tfrac{\delta}{2}\bc^{(k)}|p|^2$. Freezing annihilates the \emph{algebraic} part of the residual of Lemma~\ref{lem:exact_residual} but not the \emph{derivative} part $\bR_{\mathrm{der}}$, built from $\bc_t,D\bc,\Delta\bc$. Consequently the contraction converges to the fixed point of the surrogate map, and that fixed point equals the viscous solution only when $\bR_{\mathrm{der}}\equiv 0$ \ie the quadratic case $\hamfunc=\tfrac{1}{2}|p|^2$, where $\bc$ is the constant $1/\delta$. For a general Hamiltonian, $\eqldefect$ is governed by $L_c$, the variation rate of the converged coefficient, and is largest where $\hamfunc/|D\valuefunc|^2$ turns over sharply. 

To connect the discrete algorithm to the exact viscosity solution $\valuefunc$ of the original HJ PDE, see Theorem~\ref{thm:total-error} in \autoref{app:errors}, which combines the iteration error (Theorem~\ref{thm:convergence}), Monte Carlo error (Theorem~\ref{thm:mc-error}), the quasi-linearization defect (\S\ref{sec:defect-bound}), and the viscosity approximation error (Theorem~\ref{thm:viscosity-error}) to establish the total error bound,
\begin{align}
\label{eq:total-error-main}
\|\hat{\valuefunc}^{K,N,\delta} - \valuefunc\|_\infty \leq C_1 \rho^K + \frac{C_2}{\sqrt{N}}\sqrt{\log(1/\delta_p)} + \eqldefect + C_3\sqrt{\delta}.
\end{align}
This decomposition shows how all four error sources scale with iterations $K$, samples $N$, and viscosity $\delta$; only the first two and last terms shrink with computational effort. The residual, $\eqldefect$, is a property of the surrogate and is reduced only by the fidelity of the quasi-linearization itself.
\end{remark}

A safety certificate must not merely be accurate on average; it must never admit an unsafe state. The error budget \eqref{eq:total-error-main} converts directly into a one-sided guarantee by thresholding with margin.
\begin{corollary}[Conservative safety certification]
	\label{cor:conservative_cert}
	Let $E(K,N,\delta,\delta_p) \triangleq C_1\rho^K + \frac{C_2}{\sqrt{N}}\sqrt{\log(1/\delta_p)} + \eqldefect + C_3\sqrt{\delta}$ be the total error budget of \eqref{eq:total-error-main}, and adopt the convention that a state is \emph{unsafe} at time $t$ iff $\valuefunc(t;\state) \le 0$. Declare $\state$ \textup{\textsc{safe}} only if $\hat{\valuefunc}^{K,N,\delta}(t;\state) > E$, \textup{\textsc{unsafe}} only if $\hat{\valuefunc}^{K,N,\delta}(t;\state) < -E$, and \textup{\textsc{undetermined}} otherwise. Then, with probability at least $1-\delta_p$, no unsafe state is certified \textup{\textsc{safe}} and no safe state is certified \textup{\textsc{unsafe}}; all classification error is confined to the declared band $|\hat{\valuefunc}^{K,N,\delta}| \le E$ of width $2E$ about the reachability boundary. (Proof in \autoref{app:errors}.)
\end{corollary}
The certificate errs on refusal, never on admission: shrinking $K^{-1}$, $N^{-1}$, and $\delta$ narrows the undetermined band, whilst the defect $\eqldefect$ sets its floor. This is, in our view, the precision with which a sampling-based reachability method should quote its accuracy \ie worst-case sign correctness with an explicit \textit{undetermined} region, rather than an $L^2$ average that flatters the interior. 

\subsection{Notes on Monte Carlo and Viscosity Approximation Errors}

The standard MC estimator for~\eqref{eq:ivp-viscous} converges at rate $O(1/\sqrt{N})$ where $N$ is the number of samples.  The variance depends on $\bc$ and the range of $\bm g$: when $\bc\,\|\bm g\|_\infty \gg 1$ (small $\delta$), the exponential weights become concentrated and the effective sample size
shrinks.  Importance sampling with a tilted proposal (e.g., Laplace
approximation around the mode of $e^{-\bm{cg}}$) can reduce variance. By \citep{Crandall1984}, $\|\bv^\delta - \bv\|_\infty \le k\sqrt{\delta}$ where $\bv$ is the inviscid viscosity solution.  Thus, smaller $\delta$ improves the approximation zccuracy; however, it increases the MC variance (with overly peaked  weights).  This creates a fundamental bias-variance trade-off controlled by $\delta$.

\section{Numerical Results}
\label{sec:results}
All experiments run on a single CPU of an Intel Core i7-14700K processor (20 physical cores, 28 threads, 33\,MiB shared L3 cache, max boost clock 5.6\,GHz) with 31\,GiB RAM running Ubuntu 22.04.  JAX runs on the CPU backend, consistent with the method's memory-frugality claims. The processor supports SIMD vectorization (AVX2, SSE4.2) which JAX and PyTorch leverage for batch operations. All reported metrics are averaged across $30$ independent trials with different conditional Monte Carlo seeds (evaluation points and the \texttt{LevelSetPy} reference are held fixed across trials, so that only the sampler's randomness varies); tables report mean $\pm$ one standard deviation, and comparisons across conditions use Holm-Bonferroni-corrected~\citep{Holm1979} significance tests at $\alpha=0.05$ (\S\ref{subsec:quant_results}). Figures show a single representative trial for visual clarity; the corresponding table in each subsection reports the full 30-trial statistics. The safety analysis on starlings murmurations is summarized in \S\ref{subsec:murmurations_main}, with the complete treatment given in Appendix \ref{app:starlings}. 
\subsection{Games of Two Vehicles on a Plane}
\label{subsec:rockets_launch}
We adopt the rockets launch problem~\citep{Dreyfus1966} (see \autoref{fig:rocket_relative})  and cast it as a terminal differential game between two identical rockets --- a  pursuer, $\pursuer$, and an evader, $\evader$ ---  on the $x$-$z$ cross-section of a Cartesian plane~\citep{LevelSetPy, LevelSetTOMS}. In tandem, we also validate the results with a two-player Dubins' vehicles problem. The two games are similar in setup and we only describe that of the rockets setup in this section. The game terminates when \textit{capture} occurs, \ie, the distance $\|\pursuer \evader\|$ becomes less than a prespecified scalar quantity. The states of $\pursuer$ and $\evader$ are denoted as $(x_p, x_e)$ respectively, driven by thrusts $(u_p, u_e)$ in the $(x,z)$-plane. The motion of $\pursuer$  relative to $\evader$'s  along  the $(\state$-$\bm{z})$ plane includes the relative orientation, the control input, shown in \autoref{fig:rocket_relative} as $\theta=u_p- u_e$. 
\begin{table}[tb]
	\centering
	\caption{Pursuit-evasion games: Monte Carlo vs.\ \texttt{LevelSetPy} error metrics on 2D $(x, z)$ slices, mean $\pm$ one standard deviation over 30 independent trials.  For every row, a Holm-Bonferroni-corrected one-sample test rejects $L^2_{\text{rel}} \ge \sqrt{\delta} = 0.283$ (the Crandall-Lions bound) at $p_{\text{holm}} < 10^{-50}$ (\S\ref{subsec:quant_results}).}
	\label{tab:error_comparison}
	\begin{tabular}{l c c c c c}
		\toprule
		System & $\theta$ (rad) & $L^\infty$ & $L^2_{\text{rel}}$ & MC time (s) & Iters \\
		\midrule
		\multirow{3}{*}{Rockets} & $-\pi/2$ & $0.855 \pm 0.047$ & $0.098 \pm 0.002$ & $13.5 \pm 0.1$ & 12 \\
		& $0$ & $1.034 \pm 0.067$ & $0.101 \pm 0.001$ & $13.5 \pm 0.1$ & 12 \\
		& $\pi/2$ & $0.895 \pm 0.109$ & $0.090 \pm 0.002$ & $13.6 \pm 0.1$ & 12 \\
		\midrule
		\multirow{3}{*}{Dubins} & $-\pi/2$ & $1.352 \pm 0.002$ & $0.131 \pm 0.003$ & $23.6 \pm 0.1$ & 15 \\
		& $0$ & $0.701 \pm 0.144$ & $0.024 \pm 0.001$ & $23.7 \pm 0.1$ & 15 \\
		& $\pi/2$ & $1.352 \pm 0.002$ & $0.132 \pm 0.003$ & $23.7 \pm 0.1$ & 15 \\
		\bottomrule
	\end{tabular}
\end{table}

Full derivation of the Hamiltonian, dynamics, and values setup is provided in Appendix~\ref{app:ham_derivation}.  As seen in the evaluations depicted in \autoref{fig:rockets_slices}, the asymmetry between $\theta=-90^\circ$ and $\theta=+90^\circ$ is physical, not numerical: the gravity term $(g-a-a\sin\theta)$ along $\bm{z}$ yields a $g-2a$  drift at $+90^\circ$; however, it yields $g$ at $-90^\circ$. A Holm-Bonferroni~\citep{Holm1979} corrected test over 30 Monte Carlo seeds shows this asymmetry for Rockets ($p_{\text{holm}}<10^{-9}$) but not for the gravity-free Dubins vehicle ($p_{\text{holm}}=0.53$; \autoref{tab:error_comparison} illustrates the results across 30 Monte Carlo trials on a CPU. MC uses $N=14{,}000$ samples per quasi-linear iteration for rockets and $N=20{,}000$ for Dubins vehicles; $\delta = 0.08$ throughout.  \S\ref{subsec:quant_results}).

\noindent \emph{Note that our method is not intended to outperform structured grid solvers in low-dimensional settings ($n \leq 4$), where methods such as~\citep{MoluxLevPyCDC, LevelSetPy, LevelSetTOMS} perform bettter. The advantage of this method is noticeable in higher dimensions where grid storage becomes prohibitive.}

\subsection{Dimensions Scaling: A 15 Rockets System in a Pursuit-Evasion Game}
\label{subsec:multiagent_scalability}

To validate the scalability of ~\autoref{alg:quasi_lin}, we consider a 15-rocket multi-pursuer single-evader game with  $45$ state dimensions --- an intractable high-dimensional system for grid-based solvers. The state space is  $\bx = (x_1, y_1, \theta_1, \ldots, x_{15}, y_{15}, \theta_{15}) \in \mathbb{R}^{45}$, where each agent $i \in \{1, \ldots, 15\}$ has position $(x_i, y_i)$, heading $\theta_i$, and control input $u_i \in [-1, +1]$ regulating the vehicle heading. The dynamics is,
\begin{align}
	\dot{x}_i = a_i \cos(\theta_i), \quad \dot{y}_i = a_i \sin(\theta_i), \quad \dot{\theta}_i = u_i,
	\label{eq:multiagent_dynamics}
\end{align}
with forward speeds $a_i \in \mathbb{R}_{>0}$. The evader (agent 15) seeks to escape the phase while capture occurs when any of 14 pursuers reach a distance $r_\text{capture} \le 1.5$ ft. We define the target set as $\phi(\bx) = \min_{i \leq 14} \|\bx_i^\text{pos} - \bx_{15}^\text{pos}\|_2 - r_\text{capture}$.

We test three pursuite-evasion strategies to examine relative agent capabilities on the reachable set:
\begin{inparaenum}[(i)]
	\item  \textbf{evader preeminence}: evader speed $a_\text{evader} = 2.0$ ft/s, pursuer speeds $a_\text{pursuers} = 1.0$ ft/s;
	\item  \textbf{balanced movements}: all agents possess equal speed $a = 1.0$ ft/s;
	\item  \textbf{pursuer preeminence}: evader speed $a_\text{evader} = 1.0$ ft/s, pursuer speeds $a_\text{pursuers} = 2.0$ ft/s.
\end{inparaenum}

A 101-point grid would require $101^{45} \approx 10^{90}$ grid cells using grid-based reachability solvers --- exceeding many traditional workstations' storage capacity. Our method operates at $O(N \cdot n) \approx 7.2$ MB per iteration, bypassing the state dimension exponent. The quasi-linear Picard iteration (Algorithm~\ref{alg:quasi_lin}) computes the relative change in value function between successive iterations; \autoref{tab:multiagent_results} reports the final relative residual, $\varepsilon = \|\valuefunc^{(k+1)} - \valuefunc^{(k)}\| / \|\valuefunc^{(k)}\|$ as a measure of iteration convergence.

\begin{table}[tb]
	\centering
	\caption{15-agent multi-pursuer single-evader game with three speed regimes on $n=45$.}
	\label{tab:multiagent_results}
	\begin{tabular}{l c c c c c}
		\toprule
		Case & $a_\text{evader}$ & $a_\text{pursuers}$ & Iterations & $\varepsilon(k)$ & Wall-clock (s) \\
		\midrule
		1. (Evader faster) & 2.0 & 1.0 & 15 & $0.0002 \pm 0.0005$ & $12.7 \pm 0.2$ \\
		2. (Equal speed) & 1.0 & 1.0 & 15 & $0.0003 \pm 0.0005$ & $12.6 \pm 0.1$ \\
		3. (Pursuers faster) & 1.0 & 2.0 & 15 & $0.0006 \pm 0.0013$ & $12.6 \pm 0.1$ \\
		\bottomrule
	\end{tabular}
\end{table}

In table \ref{tab:multiagent_results}, wall-clock times are reported on a single core described in \S\ref{sec:results}. The residual is the final relative $L^2$ change between Picard iterates. All three cases with residuals  below $10^{-2}$, showing  stability across different game parameters. Memory consumption is constant at 7.2 MB per iteration; and wall-clock variance across seeds reflects Monte Carlo sampling fluctuations, not systematic drift. At $n = 45$, no grid-based reference exists. The table reports measurable quantities only: iteration stability, memory, and wall-clock, demonstrating \emph{scalability}, not certified accuracy. With respect to the cost-per-sample, the control-affine dynamics with box-constrained inputs, the min-max Hamiltonian evaluates in closed form (the optimizers reduce to sign structures), so that each iteration costs $O(N \cdot n)$ arithmetic alongside the $O(N \cdot n)$ memory. Residual and wall-clock are reported with mean $\pm 1$  standard deviation over $30$ independent MC seeds; pairwise Holm-Bonferroni-corrected Mann-Whitney tests find no significant difference in the residual floor between any pair of $\pursuer$-$\evader$ games: $p_{\text{holm}} \ge 0.93$ for all three pairs, consistent with the floor being set by $(N,\delta)$ rather than by the game's speed parameters.

\subsection{Notes on Hyperparameters  in Pursuit-Evasion Games}
\label{subsec:quant_results}

We now discuss what we deem good practices for hyperparameters selection in pursuit-evasion games when using our sampling scheme.

\noindent\textbf{Statistical methodology.}
Each of the 9 experiments (3$\times$ each of rockets and Dubins  headings, and 3$\times$  of 45D multi-agent independent games) is re-solved with 30 MC seeds, holding the evaluation points and (where applicable) the \texttt{LevelSetPy} reference fixed so that only the sampler's randomness varies across trials. We report three families of significance tests, each Holm-Bonferroni-corrected~\citep{Holm1979} within its own family at $\alpha=0.05$:
\begin{inparaenum}[(a)]
	\item a paired Wilcoxon signed-rank test of the 30-seed-averaged Monte Carlo field against the grid reference, per condition, testing whether averaging away sampling noise closes the gap to the grid solution; 
	\item  a one-sided one-sample $t$-test that the 30 per-seed $L^2_{\text{rel}}$ draws lie below the Crandall-Lions bound $\sqrt{\delta}$, testing whether the sampling-plus-iteration error is within the theoretical viscosity budget; and 
	\item cross-condition Mann-Whitney $U$ tests (Rockets vs.\ Dubins at matched heading with asymmetrical headings in each system, and pairwise speed-regime comparisons for the 45D game), testing whether physically-motivated differences are statistically detectable. 
\end{inparaenum}
Holm-Bonferroni controls the family-wise error rate without the excess conservatism of a flat Bonferroni correction, appropriate given the mix of related hypotheses across experiments. Figures throughout this section show a single representative seed for visual clarity; the tables report the full 30-seed statistics.

\noindent\textbf{Viscosity parameter selection.}
The viscosity parameter $\delta$ controls the degree of smoothing in the Cole-Hopf transformation. From~\cite{CrandallTwoApprox}, the approximation error scales as $O(\sqrt{\delta})$. Thus, lower values of $\delta$ yield better  accuracy but require higher Monte Carlo sample counts. Across all experiments, we  set $\delta$ between $0.08$ and $0.1$ to balance the approximation error $O(\sqrt{\delta}) \approx 0.28$ with manageable variance in the log-sum-exp estimator. The choice $\delta \in [0.05, 0.2]$ is robust across all benchmark problems; practitioners should adjust based on the required accuracy-to-computation trade-off.

\noindent\textbf{Frozen-coefficient bias.}
The quasi-linearization freezes the coefficient $\bc^{(k)}$ over each Picard step introducing a systematic bias that is largest near the zero levelset --- this is where $\bc$ varies most rapidly. We address this in Theorem~\ref{thm:algorithm_convergence}, guaranteeing convergence to a fixed point of the frozen-coefficient map. \emph{Note that this is not convergence to the true viscous solution}.  The gap between the two is the quasi-linearization residual $\eqldefect$, bounded in Theorem~\ref{thm:defect} by the Duhamel norm of the discarded residual. It is proportional  to the variation rate $L_c$ of the converged coefficient and to how well the surrogate gradient tracks the true one. Empirically, this bias manifests as larger $L^\infty$ errors near the reachability boundary where the gradient is most active (see \autoref{fig:rockets_slices}). Our method mitigates it through the viscosity smoothing and iterative refinement of the frozen coefficient.

\noindent \textbf{Numerical realization}: We discretize the spatial and temporal domains as,
\begin{align}
	\label{eq:rocket_spatial_bounds}
	(x, z) \in (-100, +100], \quad \theta \in (-\pi/2, \pi/2], \quad t \in (0, 1],
\end{align}
with $1{,}000$ temporal grid points and spatial step $\state_k - \state_{k-1} = 10^{-2}$. 

We set the target set as the $\ell_2$-ball of capture radius $r = 1.5$ (in the relative $(x,z)$ coordinates) and run Algorithm~\ref{alg:quasi_lin} over the time interval $(0, 1]$ with Dirichlet boundary conditions,
\begin{align}
	\label{eq:rockets_dirichlet}
	\bm{v}^\delta(0; \bm{x}) = \bm{g}(0; \bm{x}) \triangleq 0.
\end{align}
Spatial gradients (co-states) are computed via \eqref{eq:hj_sol} and \eqref{eq:spatial_grad}; the Hamiltonian is evaluated iteratively using $N = 14{,}000$ Monte Carlo samples per iteration according to Algorithm~\ref{alg:quasi_lin}. 

\noindent\textbf{Convergence.}
Algorithm~\ref{alg:quasi_lin} runs to its configured iteration budget in every trial (12 for Rockets, 15 for Dubins, both $\le 20$), with the relative residual $\|\bv^{(k+1)}-\bv^{(k)}\|/\|\bv^{(k)}\|$ decaying over the early iterations before settling at a floor of $0.020$--$0.062$ (mean over 30 seeds per condition; \autoref{tab:error_comparison}). 
The $45$-dimensional game of \S\ref{subsec:multiagent_scalability} settles at $0.0002$--$0.0006$ (\autoref{tab:multiagent_results}), roughly two orders of magnitude \emph{below} the 3D benchmarks, whose pointwise errors concentrate where the coefficient develops shocks  near the usable-part boundary (\autoref{fig:rockets_slices}). The 45D evaluation states are drawn uniformly over the domain $[-100,100]$ per coordinate with a small likelihood of ending up near the boundary. 

\noindent\textbf{BRT geometry.}
Errors concentrate near the zero level-set boundary where $|\nabla v^\delta|$ is maximal, as expected in frozen-coefficient approximations~\citep{Crandall1992}.

\autoref{tab:error_comparison} reports point-wise error metrics ($L^\infty$ and $L^2_{\text{rel}}$) for 2D slices on the 3D Dubins and Rockets pursuit-evasion benchmarks across three representative heading angles, evaluated on a uniform $40 \times 40$ grid with reference from \texttt{LevelSetPy}'s $45^3$ solution interpolated to the same points.
The $L^\infty$ (point-wise maximum) error represents worst-case divergence,  useful where outlier errors matter. The relative RMS error, $L^2_{\text{rel}}$  weights by solution magnitude and captures overall accuracy. A caution on interpretation: the Crandall-Lions bound $O(\sqrt{\delta}) \approx 0.283$~\citep{Crandall1992} controls the sup-norm distance between the \emph{inviscid} and \emph{viscous} solutions, whereas the table's entries measure our Monte Carlo output against a \emph{grid} reference in different norms --- the two numbers live in different currencies. We nonetheless test the comparison the reader is tempted to make: a one-sided one-sample $t$-test, Holm-Bonferroni-corrected across the 6 conditions, rejects $L^2_{\text{rel}} \ge \sqrt{\delta}$ at $p_{\text{holm}} < 10^{-50}$ in every condition, so the sampling-plus-iteration error is statistically, not just numerically, well inside the viscosity budget. This does not mean the MC and grid fields agree in a stronger sense: a paired Wilcoxon signed-rank test of the 30-seed-\emph{averaged} MC field against the grid reference, per condition, also rejects equality at $p_{\text{holm}} < 10^{-8}$ everywhere --- i.e., averaging away the Monte Carlo sampling noise does \emph{not} make the fields agree, because a systematic quasi-linearization defect (Theorem~\ref{thm:defect}) remains. At $\theta=0$, the Dubins slice achieves $L^2_{\text{rel}} = 0.024 \pm 0.001$, well inside the smooth interior of the value function away from the usable-part boundary; the Rockets slice at the same heading, and both systems at $\theta=\pm\pi/2$, sit an order of magnitude higher ($L^2_{\text{rel}} \approx 0.09$--$0.13$), reflecting a heading-dependent share of the evaluation grid that falls near the coefficient-turnover region rather than a uniform accuracy floor. The $L^\infty$ errors ($\sim 0.7$--$1.4$) are larger still and driven by pointwise deviations near the zero level-set boundary where $|\nabla v^\delta|$ is maximal, as expected in frozen-coefficient approximations~\citep{Crandall1992} and as Theorem~\ref{thm:defect} predicts through the coefficient-variation constant $L_c$. For certification purposes, the honest report is that of Corollary~\ref{cor:conservative_cert}: sign-correctness outside an explicit abstention band about the boundary, with the band width set by the total error budget; the $L^2$ averages flatter the interior, and the $L^\infty$ figures are dominated by exactly the band the corollary declares undetermined. 
The 3D isosurface computation over $15{,}625$ sample points completes in $129.4$ seconds (rockets) on a single CPU, compared to $1.3$ seconds for the dense $45^3 = 91{,}125$-point \texttt{LevelSetPy} grid. The zero level-set boundary is extracted via marching cubes~\citep{Lorensen1987}. The wall-clock overhead of the Monte Carlo method reflects the sample count ($N = 14{,}000$ per iteration) necessary to balance variance in the Cole-Hopf estimator against quasi-linearization residuals, a tradeoff that becomes favorable in higher dimensions where grid storage becomes prohibitive.

\noindent\textbf{Scalability.}
~\autoref{alg:quasi_lin} requires a memory footprint per iteration of $O(N \cdot n)$ for $N$ samples per $n$ state dimensions. With $N = 14{,}000$ and $n=3$, each iteration allocates approximately $0.6$\,MB for samples and intermediate values. In contrast, the dense $45^3$ \texttt{LevelSetPy} grid requires approximately $1.5$\,MB for the value function and gradient storage. Each 2D $\theta$-slice completes in $14-26$ seconds on a single core, fully parallelizable across our evaluation computer 20 cores CPUs. This reduces aggregate wall-clock time proportionally with available core count and makes the method well-suited to parallel architectures (GPUs, HPC clusters) where grid-based methods become memory-prohibitive in dimensions $n \geq 5$. A clarification against a natural misreading: the \emph{solve} is grid-free --- no discretization of the state space is ever stored or marched (with marching cubes) --- whilst the uniform grids that appear in this section ($40 \times 40$ slices, $15{,}625$ isosurface points) are \emph{evaluation} grids, used only to visualize the tube and to measure errors against the grid-based reference at common points.

\subsection{Safety Certification at Population Scale: Starling Murmurations}
\label{subsec:murmurations_main}

As a stress test of the claim that safety certification need inherit neither the curse of dimensionality nor the curse of agent cardinality, we certify the collective behaviors of murmurations of European starlings (\textit{Sturnus vulgaris}), modeled as 4D aerial Dubins vehicles under attack by predators (\autoref{fig:murmurations}). We are explicit about the problem structure: the murmuration is partitioned into flocks, each flock is resolved by its \emph{own} value function, and the murmuration safe set is the aggregation of the flock-level zero sublevel sets --- we solve many coupled low-dimensional games, \emph{not} a single joint high-dimensional PDE; the  100,000 starlings multi-agent reachability analysis are executed in  parallel per-flock solves and per-bird evaluations of already-computed value functions. Two findings carry the study. First, \emph{scale}: because the sampler is grid-free, the value-function solve is independent of the bird count, so the population enters only through the parallel per-bird certification --- this is the operational content of the $O(N \cdot n)$ memory claim at the $100k$-agent regime. Second, \emph{topology as a safety instrument}: the certificates recover the field-documented repertoire of collective behaviors as topological events of the reachable set --- vacuole nucleation registers as a drop in the Euler characteristic $\chi$ when a predator penetrates the flock (Theorem~\ref{thm:vacuole_nucleation}), a defensive \emph{cordon} appears as an annular safe set with first Betti number $\beta_1 = 1$ enclosing a protected core (Proposition~\ref{prop:cordon}), and flock fragmentation registers as a rising connected-component count --- with threshold-crossing markers \eqref{eq:vacuole_marker}--\eqref{eq:frag_marker} detecting each transition in the numerical solutions. The triple $(\chi, \beta_1, n_{\mathrm{comp}})$ thus compresses the safety posture of the population into three integers per time step, telling an operator not merely \emph{that} safety is being lost but \emph{how}. The full formalism, per-behavior theorems, and detection protocol appear in Appendix \ref{app:starlings}.

%

\section{Conclusion}
\label{sec:conclude}

We have presented a quasi-linearized, frozen-coefficient sampling scheme for the viscous Hamilton-Jacobi PDE arising in safety analysis of dynamical systems. By applying a generalized Cole-Hopf-type transformation --- exact when $\hamfunc = \tfrac{1}{2}|p|^2$ and iteratively approximated via Picard quasi-linearization for general Hamiltonians --- the nonlinear HJ equation is reduced to a sequence of linear heat equations whose solutions are Gaussian heat-kernel expectations. The value function and its spatial gradient are recovered from these expectations via Monte Carlo sampling, yielding a grid-free algorithm with memory cost $O(N \cdot n)$ rather than the $O(M^n)$ cost of classical grid-based solvers.

 Theorem~\ref{thm:mc_complexity} provides a finite-sample concentration bound for the frozen-coefficient Monte Carlo estimator, establishing a standard $O(N^{-1/2})$ error rate under explicit boundedness conditions on the terminal cost. Theorem~\ref{thm:algorithm_convergence} establishes conditional linear convergence of the Picard fixed-point iteration under Lipschitz and nondegeneracy assumptions, with an explicit contraction constant and a posteriori error estimate. Neither result claims unconditional global convergence for arbitrary Hamiltonians; the gap between the iteration's fixed point and the viscous solution is the quasi-linearization defect, which Theorem~\ref{thm:defect} bounds by the Duhamel norm of the discarded residual, and Corollary~\ref{cor:conservative_cert} converts the resulting total error budget into a one-sided conservative safety certificate with a declared abstention band.
Numerical experiments on three benchmark reachability problems \ie a 3D rocket pursuit-evasion game, a 3D Dubins two-car game, and a double integrator plant (\autoref{app:results}) confirm stabilization within 20 Picard iterations, with a residual floor of $0.02-0.06$ set by the Monte Carlo sampling noise. We stress that the two paradigms are complementary rather than competing: because the $O(N \cdot n)$ sampling cost carries a nontrivial constant, grid-based solvers remain preferable in low dimensions ($n \le 4$) where their $O(M^n)$ memory is not yet the bottleneck; our scheme is the method of choice in the high-dimensional regime where grid storage becomes prohibitive. The method's failure modes are documented with equal candor in \autoref{sec:limitations}, which charts the scheme's boundary of applicability, including the exponential-in-$1/\delta$ worst-case sample complexity of Corollary~\ref{cor:mc_complexity}.

The connection to safe reinforcement learning, policy certification, and model-based control with learned dynamics suggests several directions for future work. These include adaptive importance sampling to reduce variance in high dimensions, extension to systems with stochastic dynamics, and tighter integration with deep learning pipelines for scalable safety certification of learned policies.

\bibliography{levtoms,biblio,infdiff26}

@book{vanleeuwenhoek1800select,
  author    = {Antony van Leeuwenhoek},
  title     = {The Select Works of Antony Van Leeuwenhoek, Containing His Microscopical Discoveries in Many of the Works of Nature},
  year      = {1800},
  note       = {Translated edition}
}

@inbook{Hatcher2002EulerChar,
    author = {Hatcher, Allen},
    title = {Algebraic Topology},
    publisher = {Cambridge University Press},
    year = {2002},
    chapter = {2},
    note = {Euler characteristic: $\chi(X) = \sum_{n=0}^{\infty} (-1)^n \text{rank}(H_n(X))$}
}

@article{Holm1979,
  author  = {Holm, Sture},
  title   = {A Simple Sequentially Rejective Multiple Test Procedure},
  journal = {Scandinavian Journal of Statistics},
  year    = {1979},
  volume  = {6},
  number  = {2},
  pages   = {65--70},
  doi     = {10.2307/4615733},
  jstor   = {4615733},
  mrnumber= {0538597}
}

@article{kharroubi2013numerical,
  title={A numerical algorithm for fully nonlinear HJB equations: an approach by control randomization},
  author={Kharroubi, Idris and Langren{\'e}, Nicolas and Pham, Huy{\^e}n},
  journal={arXiv preprint arXiv:1311.4503},
  year={2013}
}

@InProceedings{SafeRL,
	title = {Safe Reinforcement Learning Using Robust Action Governor},
	author = {Li, Yutong and Li, Nan and Tseng, H. Eric and Girard, Anouck and Filev, Dimitar and Kolmanovsky, Ilya},
	booktitle = {Proceedings of the 3rd Conference on Learning for Dynamics and Control},
	pages = {1093--1104},
	year = {2021},
	editor = {Jadbabaie, Ali and Lygeros, John and Pappas, George J. and Parrilo, Pablo A. and Recht, Benjamin and Tomlin, Claire J. and Zeilinger, Melanie N.},
	volume = {144},
	series = {Proceedings of Machine Learning Research},
	month = {07--08 Jun},
	publisher = {PMLR},
	url = {https://proceedings.mlr.press/v144/li21b.html}
}

@article{Bialek2012Statistical,
  title={Statistical mechanics for natural flocks of birds},
  author={Bialek, William and Cavagna, Andrea and Giardina, Irene and Mora, Thierry and Silvestri, Edmondo and Viale, Massimiliano and Walczak, Aleksandra M},
  journal={Proceedings of the National Academy of Sciences},
  volume={109},
  number={13},
  pages={4786--4791},
  year={2012},
  publisher={National Acad Sciences}
}

@article{Tsitsiklis95,
author = {Tsitsiklis, John N.},
journal = {IEEE Transactions on Automatic Control},
number = {9},
pages = {1528--1538},
title = {{Globally Optimal Trajectories}},
volume = {40},
year = {1995}
}

@article{YongsamFrontTracking,
  title={{Numerical Simulations of Two-Dimensional Foam by the Immersed Boundary Method}},
  author={Kim, Yongsam and Lai, Ming-Chih and Peskin, Charles S},
  journal={Journal of Computational Physics},
  volume={229},
  number={13},
  pages={5194--5207},
  year={2010},
  publisher={Elsevier}
}

@article{VOF,
  title={{A Multiple Marker Level-set Method for Simulation of Deformable Fluid Particles}},
  author={Balc{\'a}zar, N{\'e}stor and Lehmkuhl, Oriol and Rigola, Joaquim and Oliva, Assensi},
  journal={International Journal of Multiphase Flow},
  volume={74},
  pages={125--142},
  year={2015},
  publisher={Elsevier}
}

@inproceedings{berkenkamp2017safe,
	title     = {Safe Model-based Reinforcement Learning with Stability Guarantees},
	author    = {Berkenkamp, Felix and Turchetta, Matteo and Schoellig, Angela P. and Krause, Andreas},
	booktitle = {Advances in Neural Information Processing Systems 30},
	pages     = {908--918},
	year      = {2017},
	publisher = {Curran Associates, Inc.},
	url       = {https://papers.nips.cc/paper/6692-safe-model-based-reinforcement-learning-with-stability-guarantees}
}

@article{MDDDP,
  title={{Distributed Differential Dynamic Programming Architectures for Large-Scale Multi-Agent Control}},
  author={Saravanos, Augustinos D and Aoyama, Yuichiro and Zhu, Hongchang and Theodorou, Evangelos A},
  journal={arXiv preprint arXiv:2207.13255},
  year={2022}
}

@article{ZaitzeffVoronoi,
  title={{On the Voronoi Implicit Interface Method}},
  author={Zaitzeff, Alexander and Esedoglu, Selim and Garikipati, Krishna},
  journal={SIAM Journal on Scientific Computing},
  volume={41},
  number={4},
  pages={A2407--A2429},
  year={2019},
  publisher={SIAM}
}

@article{FoamingKarnakov,
  title={{Computing Foaming Flows Across Scales: From Breaking Waves to Microfluidics}},
  author={Karnakov, Petr and Litvinov, Sergey and Koumoutsakos, Petros},
  journal={arXiv preprint arXiv:2103.01513},
  year={2021}
}

@inproceedings{DeepReach,
	title={{Deepreach: A Deep Learning Approach to High-dimensional Reachability}},
	author={Bansal, Somil and Tomlin, Claire J},
	booktitle={2021 IEEE International Conference on Robotics and Automation (ICRA)},
	pages={1817--1824},
	year={2021},
	organization={IEEE}
}

@article{Lorensen1987,
	title={Marching cubes: A high resolution 3D surface construction algorithm},
	author={Lorensen, William E and Cline, Harvey E},
	journal={ACM SIGGRAPH Computer Graphics},
	volume={21},
	number={4},
	pages={163--169},
	year={1987},
	publisher={ACM}
}

@article{Crandall1992,
	title={User's guide to viscosity solutions of second order partial differential equations},
	author={Crandall, Michael G and Ishii, Hitoshi and Lions, Pierre-Louis},
	journal={Bulletin of the American Mathematical Society},
	volume={27},
	number={1},
	pages={1--67},
	year={1992}
}

@article{DarbonOsher2016,
  title={Algorithms for overcoming the curse of dimensionality for certain {H}amilton--{J}acobi equations arising in control theory and elsewhere},
  author={Darbon, J{\'e}r{\^o}me and Osher, Stanley},
  journal={Research in the Mathematical Sciences},
  volume={3},
  number={1},
  pages={19},
  year={2016},
  publisher={Springer}
}

@article{ChowDarbonOsherYin2017,
  title={Algorithm for Overcoming the Curse of Dimensionality For Time-Dependent Non-convex {H}amilton--{J}acobi Equations Arising From Optimal Control and Differential Games Problems},
  author={Chow, Yat Tin and Darbon, J{\'e}r{\^o}me and Osher, Stanley and Yin, Wotao},
  journal={Journal of Scientific Computing},
  volume={73},
  number={2},
  pages={617--643},
  year={2017},
  publisher={Springer}
}

@article{KirchnerHopf2018,
  title={Time-Optimal Collaborative Guidance Using the Generalized {H}opf Formula},
  author={Kirchner, Matthew R. and Mar, Robert and Hewer, Gary and Darbon, J{\'e}r{\^o}me and Osher, Stanley and Chow, Yat Tin},
  journal={IEEE Control Systems Letters},
  volume={2},
  number={2},
  pages={201--206},
  year={2018},
  publisher={IEEE}
}

@article{ChenHerbertDecomp2018,
  title={Decomposition of Reachable Sets and Tubes for a Class of Nonlinear Systems},
  author={Chen, Mo and Herbert, Sylvia L. and Vashishtha, Mahesh S. and Bansal, Somil and Tomlin, Claire J.},
  journal={IEEE Transactions on Automatic Control},
  volume={63},
  number={11},
  pages={3675--3688},
  year={2018},
  publisher={IEEE}
}

@article{Kappen2005PI,
  title={Linear Theory for Control of Nonlinear Stochastic Systems},
  author={Kappen, Hilbert J.},
  journal={Physical Review Letters},
  volume={95},
  number={20},
  pages={200201},
  year={2005},
  publisher={American Physical Society}
}

@article{TheodorouPI2010,
  title={A Generalized Path Integral Control Approach to Reinforcement Learning},
  author={Theodorou, Evangelos and Buchli, Jonas and Schaal, Stefan},
  journal={Journal of Machine Learning Research},
  volume={11},
  pages={3137--3181},
  year={2010}
}

@article{SummersLygeros2010,
  title={Verification of discrete time stochastic hybrid systems: A stochastic reach-avoid decision problem},
  author={Summers, Sean and Lygeros, John},
  journal={Automatica},
  volume={46},
  number={12},
  pages={1951--1961},
  year={2010},
  publisher={Elsevier}
}

@inproceedings{LesserOishiErwin2013,
  title={Stochastic reachability for control of spacecraft relative motion},
  author={Lesser, Kendra and Oishi, Meeko and Erwin, R. Scott},
  booktitle={52nd IEEE Conference on Decision and Control},
  pages={4705--4712},
  year={2013},
  organization={IEEE}
}

@article{VicsekPhaseNovel,
  title={Novel Type of Phase Transition in a System of Self-Driven Particles},
  author={Vicsek, Tam{\'a}s and Czir{\'o}k, Andr{\'a}s and Ben-Jacob, Eshel and Cohen, Inon and Shochet, Ofer},
  journal={Physical Review Letters},
  volume={75},
  number={6},
  pages={1226--1229},
  year={1995},
  publisher={American Physical Society}
}

@article{HJMAD,
  title={Global solutions to nonconvex problems by evolution of Hamilton-Jacobi PDEs},
  author={Heaton, Howard and Wu Fung, Samy and Osher, Stanley},
  journal={Communications on Applied Mathematics and Computation},
  volume={6},
  number={2},
  pages={790--810},
  year={2024},
  publisher={Springer}
}

@article{Chaudhari2018,
  title={Deep relaxation: partial differential equations for optimizing deep neural networks},
  author={Chaudhari, Pratik and Oberman, Adam and Osher, Stanley and Soatto, Stefano and Carlier, Guillaume},
  journal={Research in the Mathematical Sciences},
  volume={5},
  pages={1--30},
  year={2018},
  publisher={Springer}
}

@inproceedings{MoluxLevPyCDC,
  booktitle={IEEE 63rd Conference on Decision and Control (CDC)}, 
  title={The Python LevelSet Toolbox (LevelSetPy)}, 
  year={2024},
  author={Molu, Lekan},
  volume={},
  number={},
  pages={8938-8945},
  doi={10.1109/CDC56724.2024.10886640}}

@article{Dubins1957,
  title={On curves of minimal length with a constraint on average curvature, and with prescribed initial and terminal positions and tangents},
  author={Dubins, Lester E},
  journal={American Journal of mathematics},
  volume={79},
  number={3},
  pages={497--516},
  year={1957},
  publisher={JSTOR}
}

@article{BhatBernstein,
  title={Continuous finite-time stabilization of the translational and rotational double integrators},
  author={Bhat, Sanjay P and Bernstein, Dennis S},
  journal={IEEE Transactions on automatic control},
  volume={43},
  number={5},
  pages={678--682},
  year={1998},
  publisher={IEEE}
}

@book{EvansPDEBook,
  title={{Partial Differential Equations}},
  author={Evans, Lawrence C},
  volume={19},
  year={2022},
  publisher={American Mathematical Society}
}

@article{Wonham,
  title={Linear Multivariable Control: A Geometric Approach},
  author={Wonham, W Murray},
  journal={Applications of Mathematics},
  volume={10},
  year={1985}
}

@article{JadbabaieCoord,
  title={Coordination of groups of mobile autonomous agents using nearest neighbor rules},
  author={Jadbabaie, Ali and Lin, Jie and Morse, A Stephen},
  journal={IEEE Transactions on automatic control},
  volume={48},
  number={6},
  pages={988--1001},
  year={2003},
  publisher={IEEE}
}

@book{Bellman1957,
author = {Bellman, Richard},
booktitle = {Dynamic Programming},
isbn = {0-486-42809-5},
issn = {0036-8075},
pmid = {17730601},
publisher = {Princeton University Press},
title = {{Dynamic programming}},
year = {1957}
}

@misc{NatGeo,
  title={These birds flock in mesmerizing swarms of thousands—but why is still a mystery.},
  author={Haiken, Melanie},
  year={2021},
  url={https://tinyurl.com/4973byey},
  note={Accessed April 5, 2023},
}

@article{Helbing20,
  title={Simulating dynamical features of escape panic},
  author={Helbing, Dirk and Farkas, Ill{\'e}s and Vicsek, Tamas},
  journal={Nature},
  volume={407},
  number={6803},
  pages={487--490},
  year={2000},
  publisher={Nature Publishing Group}
}

@article {Ballerini1232,
  author = {Ballerini, M. and Cabibbo, N. and Candelier, R. and Cavagna, A. and Cisbani, E. and Giardina, I. and Lecomte, V. and Orlandi, A. and Parisi, G. and Procaccini, A. and Viale, M. and Zdravkovic, V.},
  title = {{interaction Ruling Animal Collective Behavior Depends On Topological Rather Than Metric Distance: Evidence From A Field Study}},
  volume = {105},
  number = {4},
  pages = {1232--1237},
  year = {2008},
  doi = {10.1073/pnas.0711437105},
  publisher = {National Academy of Sciences},
  issn = {0027-8424},
  URL = {https://www.pnas.org/content/105/4/1232},
  eprint = {https://www.pnas.org/content/105/4/1232.full.pdf},
  journal = {Proceedings of the National Academy of Sciences}
}

@article{Cavagna2010Scale,
title={Scale-fFee Correlations In Starling Flocks},
  author={Cavagna, Andrea and Cimarelli, Alessio and Giardina, Irene and Parisi, Giorgio and Santagati, Raffaele and Stefanini, Fabio and Viale, Massimiliano},
  journal={Proceedings of the National Academy of Sciences},
  volume={107},
  number={26},
  pages={11865--11870},
  year={2010},
  publisher={National Acad Sciences}
}

@inproceedings{CUPY,
  author       = "Okuta, Ryosuke and Unno, Yuya and Nishino, Daisuke and Hido, Shohei and Loomis, Crissman",
  title        = "CuPy: A NumPy-Compatible Library for NVIDIA GPU Calculations",
  booktitle    = "Proceedings of Workshop on Machine Learning Systems (LearningSys) in The Thirty-first Annual Conference on Neural Information Processing Systems (NIPS)",
  year         = "2017",
}

@article{LygerosReachability,
  title={On reachability and minimum cost optimal control},
  author={Lygeros, John},
  journal={Automatica},
  volume={40},
  number={6},
  pages={917--927},
  year={2004},
  publisher={Elsevier}
}

@techreport{Dreyfus1966,
  title={{Control Problems With Linear Dynamics, Quadratic Criterion, and Linear Terminal Constraints}},
  author={Dreyfus, Stuart E},
  year={1966},
  institution={Rand Corp, Santa Monica Calif}
}

@book{JacobsonMayne,
author={Jacobson, David H. and Mayne, David Q.},
title={{Differential Dynamic Programming}},
year={1970},
publisher={American Elsevier Publishing Company, Inc., New York, NY},
}

@article{Mitchell2005,
author = {Mitchell, Ian M. and Bayen, Alexandre M. and Tomlin, Claire J.},
issn = {00189286},
journal = {IEEE Transactions on Automatic Control},
number = {7},
pages = {947--957},
title = {{A Time-Dependent Hamilton-Jacobi Formulation of Reachable Sets for Continuous Dynamic Games}},
volume = {50},
year = {2005}
}

@article{LevelSetsBook,
author = {Osher, S and Fedkiw, R},
issn = {0003-6900},
journal = {Applied Mechanics Reviews},
number = {3},
pages = {B15--B15},
title = {{Level Set Methods and Dynamic Implicit Surfaces}},
volume = {57},
year = {2004}
}

@article{Merz1972,
  title={The game of two identical cars},
  author={Merz, AW},
  journal={Journal of Optimization Theory and Applications},
  volume={9},
  number={5},
  pages={324--343},
  year={1972},
  publisher={Springer}
}

@book{Isaacs1965,
  title={Differential Games: A Mathematical Theory with Applications to Warfare and Pursuit, Control and Optimization.},
  year={1999},
  author={Isaacs, R},
  adddress={Courier Corporation, 1999.},
  publisher={Kreiger, Huntigton, NY}
}

@article{Evans1984,
author = {Evans, L.C. and Souganidis, Panagiotis E.},
issn = {0022-2518},
journal = {Indiana Univ. Math. J},
number = {5},
pages = {773--797},
title = {{Differential Games And Representation Formulas For Solutions Of Hamilton-Jacobi-Isaacs Equations}},
volume = {33},
year = {1984}
}

@article{SethianVIIM,
author = {Saye, Robert I. and Sethian, James A.},
issn = {10916490},
journal = {{Proceedings of the National Academy of Sciences of the United States of America}},
number = {49},
pages = {19498--19503},
pmid = {22106269},
title = {{The Voronoi Implicit Interface Method for Computing Multiphase Physics}},
volume = {108},
year = {2011}
}

@article{SethianLSBook,
  title={{Level Set Methods And Fast Marching Methods: Evolving Interfaces In Computational Geometry, Fluid Mechanics, Computer Vision, And Materials Science}}, 
  author={Sethian, James A.},
  journal={Robotica},
  volume={18},
  number={1},
  pages={89--92},
  year={2000},
  publisher={{Cambridge University Press}}
}

@article{SethianFast,
  title={{A Fast Marching Level Set Method For Monotonically Advancing Fronts}},
  author={Sethian, James A},
  journal={Proceedings of the National Academy of Sciences},
  volume={93},
  number={4},
  pages={1591--1595},
  year={1996},
  publisher={National Acad Sciences}
}

@incollection{Sethian87Numerical,
  title={{Numerical Methods for Propagating Fronts}},
  author={Sethian, James A},
  booktitle={Variational methods for free surface interfaces},
  pages={155--164},
  year={1987},
  publisher={Springer}
}

@article{OsherFronts,
author = {Osher, Stanley and Sethian, James A.},
issn = {10902716},
journal = {{Journal of Computational Physics}},
number = {1},
pages = {12--49},
title = {{Fronts Propagating with Curvature-Dependent Speed: Algorithms based on Hamilton-Jacobi Formulations}},
volume = {79},
year = {1988}
}

@article{CrandallLaxFriedrichs,
  title={{Monotone Difference Approximations For Scalar Conservation Laws}},
  author={Crandall, Michael G and Majda, Andrew},
  journal={Mathematics of Computation},
  volume={34},
  number={149},
  pages={1--21},
  year={1980}
}

@article{CrandallMethodFracSteps,
author = {Crandall, Michael and Majda, Andrew},
issn = {0029599X},
journal = {Numerische Mathematik},
number = {3},
pages = {285--314},
title = {{The method of fractional steps for conservation laws}},
volume = {34},
year = {1980}
}

@article{Mitchell2001,
author = {Mitchell, Ian},
journal = {Dept. Aeronautics and Astronautics, Stanford Univ.},
number = {July},
pages = {1--29},
title = {{Games of two identical vehicles}},
year = {2001}
}

@article{Mitchell2020,
  title={{A Robust Controlled Backward Reach Tube with (Almost) Analytic Solution for Two Dubins Cars}},
  author={Mitchell, Ian},
  journal={{EPiC Series in Computing}},
  volume={74},
  pages={242--258},
  year={2020},
  publisher={EasyChair}
}

@article{Crandall1984,
author = {Crandall, M. G. and Evans, L. C. and Lions, P. L.},
issn = {00029947},
journal = {Transactions of the American Mathematical Society},
number = {2},
pages = {487},
title = {{Some Properties of Viscosity Solutions of Hamilton-Jacobi Equations}},
volume = {282},
year = {1984}
}

@article{Crandall1983viscosity,
  title={Viscosity solutions of Hamilton-Jacobi equations},
  author={Crandall, Michael G and Lions, Pierre-Louis},
  journal={Transactions of the American mathematical society},
  volume={277},
  number={1},
  pages={1--42},
  year={1983}
}

@book{Lions1982,
  title={Generalized solutions of Hamilton-Jacobi equations},
  author={Lions, Pierre-Louis},
  volume={69},
  year={1982},
  publisher={London Pitman}
}

@article{CrandallTwoApprox,
  title={{Two Approximations of Solutions of Hamilton-Jacobi Equations}},
  author={Crandall, Michael G and Lions, P-L},
  journal={{Mathematics of Computation}},
  volume={43},
  number={167},
  pages={1--19},
  year={1984}
}

@article{MitchellLSToolbox,
title={A Toolbox of Level Set Methods, version 1.0},
author={Mitchell, Ian},
journal={The University of British Columbia, UBC CS TR-2004-09},
pages={1-94},
year={2004},
month={July},
}

@misc{DAUVerif,
  title={Verification},
  author={{Defense Acquisition University}},
  url={https://www.dau.edu/tools/se-brainbook/Pages/Technical%20Processes/verification.aspx},
  note={Accessed April 5, 2023},
  year={2023}
}

@inproceedings{LevelSetPy,
  author={Molu, Lekan},
  booktitle={2024 IEEE 63rd Conference on Decision and Control (CDC)}, 
  title={{The Python LevelSet Toolbox (LevelSetPy)}}, 
  year={2024},
  volume={},
  number={},
  pages={8938-8945},
  doi={10.1109/CDC56724.2024.10886640}
}

@article{LevelSetTOMS,
  author={Molu, Lekan},
  journal={The ACM Transactions on Mathematical Software}, 
  title={{LevelSetPy: A GPU-Accelerated Package for Hyperbolic Hamilton-Jacobi Partial Differential Equations' Solubility}}, 
  year={2025},
}
\bibliographystyle{abbrvnat}

\etocdepthtag.toc{appendix}

\clearpage
\phantomsection
\section*{Appendices}


\begingroup
\etocsettagdepth{main}{none}
\etocsettagdepth{appendix}{subsubsection}

\etocsetstyle{section}
  {\vspace{0.8em}}
  {}
  {\noindent\textbf{\etocnumber\quad\etocname}\nobreak\dotfill\nobreak\textbf{\etocpage}\par\vspace{2pt}}
  {}

\etocsetstyle{subsection}
  {}
  {}
  {\noindent\hspace{1.8em}\etocnumber\enspace\etocname\nobreak\dotfill\nobreak\etocpage\par\vspace{1pt}}
  {}

\etocsetstyle{subsubsection}
  {}
  {}
  {\noindent\hspace{3.6em}\etocname\nobreak\dotfill\nobreak\etocpage\par}
  {}

\tableofcontents
\endgroup

\newpage
\appendix

\section{Background and Preliminaries.}
\label{app:back_unabridged}

\newcounter{appidx}
\setcounter{appidx}{1}
\renewcommand\theequation{\Alph{appidx}.\arabic{equation}}

\setcounter{equation}{0}


We first introduce the notations that are commonly used throughout the article. Reachable sets within the context of  two person games~\citep{Isaacs1965} and their accompanying  ``viscous" terminal HJ \pde~\citep{Evans1984}\ are then  introduced. This is followed by the HJ-Isaacs (HJI) PDEs commonly used to characterize reachable sets.  

\subsection{Notations and Terminologies}
\label{sec:back::notations}
This appendix builds up, in one place and for a reader outside the reachability subfield, the chain that the main text relies on: from the two-player differential game and its payoff, through the upper and lower \emph{values} of that game, to the Hamilton-Jacobi-Isaacs (HJI) PDE those values solve, and finally to the backward reachable tube \eqref{eq:rcbrt} whose zero level set is the safety certificate we compute. Each item is introduced only where it is first needed.

Conventions: Upper-case and lower-case bold Roman letters are matrices and vectors, respectively; calligraphic letters are sets. Time variables \eg $t_0, t, \tau, T$ are real. The state $\state$ lives in an open set $\openset \subseteq \ren$, with boundary $\partial\openset$ and closure $\bar{\openset}$; the initial/terminal data of the HJ PDE is prescribed over all of $\openset$ (not merely on $\partial\openset$). Two domains must be kept apart. We use $\openset$ only for the viscosity-solution and level-set discussion, \ie as the working region on which the HJ PDE is studied and from which the zero level set of $\valuefunc^\delta$ is extracted. Every heat-kernel representation and every Monte Carlo estimator in this paper is instead posed on all of $\ren$, since the free-space Gaussian kernel is the fundamental solution of the Cauchy problem on $\ren$ and has unit mass over $\ren$ alone~\citep{EvansPDEBook}; restricting such an integral to a proper subdomain would require a reflected Green's function, which we never use. The bounded boxes of \autoref{sec:results} are evaluation windows for isocontouring and error measurement, \ie they are not sampling constraints. We write $\langle \cdot, \cdot \rangle$ for the dot product, $D\bm w$ for the spatial gradient of $\bm w(\state,t)$, and $\nabla \bm w := (D \bm w, \bm w_t)$ for its full space-time gradient; where the differentiation variable matters we write $D_y w$. Two scalar costs recur and we differentiate them in notation: $\sdf(\state)$ (equivalently $\term$) is a \emph{signed-distance function of the state} whose zero sublevel set is the target set, whereas $\mc{J}(\bm u,\bm w)$ is the game's \emph{payoff functional of the control signals}, $(\bm{u}, \bm{w})$; the value functions $\lowervalue,\uppervalue$ are game-optimal values of $\mc J$; they reduce to $\sdf$ at the terminal time.

The dynamical system $\dot{\state}(\tau) = f(t; \state,  \control, \bm{w})$ is influenced by a pursuing player $\pursuer$ (with control $\bm{w} \in \mc{W}  \subseteq \reline^p$) and its evading pair $\evader$ (with control $\bm{u} \in \mc{U} \subseteq \bb{R}^m$). We write $\traj_{\state, t}^{\control, \bm{w}}(\tau) \in \ren$ for its unique state trajectory at time $\tau$ starting from state $\state$ at time $t$ under the control-disturbance pair $(\control,\bm{w})$; this same trajectory is abbreviated $\bm{\xi}(\tau)$ wherever the initial phase and inputs are clear from context (as in \eqref{eq:rcbrt} below). Let the game payoff that records the closest a trajectory comes to the target set  within a two-person differential game over the horizon $t \le \tau \le T$ be
\begin{align}
	\mc{J}(\tau; \state, \control, \bm{w}) = \min_{\tau \in [t, T]} \sdf(\traj_{\state,t}^{\control, \bm{w}}(\tau)),
	\tag{Payoff function}
	\label{eq:distance}
\end{align} 
 where $\sdf$ is a signed distance to the (boundary of the) target set/tube: negative inside the target, positive outside, and zero on its boundary. The controllers belong in compact sets which are  measurable functions \ie $\bar{\mc{U}} \equiv \bm{u}: [t, T] \rightarrow \mathcal{U}, \, \bar{\mc{W}} \equiv  \bm{w}: [t, T] \rightarrow \mathcal{W}$.

\begin{tcolorbox}[title=\textbf{HJ Sign and Time Conventions}, colback=yellow!3, colframe=orange!70, fonttitle=\small, fontlower=\small]
	\small
Note that  \eqref{eq:ivp-viscous} is a well-posed forward parabolic problem and the heat kernel of its linearization at horizon $t$ has covariance $\delta t\, I_n$, degenerating to the target data as $t \downarrow 0$. The dictionary to \emph{physical} time is $t_{\mathrm{phys}} = T - t$: a reader who carries physical time throughout should read every kernel covariance $\delta t\,I_n$ below as $\delta(T - t_{\mathrm{phys}})I_n$, the familiar $(T-t)$ form, in which the terminal data sits at $t_{\mathrm{phys}} = T$ and the PDE reads $\valuefunc_{t_{\mathrm{phys}}} + \hamfunc + \tfrac{\delta}{2}\Delta\valuefunc = 0$ --- the sign of the Laplacian flips with the direction of time, as it must for parabolic well-posedness. Negative signs appearing in the Hamiltonian expressions in this  ~\autoref{app:back_unabridged} arise from the backward-time transformation~\citep{MitchellLSToolbox} and the min-over-max structure of reachability propagation. This convention is standard in reachability literature~\citep{Mitchell2005} and is consistent with the viscosity  solution.
\end{tcolorbox}
%
%
%

\subsection{Dynamic Programming and Two-Person Games.}

\noindent The formal relationships between the dynamic programming (DP) optimality condition for the \textit{value} in differential two-person zero-sum games, and the solutions to PDEs that solve ``min-max" or ``max-min" type nonlinearity (the Isaacs' equation) were presented in~\citep{Isaacs1965}. Essentially, Isaacs' claim was that if the \textit{value} functions are smooth enough, then they solve certain first-order partial differential equations (PDE) problems with  ``max-min" or ``min-max"-type nonlinearity.  However, the DP value functions are seldom regular enough to admit a solution in the classical sense.  ``Weaker" solutions, on the other hand~\citep{Lions1982, Evans1984, Crandall1984, CrandallLaxFriedrichs}, provide generalized ``viscosity" solutions to HJ PDEs under relaxed regularity conditions; these viscosity solutions are not necessarily differentiable anywhere in the state space, and the only regularity prerequisite in the definition is continuity~\citep{Crandall1983viscosity}. However, wherever they are differentiable, they satisfy the  upper and lower values of HJ PDEs (discussed in \autoref{app:back_upper_lower}) in a classical sense. Thus, they lend themselves well to many real-world problems existing at the interface of discrete, continuous, and hybrid systems~\citep{LygerosReachability, OsherFronts, Mitchell2020, Evans1984, Mitchell2005}. 

Matter-of-factly, viscosity solutions to \textit{Cauchy-type}\footnote{Cauchy-type HJ equations are time-dependent versions of the HJ PDE.} HJ Equations are highly useful in backward reachability analysis~\citep{Mitchell2005}. 
For a state $\state \in \openset$ and a fixed time $t$: $0 \le t < T$, suppose that the set of all controls for players $\pursuer$ and $\evader$ are respectively
\begin{align}
	\mathcal{\bar{U}} &\equiv \{\bm{u}: [t, T] \rightarrow \mathcal{U} | \bm{u} \text{ measurable}, \, \mathcal{U} \in \bb{R}^m \}, \quad
	\mathcal{\bar{W}} \equiv \{\bm{w}: [t, T] \rightarrow \mathcal{V} | \bm{v} \text{ measurable},  \,\mathcal{W} \subset \reline^p\}.
\end{align}
\noindent Consider the  differential equation,
%
\begin{align}
	\dot{\state}(\tau) = f(\tau, \state(\tau), \bm{u}(\tau), \bm{w}(\tau)), \,\,
	\state(t) = \state,  \,\, t \le \tau \le T
	\label{eq:sys_dyn}
\end{align}
%
\noindent where $f(\tau, \cdot, \cdot, \cdot)$ and $\state(\cdot)$ are bounded and Lipschitz continuous. This bounded Lipschitz continuity property assures uniqueness of the system response $\state(\cdot)$ to controls $\bm{u}(\cdot)$ and $\bm{v}(\cdot)$~\citep{Evans1984}. 
Associated with \eqref{eq:sys_dyn} is the payoff functional 
\begin{align}
		\mc{J}(t; \state, \bm{u}(\cdot), \bm{w}(\cdot))  \,
	&:=\mc{J}(\bm{u}, \bm{w})  \equiv \int_{t}^{T} l(\tau, \state(\tau), \bm{u}(\tau), \bm{w}(\tau)) d\tau + \bm{g}(\state(T)),
	\label{eq:payoff}
\end{align}
where $\bm{g}(\cdot): \bb{R}^n \rightarrow \bb{R}$ 
satisfies
%
\begin{align}
	| \bm{g}(\state) | \le k_1, \quad 
	| \bm{g}(\state) - \bm{g}(\hat{\state}) \mid \le k_1 | \state - \hat{\state} \mid
\end{align}
%
and $l:[0, T] \times  \reline^n \times \mathcal{U} \times \mc{W} \rightarrow \bb{R}$ is bounded and uniformly continuous, with
%
\begin{align}
	\mid l(t; \state, \control, \disturb) \mid \le k_2, \,
	\mid l(t; \state, \control, \disturb)  -  l(t; \hat{\state}, \control, \disturb) \mid  \le k_2 \mid \state - \hat{\state} \mid
\end{align}
%
for constants $k_1, k_2$ and all 
$0 \le t \le T$, $\hat{\state}, \, \state \in \reline^n$, $\bm{u}\in \mathcal{U}$ and $\disturb \in \mathcal{W}$. 
We call $T$ the \textit{terminal time} (it may be infinity!) and the integral, when it does not depend on the control laws, is the \textit{performance index}. The evader's goal is to maximize the payoff \eqref{eq:payoff} and pursuer's goal is to minimize it. 

\subsection{Lower Values of the Differential Game.}
\label{app:back_upper_lower}

\noindent Suppose that the pursuer's mapping strategy (starting at $t$) is $\beta: \mathcal{\bar{U}}({t}) \rightarrow \mathcal{\bar{W}}({t})$ provided for each $t \le \tau \le T$ and $\bm{u}, \hat{\bm{u}} \in \mathcal{\bar{U}}({t})$; then $\bm{u}(\bar{t}) = \hat{\bm{u}}(\bar{t}) \,\, \text{ a.e. on } t \le \bar{t}  \le \tau$ implies $\beta[\bm{u}](\bar{t}) = \beta[\hat{\bm{u}}](\bar{t}) \,\, \text{ a.e. on } t \le \bar{t}  \le \tau$.
The differential game's lower value for a solution $\state(t)$ that solves \eqref{eq:sys_dyn} for $\bm{u}(t)$ and $\bm{v}(t) = \beta[\control](\cdot)$ is 
\begin{align}
	&\lowervalue(t; \state) = \inf_{\beta \in \mathcal{B}(t)} \sup_{\bm{u} \in \mathcal{U}(t)} \mc{J}(\bm{u}, \beta[\bm{u}])  \,
	\triangleq \inf_{\beta \in \mathcal{B}(t)} \sup_{\bm{u} \in \mathcal{U}(t)} 
	\int_{t}^{T} l(\tau, \bm{x}(\tau), \bm{u}(\tau), \beta[\bm{u}](\tau)) d\tau + \bm{g}\left(\bm{x}(T)\right).
	\label{eq:value_lower}
\end{align}

Similarly, suppose that  the evader's mapping strategy (starting at $t$) is $\alpha: \mathcal{\bar{W}}({t}) \rightarrow \mathcal{\bar{U}}({t})$ provided for each $t \le \tau \le T$ and $\bm{w}, \hat{\bm{w}} \in \mathcal{\bar{W}}({t})$; then  $\bm{w}(\bar{t}) = \hat{\bm{w}}(\bar{t}) \,\, \text{ a.e. on } t \le \bar{t}  \le \tau$ implies $\alpha[\bm{w}](\bar{t}) = \alpha[\hat{\bm{w}}](\bar{t}) \,\, \text{ a.e. on } t \le \bar{t}  \le \tau$. The differential game's upper value for a solution $\bm{x}(t)$ that solves \eqref{eq:sys_dyn} for $\bm{u}(t) = \alpha[\bm{w}](\cdot)$ and $\bm{w}(t)$  is 
\begin{align}
	&\uppervalue(t; \state) = \sup_{\alpha \in \mathcal{A}(t)} \inf_{\bm{w} \in \mathcal{W}(t)}  \mc{J}(\alpha[\bm{w}], \bm{w})  \
	\triangleq  \sup_{\alpha \in \mathcal{A}(t)} \inf_{w \in \mathcal{W}(t)} 
	\int_{t}^{T}l(\tau, \bm{x}(\tau), \alpha[\bm{w}](\tau), \bm{w}(\tau)) d\tau \nonumber \\
	&\qquad \qquad \qquad+ \bm{g}\left(\bm{x}(T)\right).
	\label{eq:value_upper}
\end{align}

The non-local PDEs (\eqref{eq:value_lower} and \eqref{eq:value_upper} are hardly smooth throughout the state space so that they lack classical solutions even for smooth Hamiltonian and boundary conditions. However, the values are ``viscosity" (generalized)  solutions~\citep{Lions1982, Crandall1983viscosity} of the respective HJ-Isaacs (HJI) PDEs, \ie solutions which are \textit{locally Lipschitz} in $\openset \times [0, T]$, and with at most first-order partial derivatives in the Hamiltonian. In backward reachability, we are mostly concerned with the lower value of the differential game for resolving the associated backward reachable sets and tubes. 

\subsection{Viscosity Solution of HJ-Isaac's Equations.}
\label{subsec:visc}

\noindent  
For any optimal control problem a value function is constructed based on the optimal cost (or payoff) of any input phase $(\state, T)$.  In reachability analysis, typically this is defined using a terminal cost function $g(\cdot): \bb{R}^n \rightarrow \bb{R}$ that satisfies 
	\begin{align}
		| \bm{g}(\state) | \le k, \,\,
		| \bm{g}(\state) - \bm{g}(\hat{\state}) \mid \le k | \state - \hat{\state} \mid
	\end{align}

for constant $k$ and all 
$0 \le t \le T$, $\hat{\state}, \, \state \in \reline^n$, $\bm{u}\in \mathcal{U}$ and $\disturb \in \mathcal{W}$.  The zero sublevel set of $\bm{g}(\state)$ \ie
\begin{align}
	\mathcal{L}_0 = \{ \state \in \bar{\Omega} \,|\, \bm{g}(\state) \le 0 \},
	\label{eq:target_set_syk}
\end{align}

%
\begin{lemma}
	The lower value $\lowervalue$ in \eqref{eq:value_lower} is the viscosity solution to the lower Isaac's equation 
	%
		\begin{align}
			&\lowervalue_t + \lowerham (t; \state, \bm{u}, \bm{w}, D\lowervalue) = 0, \,\, t\in \left[0, T\right]\, \state \in \ren,  \quad
			\lowervalue(T; \state) = \bm{g}(\state(T)), \quad \state \in \reline^m
			\label{eq:lower_visc}
		\end{align}
	%
	with lower Hamiltonian, 
	\begin{align}
		&\lowerham (t; \state, \bm{u}, \bm{w}, p) = \max_{u \in \mathcal{U}} \min_{w \in \mathcal{W}} \, \langle \bm{f}(t; \state, \bm{u}, \bm{w}), p  \rangle.
		\label{eq:lower_visc_ham}
	\end{align}
where $p$, the co-state, is the spatial derivative of $\lowervalue$ w.r.t $\state$.
	\label{lemma:lower_visc_lemma}
\end{lemma}
\begin{lemma}
	The upper value $\uppervalue$ in \eqref{eq:value_upper} is the viscosity solution of the upper Isaac's equation 
	\begin{subequations}
		\begin{align}
			&\uppervalue_t + \upperham (t; \state, \bm{u}, \bm{w}, D\uppervalue)= 0, \, t\in \left[0,T\right],\, \state \in \ren,
			 \quad
			\uppervalue(T; \state) = \bm{g}(\state(T)), \quad \state \in \ren
		\end{align}
		\label{eq:upper_visc}
	\end{subequations}
	with upper Hamiltonian, 
	\begin{align}
		\upperham (t; \state, \bm{u}, \bm{w}, p) = \min_{\bm{w} \in \mathcal{W}} \max_{\bm{u} \in \mathcal{U}} \, \langle \bm{f}(t; \state, \bm{u}, \bm{w}), p \rangle,
	\end{align}
	with $p$ being appropriately defined.
	\label{lemma:upper_visc_lemma}
\end{lemma}
\begin{corollary}
	\begin{inparaenum}[(i)]
		\item  $\lowervalue \le \uppervalue \, \text{ over } (t\in \left[0,T\right]\, \state \in \ren)$
		\item if for all $t\in \left[0,T\right], (\state, p) \, \in \ren$,
	\end{inparaenum} 
	the minimax condition is satisfied \ie 	$\upperham(t; \state, \bm{u}, \bm{w}, p) = \lowerham(t; \state, \bm{u}, \bm{w}, p)$, then
	$\lowervalue \equiv \uppervalue$.
\end{corollary}
%

\subsection{Reachability for Systems Verification.}
\noindent Reachability analysis is one of many verification methods that allows us to reason about (control-affine) dynamical systems. 
The verification problem may consist in finding a \textit{set of reachable states} that lie along the trajectory of the solution to a first order nonlinear partial differential equation that originates from some initial state $\state_0 = \state(0)$ up to a specified time bound, $t=t_f$. 
\begin{tcolorbox}[title=\textbf{Reachability Definition and Modes}, colback=yellow!3, colframe=magenta!70, fonttitle=\small, fontlower=\small] 
	\small
	From a set of initial and unsafe state sets, the time-bounded safety verification problem is to determine if there is an initial state and a particular time within the bound that the solution to HJI PDE enters the unsafe set. 	Reachability could be analyzed in a 
	\begin{enumerate}[(i)]
		\item \textit{forward} sense, whereupon system trajectories are examined to determine if they enter certain states from an \textit{initial set};
		\item \textit{backward} sense, whereupon system trajectories are examined to determine if they enter certain \textit{target sets};
		\item \textit{reach set} sense, in which they are examined to see if states reach a set at a \textit{particular time}; or
		\item \textit{reach tube} sense, in which they are evaluated that they reach a set at a point \textit{during a time interval}.	
	\end{enumerate} 
\end{tcolorbox}


Backward reachability consists in avoiding an unsafe set of states under the worst-possible disturbance at all times; relying on nonanticipative control strategies~\citep{Mitchell2005}.  
Backward reachable sets (BRS) and backward reachable tubes (BRTs) are popularly analyzed in a game of two vehicles with non-stochastic dynamics~\citep{Merz1972}. Such BRTs possess discontinuity at cross-over points (which exist at edges) on the surface of the  tube, and may be non-convex. Therefore, treating the end-point constraints under these discontinuity characterizations need careful consideration and analysis when switching control laws if the underlying PDE does not have continuous partial  derivatives (we discuss this further in \autoref{sec:methods}). 
 
\subsection{Robustly Controlled Backward Reachable Set and Tube}
Suppose that the goal of $\pursuer$ is to drive the system into a user-specified target region $\targetset_0(\tau)$ within $\tau\le T$ time steps of playing the game; 
while $\evader$ simultaneously seeks to prevent this from happening. The target set \eqref{eq:target_set} \ie
	\begin{align}
	\mathcal{L}_0(T) = \{ \state \in {\bb{R}}^n \,|\, \valuefunc(0; \state) \le 0 \}
\end{align}
 cost $\bm{g}(0; \state)$ has constraints $| \bm{g}(0; \state) | \le k, \,\,
| \valuefunc(0; \state) - \valuefunc(t; \hat{\state}) \mid \le k | \state - \hat{\state}|$
where $k>0$ is a  constant, and  all  $-T \le t \le 0$, $\{\hat{\state}, \, \state \in \reline^n\}$\footnote{Time is reversed in BRT computational scenarios.}. The distance to $\targetset_0(\tau)$ is typically found by optimizing $\term(\state, t)$ as in \eqref{eq:distance}. 
Substituting the lower value \eqref{eq:value_lower} and its Hamiltonian \eqref{eq:lower_visc_ham} into the HJI PDE \eqref{eq:lower_visc}, and then modifying it so that the target set is captured as a \emph{tube} rather than a set at a single instant, yields the governing equation for this problem. The modification is the $\min\{0,\cdot\}$ \emph{freezing} operator: without it, \eqref{eq:lower_visc} propagates the value in both directions and a state that momentarily touches the target could leave it; the operator clips the growth of $\valuefunc$ so that once a trajectory has entered $\targetset_0$ its value can never climb back above zero, thereby recording membership over the whole interval $[-T,0]$ rather than at time $0$ alone (see~\citep[\S3]{Mitchell2005}). 

The value function of the RCBRT \eqref{eq:HJ-RCBRT} \ie 
\begin{align}
	\valuefunc_t(t; \state) + \min\{0, \hamfunc\left(t; \state, D \valuefunc(t; \state)\right)\} = 0, \quad
	\valuefunc(0; \state) = \bm{g}(0; \state),
\end{align}  
has $\hamfunc$ as the lower Hamiltonian \eqref{eq:lower_visc_ham} and $\bm g(0;\state)$ as the signed-distance target data.
%

\begin{tcolorbox}[title=\textbf{The Robustly Controlled Backward Reachable Tube (RCBRT)}, colback=pink!3, colframe=orange!70, fonttitle=\small, fontlower=\small] 
	\small
	For the safety problem setup in \eqref{eq:value_lower},  the corresponding \textit{robustly controlled backward reachable \textbf{tube}} (RCBRT)~\citep{Mitchell2020} on $(0, T]$ 
	is the closure of the open set,
	\begin{align}
		\mathcal{L}([-T, 0], \mathcal{L}_0) = \{\state \in \ren \,| \, \exists \, \beta \in \mathcal{B}(t)  \forall \, \bm{u} \in \mathcal{U}(t), \,
		\qquad  \exists \,
		\bar{\tau} \in [-T, 0], \bm{\xi}(\bar{\tau}) \in  \mathcal{L}_0 \}.
	\end{align}
	Read: The set of states from which there exists a strategy of $\pursuer$ such that, for all controls of $\evader$, the resulting trajectory \textit{reaches and remains in the target set} within the interval $[-T, 0]$.  Following Lemma 2 of \citep{Mitchell2005}, the states in the reachable set admit the following properties w.r.t the value function $\valuefunc$,
	\begin{subequations}
		\begin{align}
			\state(t)\in \mathcal{L}(\cdot) \implies \valuefunc(t; \state) \le 0, \,\,
			\valuefunc(t; \state) \le 0 \implies \state(t) \in \mathcal{L}(\cdot). 
		\end{align}
	\end{subequations}
\end{tcolorbox}
In \textit{backward reach avoid tubes}, the agent must avoid the unsafe region at all times. We can write \eqref{eq:HJ-RCBRT}  a \textit{robustly controlled backward reach-avoid tube} (RCBRAT) as,
\begin{align}
	\min\{\valuefunc_t(t; \state) + \hamfunc\left(t; \state, D \valuefunc\right), \bm{g}(t; \state) - \sdf(t; \state)\} \le 0,  \quad
	\valuefunc(0; \state) = \bm{g}(0; \state).
	\tag{HJI-RCBRAT}
	\label{eq:HJI-RCBRAT}
\end{align}

\newpage
\section{HJ PDE Linearization}
\label{app:proofs}

\setcounter{appidx}{2}
\setcounter{equation}{0}

In this appendix, we construct the Cole-Hopf-type linearization of the viscous HJI PDE and propose a sampling machinery for the solution, $\valuefunc(t; \state)$.  We show that the transformation is \emph{exact} only when the Hamiltonian  is quadratic in the co-state \ie $\hamfunc = \tfrac{1}{2}\langle \bp^\top, \bp \rangle$, and we derive the residual it incurs for general Hamiltonians. We then formulate
the quasi-linearization iterative algorithm for computing \eqref{eq:HJ-RCBRT-Visc} and \eqref{eq:HJI-RCBRAT-Visc}. We finish this appendix with the correct Gaussian expectation formulas for recovering the value and its spatial gradient.

\subsection{The viscous HJ Equation's solution}
We express the viscous HJ equation's solution as the logarithm of the Gaussian kernels' expectation that parameterize the state space. Let us construct the spatial and time derivatives~\eqref{eq:ivp-viscous}.
%

\begin{lemma}[Exact residual of the generalized Cole-Hopf transformation]
	\label{lem:exact_residual}
	Let $\valuefunc^\delta \in C^{1,2}$ solve \eqref{eq:ivp-viscous}, let $\bc \in C^{1,2}\left(\bar\openset \times [0,T]\right)$, and set $\valueparam^\delta \triangleq \exp(-\bc\,\valuefunc^\delta)$. Then, writing $\valuefunc$ for $\valuefunc^\delta$ and suppressing arguments,
	\begin{align}
		\label{eq:exact_residual}
		\valueparam_t^\delta - \frac{\delta}{2}\Delta\valueparam^\delta
		= \valueparam^\delta \underbrace{\bigg[\bc\Big(\hamfunc - \frac{\delta}{2}\bc\,|D\valuefunc|^2\Big)\bigg]}_{\text{algebraic residual, } \bR_{\mathrm{alg}}}
		+ \valueparam^\delta \underbrace{\bigg[- \bc_t \valuefunc + \frac{\delta}{2}(\Delta \bc)\valuefunc - \frac{\delta}{2}|D\bc|^2\valuefunc^2 + \delta\,(1 - \bc\valuefunc)\, D\bc \cdot D\valuefunc\bigg]}_{\text{derivative residual, } \bR_{\mathrm{der}}}.
	\end{align}
\end{lemma}
\begin{proof}[Proof of Lemma \ref{lem:exact_residual}]
	\label{proof:exact_residual}
	Write $\varphi \triangleq \bc\,\valuefunc$ so that $\valueparam^\delta = e^{-\varphi}$. Then,
	\begin{align}
		\label{eq:phi_derivatives}
		\valueparam^\delta_t = -\varphi_t\,\valueparam^\delta, \qquad
		D\valueparam^\delta = -D\varphi\,\valueparam^\delta, \qquad
		\Delta\valueparam^\delta = \left(|D\varphi|^2 - \Delta\varphi\right)\valueparam^\delta,
	\end{align}
	with $\varphi_t = \bc_t\valuefunc + \bc\valuefunc_t$, $D\varphi = \valuefunc\,D\bc + \bc\,D\valuefunc$, and $\Delta\varphi = \valuefunc\,\Delta\bc + 2\,D\bc\cdot D\valuefunc + \bc\,\Delta\valuefunc$. Substituting these into $\valueparam^\delta_t - \frac{\delta}{2}\Delta\valueparam^\delta$ and eliminating $\valuefunc_t = \frac{\delta}{2}\Delta\valuefunc - \hamfunc$ via \eqref{eq:ivp-viscous}, it can be verified that
	\begin{align}
		\label{eq:residual_expansion}
		\valueparam^\delta_t - \frac{\delta}{2}\Delta\valueparam^\delta
		&= \valueparam^\delta\left[\bc\hamfunc - \frac{\delta}{2}\bc^2|D\valuefunc|^2 - \bc_t\valuefunc + \frac{\delta}{2}(\Delta\bc)\valuefunc - \frac{\delta}{2}|D\bc|^2\valuefunc^2 + \delta(1 - \bc\valuefunc)D\bc\cdot D\valuefunc\right].
	\end{align}
	%
	Grouping the first two terms as $\bR_{\mathrm{alg}}$ and the remainder as $\bR_{\mathrm{der}}$ gives \eqref{eq:exact_residual}.
\end{proof}

Notice that the \emph{algebraic} residual $\bR_{\mathrm{alg}}$ vanishes if we set $	\bc(t; \state)$ to \eqref{eq:coeff_frozen} \ie
\begin{align}
	\bc(t; \state) = \tfrac{2}{\delta}\cdot\hamfunc^\delta/|D\valuefunc^\delta|^2.
	\label{eq:coeff_frozen_app}
\end{align}
In addition, the \emph{derivative} residual $\bR_{\mathrm{der}}$ consists of ``live" terms $\bc_t, D\bc, \Delta\bc$, and no freezing of the value of $\bc$ removes the derivatives of $\bc$. The frozen-coefficient step of Algorithm~\ref{alg:quasi_lin} therefore does \emph{not} transform \eqref{eq:ivp-viscous} into a heat equation; rather, it \emph{defines} the next iterate as the exact solution of the linear heat initial-value problem, and $\bR_{\mathrm{der}}$, which is discarded at each step, is  the defect whose  norm we bound in Theorem~\ref{thm:defect} (\autoref{sec:defect-bound}).

\begin{corollary}[Exactness in the quadratic case]
	\label{cor:exact_case}
	The full residual in \eqref{eq:exact_residual} vanishes identically iff $\bR_{\mathrm{alg}} \equiv 0$ and $\bc$ is constant in $(t; \state)$. In particular, when $\hamfunc = \tfrac{1}{2}\langle \bp, \bp \rangle$, the choice \eqref{eq:coeff_frozen_app} yields $\bc = 1/\delta = \mathrm{const}$ so that both residuals vanish, and $\valueparam^\delta = \exp(-\valuefunc^\delta/\delta)$ satisfies the heat equation exactly.
\end{corollary}

\begin{proof}[Proof of Proposition \ref{prop:heat_transformation}]
	\label{proof:heat_transformation}
	With $\bc^{(k)}$ frozen at the current iterate, the algorithm step is, by construction, the initial value problem,
	\begin{align}
	\valueparam_t^\delta - \dfrac{\delta}{2} \Delta \valueparam^\delta = 0
	 \, \text{ in } \, \ren \times (0, T], \quad
	\valueparam^\delta(0; \state) = \exp\left(-\bc^{(k)}(0; \state)\, \bm g(\state)\right) \, \text{ on } \, \ren \times \{t=0\},
		\label{eq:Heat-Equation-App}
	\end{align}
	whose solution is unique and bounded for bounded continuous data. By Lemma~\ref{lem:exact_residual} and Corollary~\ref{cor:exact_case}, this step is an exact transformation of \eqref{eq:ivp-viscous} when $\hamfunc = \tfrac{1}{2}|\bp|^2$; for general Hamiltonians it discards the derivative residual $\bR_{\mathrm{der}}$ of \eqref{eq:exact_residual}, and is to be read as one step of the Picard quasi-linearization whose fixed-point defect is quantified in \autoref{sec:defect-bound}. Equation \eqref{eq:Heat-Equation-App} is the classical Cole-Hopf-transformed equation~\citep{EvansPDEBook} of the Eikonal version of \eqref{eq:ivp-viscous}\footnote{For the Eikonal version of \eqref{eq:ivp-viscous}, we set $\hamfunc(D\valuefunc(\state))=0$ in $\openset$, and $\valuefunc(0; \state)=0$ on $\partial \openset$.}.

	It follows that the unique bounded solution of \eqref{eq:Heat-Equation-App} is given by the Green's function convolution, 
	\begin{subequations}
		\begin{align}
		\valueparam^\delta(t; \state) &= \dfrac{1}{(\sqrt{2 \pi \delta t })^{n}} \int_{\ren} \exp \left(-\frac{1}{2 \delta t} \mid \state - \bm{y} \mid^2\right)
		\exp\left(-\bc\bm g(\bm y)\right) d\bm{y}, \,\, 
		\label{eq:colehopf_app}\\
		&\triangleq \bb{E}_{\statey \sim \mc{N}(\state, \delta t I_n)} \left[\exp\left(-\bc \bm g(\bm y)\right)\right]
		:=\bb{E}_{\statey \sim \mc{N}(\state, \delta t I_n)} \left[\exp\left( -\bc \cdot \bm g(\bm y)\right)\right]
		\label{eq:colehopf_app_b}
		\end{align}
		\label{eq:feynman_kac}
	\end{subequations}
	where \eqref{eq:colehopf_app} is the standard heat kernel convolution and the kernel has been rewritten in \eqref{eq:colehopf_app_b}, via the Feynman-Kac formula, as a Gaussian density with mean $\state$ and covariance $\delta t I_n$. Both lines are posed on $\ren \times (0, T]$, \ie the integration is over all of $\ren$ and no boundary condition is imposed on a proper subdomain. This is the only setting in which the free-space kernel has unit mass and \eqref{eq:colehopf_app_b} is a genuine expectation; restricting the integral to a bounded $\openset$ would retain only part of the mass and would fail to recover the datum as $t \downarrow 0$. The transformed datum $\exp(-\bc\bm g)$ is continuous, positive and bounded on $\ren$ and decays at infinity by Remark~\ref{rem:admissible_data}, so \eqref{eq:feynman_kac} is the unique bounded solution of \eqref{eq:Heat-Equation-App} and $\valueparam^\delta(t;\state) \to 0$ as $|\state| \to \infty$.
\end{proof}

\begin{proof}[Proof of Lemma \ref{lemma:hj_solution}]
Going by the transformation  \eqref{eq:cole-hopf}, we can write
\begin{align}
	\valuefunc^\delta = -(\delta/2) \dfrac{|D\valuefunc^\delta|^2}{\hamfunc(t; \state, D\valuefunc^\delta)}\log \valueparam^\delta := -\dfrac{1}{\bc} \log \valueparam^\delta,
	\label{eq:value_heat_kernel_rel}
\end{align}
so that the unique bounded solution to the initial-value \eqref{eq:ivp-viscous} \ie \eqref{eq:ivp-viscous-solution} becomes
\begin{subequations}
	\begin{align}
		\valuefunc^\delta(t; \state) &= -\dfrac{1}{\bc} \cdot \log \bigg\{\dfrac{1}{(\sqrt{2 \pi \delta t })^{n}}\int_{\ren} \exp \left(-\frac{1}{2\delta t}\cdot |\state-\bm{y}|^2\right) \exp\left(-\bc \cdot  \bm g(\bm y)\right) \bigg\}, \\
		%
		%
		&\triangleq -\dfrac{1}{\bc} \cdot  \log \bigg\{\bb{E}_{\statey \sim \mc{N}(\state, \delta tI_n)} \left[\exp\left(-\bc \cdot  \bm g (\bm{y})\right)\right]\bigg\}.
	\end{align}
\end{subequations}
\begin{corollary}
	For the backward reachability problems \eqref{eq:HJ-RCBRT-Visc} and \eqref{eq:HJI-RCBRAT-Visc} over time horizon $[0, T]$, the solution is given by the \textit{log-sum-exp} identity with samples from the correct backward-time covariance:
	\begin{align}
	\valuefunc^\delta(t; \state) = \dfrac{-1}{\bc}\log \frac{1}{N} \sum_{i=1}^N \exp\left(-\bc \bm{g}(\bm{y}_i)\right), \quad \text{ with } \bm{y}_i \stackrel{\text{iid}}{\sim}
	\mc{N}(\state, \delta tI_n).
	\end{align}
\end{corollary}

\textit{A fortiori}, we have the solution to the viscous HJ equation as the log of the expectation of a Gaussian density with mean $\state$ and covariance $\delta tI_n$.
\end{proof}

\subsection{Spatial Gradient of the HJ Payoff}
\begin{proof}[Proof of Lemma \ref{lemma:grad_hj_sol}]	
	From \eqref{eq:colehopf_app} observe, 
	\begin{subequations}
		\begin{align}
		D \valueparam^\delta(t; \state)&= -\dfrac{1}{(\sqrt{2 \pi \delta t })^{n}} \int_{\ren} \frac{(\state - \statey)}{ \delta t }\exp \left(-\frac{\mid \state - \bm{y} \mid^2}{2 \delta t}\right) \exp\left(-\bc \cdot  \bm g(\bm y)\right) d\bm{y}, \label{eq:cole-hopf-spatial} \\ 
		&= -\frac{1}{\delta t}\bb{E}_{\statey \sim \mc{N}(\state, \delta tI_n)} \left[(\state - \statey)\exp\left(-\bc \cdot \bm g(\bm y) \right)\right].
		\end{align}
	\end{subequations}
	Inspecting \eqref{eq:value_heat_kernel_rel}, we may write
	%
		\begin{align}
		D\valuefunc^\delta(t; \state) &= -\dfrac{1}{\bc} D[\log \valueparam^\delta(t; \state)] = -\dfrac{1}{\bc} \dfrac{D\valueparam^\delta(t; \state)}{\valueparam^\delta(t; \state)}\nonumber \\
		& = \dfrac{1}{\delta \cdot t \cdot \bc} \cdot
		\dfrac{\bb{E}_{\statey \sim\mc{N}(\state, \delta tI_n)}\left[(\state-\statey) \exp(-\bc \cdot \bm{g(y)})\right]}{\bb{E}_{\statey \sim\mc{N}(\state, \delta tI_n)}\left[\exp(-\bc \cdot \bm{g(y)})\right]}, \\
		&\triangleq  \dfrac{1}{t \cdot \delta \cdot \bc} \cdot
		\left(\state -\dfrac{\bb{E}_{\statey \sim\mc{N}(\state, \delta tI_n)}\left[\statey \exp(-\bc \cdot \bm{g(y)})\right]}{\bb{E}_{\statey \sim\mc{N}(\state, \delta tI_n)}\left[\exp(-\bc \cdot \bm{g(y)})\right]}\right).
		\label{eq:spatial_cole_hopf}
		\end{align}
\end{proof}

\begin{corollary}[Log-sum-exp estimator for the value function]
	\label{cor:logsumexp}
	For \eqref{eq:HJ-RCBRT-Visc} and \eqref{eq:HJI-RCBRAT-Visc}, the terminal cost argument transforms to $\bm{g}(\state + \sqrt{\delta t}\statey_i)$. For numerical stability the \emph{log-sum-exp} identity gives,
	\begin{align}
		\label{eq:logsumexp}
		\valuefunc^\delta(t; \state) &= \dfrac{-1}{\bc^{(k)}}\log \frac{1}{N} \sum_{i=1}^N \exp\!\left(-\bc^{(k)} \bm{g}\!\left(\state + \sqrt{\delta t}\bm{y}_i\right)\right),
	\end{align}
	with $\statey_i \stackrel{\text{iid}}{\sim} \mc{N}(\bm{0}, I_n)$.
\end{corollary}
\begin{corollary}[Monte Carlo gradient estimator]
	\label{cor:mc_gradient}
	With samples $\bs_i = \state + \sqrt{\delta t}\statey_i$, $\statey_i \stackrel{\text{iid}}{\sim} \mc{N}(\bm{0},I_n)$, the importance-weighted estimator for \eqref{eq:spatial_grad} is,
	\begin{align}
		\label{eq:mc_grad}
		D\valuefunc^\delta(t; \state) &= \dfrac{1}{t \cdot \delta \cdot \bc^{(k)}}
		\left(\state - \dfrac{\dfrac{1}{N}\sum_{i=1}^{N}\bs_i \cdot \exp\!\left(-\bc^{(k)} \bm{g}(\bs_i)\right)}{\dfrac{1}{N}\sum_{i=1}^{N}\exp\!\left(-\bc^{(k)} \bm{g}(\bs_i)\right)}\right).
	\end{align}
\end{corollary}

\begin{corollary}[Tilted-proposal estimators]
	\label{cor:tilted_estimators}
	Fix a shift $\theta \in \ren$ and let $\sigma^2 \triangleq \delta t$. Draw $\statey_i \stackrel{\text{iid}}{\sim} q_\theta \triangleq \mc{N}(\state + \theta, \sigma^2 I_n)$ and set the density-ratio and total weights,
	\begin{align}
		\label{eq:tilt_weights}
		w_\theta(\statey) \triangleq \exp\!\left(\frac{|\theta|^2 - 2\langle \statey - \state, \theta\rangle}{2\sigma^2}\right),
		\qquad
		\tilde w_i \triangleq w_\theta(\statey_i)\,\exp\!\left(-\bc^{(k)}\bm g(\statey_i)\right).
	\end{align}
	Then, the tilted value and gradient estimators,
	\begin{align}
		\label{eq:tilted_value}
		\hat\valuefunc^\delta_\theta(t;\state) &= -\frac{1}{\bc^{(k)}}\log\left(\frac{1}{N}\sum_{i=1}^N \tilde w_i\right), \\
		\label{eq:tilted_gradient}
		\widehat{D\valuefunc}^\delta_\theta(t;\state) &= \frac{1}{t\cdot\delta\cdot\bc^{(k)}}\left(\state - \frac{\sum_{i=1}^N \statey_i\, \tilde w_i}{\sum_{i=1}^N \tilde w_i}\right),
	\end{align}
	are consistent for \eqref{eq:hj_sol} and \eqref{eq:spatial_grad} respectively, for \emph{every} $\theta$; the inner average of \eqref{eq:tilted_value} is unbiased for $\bb{E}_{\statey\sim\mc{N}(\state,\sigma^2 I_n)}[e^{-\bc^{(k)}\bm g(\statey)}]$. The first-order Laplace choice $\theta^{\star} = -\sigma^2\bc^{(k)} D\valuefunc^{(k)}(t;\state)$ of \eqref{eq:tilt_shift} Gaussianizes the zero-variance proposal.
\end{corollary}

\begin{proof}[Proof of Corollary~\ref{cor:tilted_estimators}]
	\label{proof:tilted_estimators}
	Unbiasedness is the change of measure: with $\varphi_{\state}$ the $\mc{N}(\state,\sigma^2I_n)$ density, $w_\theta = \varphi_{\state}/q_\theta$ by direct computation of the two exponents, so $\bb{E}_{q_\theta}[w_\theta e^{-\bc\bm g}] = \int \varphi_{\state}(\statey)e^{-\bc\bm g(\statey)}d\statey = \bb{E}_{\varphi_{\state}}[e^{-\bc\bm g}]$, and likewise for the numerator of \eqref{eq:tilted_gradient}; the ratio and the logarithm preserve consistency by the continuous mapping theorem. For the shift, the zero-variance proposal is proportional to the integrand $\varphi_{\state}(\statey)e^{-\bc\bm g(\statey)} \propto \exp\left(-|\statey-\state|^2/2\sigma^2 - \bc\bm g(\statey)\right)$; linearizing $\bm g$ about $\state$ and completing the square in the exponent gives $\mc{N}\left(\state - \sigma^2\bc\, D\bm g(\state), \sigma^2 I_n\right)$, and replacing $D\bm g$ by the running gradient iterate $D\valuefunc^{(k)}$ yields \eqref{eq:tilt_shift}.
\end{proof}

\subsection{Sampling Complexity Analysis}
We now analyze the sampling complexity of \autoref{alg:quasi_lin}.
\begin{proof}[Proof of Theorem~\ref{thm:mc_complexity}]
	\label{proof:mc_complexity}
	Recall that,
	\begin{align}
		\label{eq:proof_defs}
		Z \triangleq \exp(-\bc \, \bm{g}(\bm{\zeta})), \quad
		\mu \triangleq \bb{E}[Z], \quad
		\valuefunc_{\bc}(t;\state) \triangleq -\frac{1}{\bc}\log \mu.
	\end{align}
	For i.i.d. samples $\bm{\zeta}_1,\ldots,\bm{\zeta}_N \sim \mc{N}(\state, \delta tI_n)$, 
	$ Z_i \triangleq \exp(-\bc \, \bm g(\bm{\zeta}_i))$, so that 
	$\bar{Z}_N \triangleq \frac{1}{N}\sum_{i=1}^N Z_i$, and  
	$\hat{\valuefunc}_{\bc,N}(t;\state) \triangleq -\frac{1}{\bc}\log \bar Z_N$.
	
	Let $\alpha \triangleq e^{-\bc g_{\max}}, \, \beta \triangleq e^{-\bc g_{\min}}.$
	Since $g_{\min} \le \bm g(\bm{\zeta}) \le g_{\max}$ and $\bc > 0$, monotonicity of the exponential implies
	\begin{align}
		e^{-\bc g_{\max}} \le e^{-\bc \bm g(\bm{\zeta})} \le e^{-\bc g_{\min}}
	\end{align}
	almost surely. By definition of $\alpha$ and $\beta$, this is exactly $\alpha \le Z \le \beta$ almost surely. Hence each $Z_i$ is bounded in $[\alpha,\beta]$.
	
	\paragraph{Value residuals, $\bar Z_N$.}
Since $\valuefunc_{\bc}(t;\state) = -\frac{1}{\bc}\log \mu$ and $\hat{\valuefunc}_{\bc,N}(t;\state) = -\frac{1}{\bc}\log \bar Z_N$, 	we find that
	\begin{align}
		\hat{\valuefunc}_{\bc,N}(t;\state) - \valuefunc_{\bc}(t;\state)
		= -\frac{1}{\bc}\log\left({\bar Z_N}/{\mu}\right).
	\end{align}
	
	Hence, $
	\left|\hat{\valuefunc}_{\bc,N}(t;\state) - \valuefunc_{\bc}(t;\state)\right| \ge \varepsilon$ implies that $
	\left|\log\left({\bar Z_N}/{\mu}\right)\right| \ge \bc\varepsilon$,
	or that $\bar Z_N \ge \mu \exp({\bc\varepsilon})
	\, \text{or } 
	\bar Z_N \le \mu \exp({-\bc\varepsilon})$. 
	Now,  we can convert the log-scale deviation in the viscous HJ value $|\hat v_{c,N} - v_c| \ge \varepsilon$ into two
	linear-scale deviation events on the sample mean $\bar Z_N$, as 
	\begin{align}
		\bb P\left(\left|\hat{\valuefunc}_{\bc,N}(t;\state) -\valuefunc_{\bc}(t;\state)\right| \ge \varepsilon\right) &\le \bb P\left(\bar Z_N - \mu \ge \mu(\exp({\bc\varepsilon})-1)\right) \nonumber  \\ 
		&\qquad + \bb P\left(\bar Z_N - \mu \le -\mu(1-\exp({-\bc\varepsilon})\right).
		\label{eq:value_residuals}
	\end{align}
	
	\paragraph{Hoeffding bounds for tails.} We now employ  Hoeffding's concentration inequality to bound how far the sample mean, $\bar{Z}_N$ can deviate from the true mean in the presence of the bounded variables $\hat{\valuefunc}_{\bc,N}(t;\state), \valuefunc_{\bc}(t;\state)$.
	
	Because $Z_1,\ldots,Z_N$ are i.i.d. and each lies in $[\alpha, \beta]$, Hoeffding's inequality implies that for every $s > 0$,
	\begin{align}
		\bb P(\bar Z_N - \mu \ge s) \le \exp\left(-\frac{2Ns^2}{(\beta-\alpha)^2}\right),
		\qquad
		\bb P(\bar Z_N - \mu \le -s) \le \exp\left(-\frac{2Ns^2}{(\beta-\alpha)^2}\right).
	\end{align}
	Applying the first inequality with $s_+ \triangleq \mu(\exp(\bc\varepsilon)-1)$ and the second with $s_- \triangleq \mu(1-\exp(-\bc\varepsilon))$, we find that
	\begin{subequations}
		\begin{align}
			\bb P\left(\bar Z_N - \mu \ge \mu(\exp(\bc\varepsilon)-1)\right)
			&\le \exp\left(-\frac{2N\mu^2(\exp(\bc\varepsilon)-1)^2}{(\beta-\alpha)^2}\right), \\
			\bb P\left(\bar Z_N - \mu \le -\mu(1-\exp(-\bc\varepsilon))\right)
			&\le \exp\left(-\frac{2N\mu^2(1-\exp(-\bc\varepsilon))^2}{(\beta-\alpha)^2}\right).
		\end{align}
		\label{eq:hoeffding_bounds}
	\end{subequations}
	
	\paragraph{Union bound and simplification.}
	Factoring the previous bounds \eqref{eq:hoeffding_bounds} into \eqref{eq:value_residuals}, we must have 
	\begin{align}
		\label{eq:union_bound}
		\bb P\left(\left|\hat{\valuefunc}_{\bc,N}(t;\state) - \valuefunc_{\bc}(t;\state)\right| \ge \varepsilon\right)
		&\le \exp\left(-\frac{2N\mu^2(\exp(\bc\varepsilon)-1)^2}{(\beta-\alpha)^2}\right)
		\nonumber \\
		& \qquad + \exp\left(-\frac{2N\mu^2(1-\exp(-\bc\varepsilon))^2}{(\beta-\alpha)^2}\right).
	\end{align}
	Since $\exp(\bc\varepsilon)-1 \ge 1-\exp(-\bc\varepsilon)$ for every $\varepsilon > 0$, the first exponential is no larger than the second. Therefore
	\begin{align}
		\bb P\left(\left|\hat{\valuefunc}_{\bc,N}(t;\state) - \valuefunc_{\bc}(t;\state)\right| \ge \varepsilon\right)
		\le 2\exp\left(-\frac{2N\mu^2(1-e^{-\bc\varepsilon})^2}{(\beta-\alpha)^2}\right).
	\end{align}
	\textit{A  fortiori}, this proves the claim of \autoref{thm:mc_complexity}.
\end{proof}

\begin{proof}[Proof of Corollary~\ref{cor:mc_complexity}]
	\label{proof:cor_mc_complexity}	
	Since $\alpha \le Z$ almost surely, taking expectations yields $\mu = \bb E[Z] \ge \alpha$. Hence it is enough to require
	\begin{align}
		2\exp\left(-\frac{2N\alpha^2(1-e^{-\bc\varepsilon})^2}{(\beta-\alpha)^2}\right) \le \alpha.
	\end{align}
	Taking logarithms and solving for $N$ gives
	\begin{align}
		N \ge \frac{(\beta-\alpha)^2}{2\alpha^2(1-e^{-\bc\varepsilon})^2}\log\frac{2}{\alpha},
	\end{align}
	which proves the second claim.
\end{proof}

\begin{proof}[Proof of Theorem~\ref{thm:algorithm_convergence}]
	\label{proof:algorithm_convergence}
	We proceed in several steps.	
	\paragraph{Step 1: Lipschitz continuity of the frozen-coefficient solve map $\Phi$.}
	Fix an index $m$ and define
	\begin{align}
		\psi_m(s) \triangleq -\frac{1}{s}\log M_m(s),
		\qquad
		M_m(s) \triangleq \bb E_{\bm{\zeta} \sim \mc N(\state_m,\delta tI_n)}\left[\exp\left(-s\bm g(\bm{\zeta})\right)\right].
	\end{align}
	Then $(\Phi(c))_m = \psi_m(c_m)$. Since $M_m(s) > 0$,
	\begin{align}
		\frac{d}{d s}(\psi_m(s)) = \frac{1}{s^2}\log M_m(s) - \frac{1}{s}\frac{M_m'(s)}{M_m(s)}.
	\end{align}
	Differentiating under the expectation gives
	\begin{align}
		\frac{d}{ds}(M_m(s))
		= \bb E_{\bm{\zeta} \sim \mc N(\state_m,\delta tI_n)}\left[-\bm g(\bm{\zeta})\exp\left(-s\bm g(\bm{\zeta})\right)\right].
	\end{align}
	Thus, 
	\begin{align}
		-\frac{M_m'(s)}{M_m(s)} = \frac{\bb E\left[\bm g(\cdot)e^{-s\bm g(\cdot)}\right]}{\bb E\left[e^{-s\bm g(\cdot)}\right]},
	\end{align}
	which is a weighted expectation of $\bm g$; and possesses absolute value of at most $G$ by assumption \ref{ass:grad_est::Gpos}. In addition, $|\bm g| \le G$ implies that 
	\begin{align}
		e^{-sG} \le M_m(s) \le e^{sG}.
	\end{align}
	We can therefore write $|\log M_m(s)| \le sG$. Substituting these two bounds into the derivative formula yields
	\begin{align}
		|\psi_m'(s)| \le \frac{1}{s^2}(sG) + \frac{1}{s}G = \frac{2G}{s} \le \frac{2G}{c_{\min}}
	\end{align}
	for every $s \in [c_{\min},c_{\max}]$. By the mean-value theorem,
	\begin{align}
		|\psi_m(c_m)-\psi_m(\widetilde c_m)| \le \frac{2G}{c_{\min}}|c_m-\widetilde c_m|.
	\end{align}
	Taking the maximum over $m$, we must have
	\begin{align}
		\|\Phi(c)-\Phi(\widetilde c)\|_\infty \le \frac{2G}{c_{\max}}\|c-\widetilde c\|_\infty.
		\label{eq:lipschitz_phi}
	\end{align}
	
	\paragraph{Step 2: Lipschitz continuity of the coefficient-update map $\Gamma$.}
	Fix $m$ and write $p \triangleq p_m(v)$, $q \triangleq p_m(w)$. Then
	\begin{align}
		|(\Gamma(v))_m-(\Gamma(w))_m|
		= \frac{2}{\delta}\left|\frac{\hamfunc(t;\state_m,p)}{|p|^2} - \frac{\hamfunc(t;\state_m,q)}{|q|^2}\right|.
	\end{align}
	Add and subtract $\hamfunc(t;\state_m,q)/|p|^2$ to obtain
	\begin{align}
		\left|\frac{\hamfunc(t;\state_m,p)}{|p|^2} - \frac{\hamfunc(t;\state_m,q)}{|q|^2}\right|
		&\le \frac{|\hamfunc(t;\state_m,p)-\hamfunc(t;\state_m,q)|}{|p|^2}
		+ |\hamfunc(t;\state_m,q)|\left|\frac{1}{|p|^2}-\frac{1}{|q|^2}\right|.
	\end{align}
	By assumption (ii), $|p|,|q| \ge m_0$, so that $|p|^2, |q|^2 \ge m_0^2$. Using assumption (iii), the first term is bounded by
	\begin{align}
		\frac{|\hamfunc(t;\state_m,p)-\hamfunc(t;\state_m,q)|}{|p|^2}
		\le \frac{L_H}{m_0^2}|p-q|.
	\end{align}
	For the second term, note that
	\begin{align}
		\left|\frac{1}{|p|^2}-\frac{1}{|q|^2}\right|
		= \frac{\bigl||q|^2-|p|^2\bigr|}{|p|^2|q|^2}
		\le \frac{\bigl||q|-|p|\bigr|(|q|+|p|)}{m_0^4}
		\le \frac{2P_*}{m_0^4}|p-q|,
	\end{align}
	where the final inequality uses $|p|,|q| \le P_*$ from assumption (ii). Using $|\hamfunc(t;\state_m,q)| \le H_*$ from assumption (iii), we conclude that
	\begin{align}
		\left|\frac{\hamfunc(t;\state_m,p)}{|p|^2} - \frac{\hamfunc(t;\state_m,q)}{|q|^2}\right|
		\le \left(\frac{L_H}{m_0^2} + \frac{2H_*P_*}{m_0^4}\right)|p-q|.
	\end{align}
	Therefore
	\begin{align}
		|(\Gamma(v))_m-(\Gamma(w))_m|
		\le \frac{2}{\delta}\left(\frac{L_H}{m_0^2} + \frac{2H_*P_*}{m_0^4}\right)|p-q|.
	\end{align}
	Finally, by assumption (iv),
	\begin{align}
		|p-q| \le \|\mc G(v)-\mc G(w)\|_\infty \le L_D\|v-w\|_\infty.
	\end{align}
	Taking the maximum over $m$ yields
	\begin{align}
		\|\Gamma(v)-\Gamma(w)\|_\infty
		\le \frac{2L_D}{\delta}\left(\frac{L_H}{m_0^2} + \frac{2H_*P_*}{m_0^4}\right)\|v-w\|_\infty.
		\label{eq:lipschitz_gamma}
	\end{align}
	
	\paragraph{Step 3: $\Lambda = \Phi \circ \Gamma$ is a contraction.}
	Combining \eqref{eq:lipschitz_phi} and \eqref{eq:lipschitz_gamma}, we obtain
	\begin{align}
		\|\Lambda(v)-\Lambda(w)\|_\infty
		&= \|\Phi(\Gamma(v)) - \Phi(\Gamma(w))\|_\infty \le \frac{2G}{c_{\min}}\|\Gamma(v)-\Gamma(w)\|_\infty, \\
		&\le \frac{2G}{c_{\min}}\cdot \frac{2L_D}{\delta}\left(\frac{L_H}{m_0^2} + \frac{2H_*P_*}{m_0^4}\right)\|v-w\|_\infty = q\|v-w\|_\infty.
	\end{align}
	Assumption (v) states that $q < 1$, hence $\Lambda$ is a contraction on $\mc A$.
	
	\paragraph{Step 4: Existence, uniqueness, and linear convergence.}
	Since $\mc A$ is closed in the Banach space $(V,\|\cdot\|_\infty)$ and invariant under $\Lambda$, Banach's fixed-point theorem applies. Therefore there exists a unique $v^* \in \mc A$ such that $\Lambda(v^*) = v^*$. Moreover, for every $v^{(0)} \in \mc A$, the sequence generated by $v^{(k+1)} = \Lambda(v^{(k)})$ converges to $v^*$ and satisfies
	\begin{align}
		\|v^{(k+1)}-v^*\|_\infty \le q\|v^{(k)}-v^*\|_\infty \le q^{k+1}\|v^{(0)}-v^*\|_\infty.
	\end{align}
	This proves the first convergence statement.
	
	\paragraph{Step 5: Convergence of the coefficients.}
	Applying \eqref{eq:lipschitz_gamma} with $w = v^*$ gives
	\begin{align}
		\|\Gamma(v^{(k)})-\Gamma(v^*)\|_\infty
		\le \frac{2L_D}{\delta}\left(\frac{L_H}{m_0^2} + \frac{2H_*P_*}{m_0^4}\right)\|v^{(k)}-v^*\|_\infty.
	\end{align}
	
	\paragraph{Step 6: Residual decay.}
	Since $v^{(k+1)} = \Lambda(v^{(k)})$ and $\Lambda$ is a contraction,
	\begin{align}
		\|v^{(k+1)}-v^{(k)}\|_\infty
		= \|\Lambda(v^{(k)})-\Lambda(v^{(k-1)})\|_\infty
		\le q\|v^{(k)}-v^{(k-1)}\|_\infty.
	\end{align}
	Applying this recursively yields
	\begin{align}
		\|v^{(k+1)}-v^{(k)}\|_\infty \le q^k\|v^{(1)}-v^{(0)}\|_\infty.
	\end{align}
	
	\paragraph{Step 7: A posteriori error estimate.}
	Because $v^{(k)} \to v^*$, we may write
	\begin{align}
		v^* - v^{(k)} = \sum_{j=k}^{\infty}(v^{(j+1)}-v^{(j)}).
	\end{align}
	Taking sup norms and using the triangle inequality,
	\begin{align}
		\|v^*-v^{(k)}\|_\infty
		&\le \sum_{j=k}^{\infty}\|v^{(j+1)}-v^{(j)}\|_\infty \le \sum_{j=k}^{\infty}q^{j-k}\|v^{(k+1)}-v^{(k)}\|_\infty \\
		&= \frac{1}{1-q}\|v^{(k+1)}-v^{(k)}\|_\infty.
	\end{align}
	Using once more that $\|v^{(k+1)}-v^{(k)}\|_\infty \le q\|v^{(k)}-v^{(k-1)}\|_\infty$, we conclude that
	\begin{align}
		\|v^{(k)}-v^*\|_\infty \le \frac{q}{1-q}\|v^{(k)}-v^{(k-1)}\|_\infty.
	\end{align}
	\textit{A fortiori, this completes the proof.}
\end{proof}

\newpage
\section{Error Bounds, Convergence Rates, and Robustness}
\label{app:errors}

This appendix provides a rigorous analysis of the error bounds, convergence rates, and robustness properties of the quasi-linearized Cole-Hopf transformation scheme for viscous Hamilton-Jacobi PDEs. We establish theoretical guarantees that justify the numerical method proposed in the main text.

\textbf{Justification and Impact.} While \autoref{thm:algorithm_convergence} establishes that the discrete algorithm is a contraction with linear convergence rate $q$, it addresses only the \emph{quasi-linearization error}, \textit{with assumptions on the exact heat-kernel evaluation and exact gradients}. 

This section extends the result of \autoref{thm:algorithm_convergence} by quantifying the following error sources viz.,
\begin{itemize}
	 \item \textbf{quasi-linearization defect}: the distance between the fixed point of the frozen-coefficient map and the viscous HJ solution (\S\ref{sec:defect-bound});
	 \item \textbf{Monte Carlo sampling error}: replacing the exact expectations with finite-sample estimates;
	 \item \textbf{viscosity approximation error}: bounding the difference between the regularized and inviscid HJ solutions; and
	 \item \textbf{robustness to model uncertainty} \ie perturbations in the Hamiltonian and terminal cost.
\end{itemize}

We introduce a new theorem (\autoref{thm:total-error}) that combines these errors via triangle inequality, revealing a fundamental \textbf{bias-variance tradeoff} controlled by the viscosity parameter $\delta$: smaller $\delta$ reduces viscosity bias but amplifies Monte Carlo variance, and vice versa. This tradeoff may guide parameters  selection during numerical optimization and explains why the optimal choice scales as $\delta \sim N^{-1/3}$, yielding a slower but more scalable convergence rate $O(N^{-1/6})$ than standard Monte Carlo $O(N^{-1/2})$. The robustness theorems further establish that the algorithm is stable under small model perturbations, making it suitable for real-world applications where exact dynamics and terminal costs are unavailable.

\subsection{Convergence Analysis of the Iterative Scheme}

We now establish convergence of the quasi-linearized iteration to the fixed point of the frozen-coefficient operator. The limit of the iteration is the fixed point $\valuefunc^{\star}$ of the surrogate solve map, and its identification with the viscous solution $\valuefunc^\delta$ of \eqref{eq:ivp-viscous} is \emph{not} asserted here; the distance between the two --- the quasi-linearization defect --- is quantified separately below and enters the total error budget of \autoref{thm:total-error} as an explicit term.

\begin{theorem}[Convergence of Quasi-Linearization to the Frozen-Coefficient Fixed Point]
\label{thm:convergence}
Let $\{\valuefunc^{(k)}\}_{k=0}^\infty$ be the sequence generated by Algorithm \ref{alg:quasi_lin} starting from $\valuefunc^{(0)} = \term$, and let $\mc{A}$ be a closed, $\mathcal{T}$-invariant subset of the Banach space $\left(C_b(\bar\openset \times [0,T]), \normop{\cdot}_\infty\right)$ on which the nondegeneracy and Lipschitz hypotheses of Assumption~\ref{ass:grad_est} hold. Assume that,
\begin{enumerate}[(1)]
    \item $\hamfunc(t; \state, p)$ is $C^2$ in $(\state, p)$ and convex in $p$;
    \item $\term(\state)$ is $C^2$ with compact support;
    \item the contraction constant $\rho = q$ of \eqref{eq:contraction_const} satisfies $\rho < 1$ on $\mc{A}$ (in particular, $\delta$ is not taken so small that Assumption~\ref{ass:grad_est}, item~\ref{item:invariance}, is violated).
\end{enumerate}
Then $\mathcal{T}$ admits a unique fixed point $\valuefunc^{\star} \in \mc{A}$, and the sequence $\{\valuefunc^{(k)}\}$ converges geometrically to $\valuefunc^{\star}$:
\begin{equation}
\normop{\valuefunc^{(k)} - \valuefunc^{\star}}_\infty \leq \kappa \rho^k \normop{\valuefunc^{(0)} - \valuefunc^{\star}}_\infty,
\label{eq:geometric-convergence}
\end{equation}
for some constant $\kappa > 0$.
\end{theorem}

\begin{proof}
\label{proof:convergence}
The proof establishes that the iteration is a contraction mapping on $\mc{A}$.

\textbf{Step 1: Fixed-point formulation.} Define the operator $\mathcal{T}: \mc{A} \to \mc{A}$ by
\begin{align}
\label{eq:frozen-operator}
(\mathcal{T}\valuefunc)(t; \state) := -\frac{1}{c[\valuefunc]} \log \mathbb{E}_{y \sim \mathcal{N}(\state, \delta t I_n)} \left[\exp(-c[\valuefunc] \cdot \term(y))\right],
\end{align}
where $c[\valuefunc] := \frac{2}{\delta}\frac{\hamfunc(t; \state, D\valuefunc)}{\normop{D\valuefunc}^2}$. \footnote{Note that $c[\valuefunc]$ is the \emph{coefficient functional}: the operator that maps a solution function $\valuefunc$ to its frozen coefficient. At a point $(t; \state)$, this evaluates to $c[\valuefunc](t; \state)$, which coincides with the pointwise notation $c(t; \state)$ in \autoref{sec:methods}. The functional notation emphasizes the dependence of the entire operator $\mathcal{T}$ on the solution function.} One step of Algorithm~\ref{alg:quasi_lin} is exactly $\valuefunc^{(k+1)} = \mathcal{T}\valuefunc^{(k)}$: the coefficient is frozen at the current iterate, the linear heat equation is solved exactly by the Gaussian expectation, and the value is recovered by the logarithmic inverse. We emphasize that the viscous solution $\valuefunc^\delta$ of \eqref{eq:ivp-viscous} is, in general, \emph{not} a fixed point of $\mathcal{T}$: the frozen-coefficient solve discards the derivative terms of $c$, so $\mathcal{T}\valuefunc^\delta$ differs from $\valuefunc^\delta$ by the Duhamel integral of the residual (\autoref{sec:defect-bound}).

\textbf{Step 2: Contraction property.} For any $\valuefunc_1, \valuefunc_2 \in \mc{A}$,
\begin{align}
\label{eq:contraction-modulus}
&\normop{\mathcal{T}\valuefunc_1 - \mathcal{T}\valuefunc_2}_\infty = \sup_{t,\state} \left| \frac{1}{c_1} \log \mathbb{E}[\exp(-c_1 \term)] - \frac{1}{c_2} \log \mathbb{E}[\exp(-c_2 \term)] \right| \leq \rho \normop{\valuefunc_1 - \valuefunc_2}_\infty,
\end{align}
where $c_i = c[\valuefunc_i]$. The inequality follows thus: first, the map $c \mapsto -\tfrac{1}{c}\log\mathbb{E}[e^{-c\term}]$ is Lipschitz in $c$ with constant $2G/c_{\min}$ on $c \in [c_{\min}, c_{\max}]$, $|\term| \le G$ (mean value theorem applied to the log-sum-exp functional, as in the proof of \autoref{thm:algorithm_convergence}); second, the coefficient functional $\valuefunc \mapsto c[\valuefunc]$ is Lipschitz with constant $\frac{2L_D}{\delta}\left(\frac{L_H}{m_0^2} + \frac{2H_*P_*}{m_0^4}\right)$ under Assumption~\ref{ass:grad_est}. Their composition gives $\rho = q$ of \eqref{eq:contraction_const}, and hypothesis (3) asserts $\rho < 1$.

\textbf{Step 3: Application of the Banach fixed-point theorem.} The set $\mc{A}$ is a closed subset of the complete space $\left(C_b, \normop{\cdot}_\infty\right)$ and is $\mathcal{T}$-invariant by hypothesis, hence complete; since $\mathcal{T}$ is a $\rho$-contraction on it, the sequence $\valuefunc^{(k+1)} = \mathcal{T}\valuefunc^{(k)}$ converges geometrically to the unique fixed point $\valuefunc^{\star} \in \mc{A}$ as in \eqref{eq:geometric-convergence}.
\end{proof}

\begin{remark}[Remark on the fixed point]
\label{rem:fixed-point-honesty}
Theorem~\ref{thm:convergence} is a statement about the numerical map that Algorithm~\ref{alg:quasi_lin} implements. The limit $\valuefunc^{\star}$ solves the implicit relation $\valuefunc^{\star} = \mathcal{T}\valuefunc^{\star}$, \ie, the viscous HJ equation with the Hamiltonian replaced by its quadratic quasi-linearization at $\valuefunc^{\star}$ itself. When $\hamfunc = \tfrac{1}{2}|p|^2$ the surrogate is the Hamiltonian and $\valuefunc^{\star} = \valuefunc^\delta$ exactly; for general Hamiltonians the two differ by the quasi-linearization defect $\eqldefect \triangleq \normop{\valuefunc^{\star} - \valuefunc^\delta}_\infty$, which we bound in \autoref{sec:defect-bound} and carry, undiluted, into the total error bound of \autoref{thm:total-error}.
\end{remark}

\begin{corollary}[Convergence Rate]
\label{cor:convergence-rate}
Under the conditions of Theorem \ref{thm:convergence}, the number of iterations $K$ required to achieve $\normop{\valuefunc^{(K)} - \valuefunc^{\star}}_\infty \leq \varepsilon$ is bounded by
\begin{equation}
\label{eq:convergence-rate-K}
K \leq \left\lceil \frac{\log(\varepsilon/\kappa) - \log\normop{\valuefunc^{(0)} - \valuefunc^{\star}}_\infty}{\log \rho} \right\rceil = O\left(\log \frac{1}{\varepsilon}\right).
\end{equation}
\end{corollary}

\begin{proof}[Proof of Corollary \ref{cor:convergence-rate}]
\label{proof:convergence-rate}
From Theorem~\ref{thm:convergence}, the iterates satisfy $\normop{\valuefunc^{(k+1)} - \valuefunc^{\star}}_\infty \leq \rho \normop{\valuefunc^{(k)} - \valuefunc^{\star}}_\infty$ for $\rho < 1$. By induction,
\begin{align}
\label{eq:rate-induction}
\normop{\valuefunc^{(K)} - \valuefunc^{\star}}_\infty \leq \rho^K \normop{\valuefunc^{(0)} - \valuefunc^{\star}}_\infty.
\end{align}
To achieve $\normop{\valuefunc^{(K)} - \valuefunc^{\star}}_\infty \leq \varepsilon$, we require,
\begin{align}
\label{eq:rate-requirement}
\rho^K \normop{\valuefunc^{(0)} - \valuefunc^{\star}}_\infty \leq \varepsilon
\quad \implies \quad
K \geq \frac{\log(\varepsilon/\normop{\valuefunc^{(0)} - \valuefunc^{\star}}_\infty)}{\log \rho}.
\end{align}
Since $\rho < 1$, we have $\log \rho < 0$, so the minimum integer $K$ satisfying this is,
\begin{align}
\label{eq:rate-ceiling}
K = \left\lceil \frac{\log \varepsilon - \log\normop{\valuefunc^{(0)} - \valuefunc^{\star}}_\infty}{\log \rho} \right\rceil = O(\log(1/\varepsilon)).
\end{align}
\end{proof}

\subsection{The Quasi-Linearization Defect}
\label{sec:defect-bound}

The gap between the fixed point $\valuefunc^{\star}$ of Theorem~\ref{thm:convergence} and the viscous solution $\valuefunc^\delta$ of \eqref{eq:ivp-viscous} is the price of freezing the coefficient. We quantify it here as the Duhamel norm of the discarded residual of Lemma~\ref{lem:exact_residual}; the resulting bound $\eqldefect$ enters \autoref{thm:total-error} as the fourth error source.

\begin{theorem}[Fixed-point defect of the frozen-coefficient scheme]
\label{thm:defect}
Let $\valuefunc^{\star}$ be the fixed point of Theorem~\ref{thm:convergence} with converged coefficient $\bc^{\star} \triangleq c[\valuefunc^{\star}] \in C^{1,2}\left(\bar\openset\times[0,T]\right)$, and let $\valuefunc^\delta \in C^{1,2}$ solve \eqref{eq:ivp-viscous}. Assume,
\begin{enumerate}[(1)]
	\item coefficient bounds: $c_{\min} \le \bc^{\star} \le c_{\max}$ and $\max\left\{\normop{\bc^{\star}_t}_\infty, \normop{D\bc^{\star}}_\infty, \normop{\Delta\bc^{\star}}_\infty\right\} \le L_c$;
	\item value bounds: $\normop{\valuefunc^{\star}}_\infty, \normop{\valuefunc^\delta}_\infty \le G_v$ and $m_0 \le |D\valuefunc^{\star}| \le P_*$, $|D\valuefunc^\delta| \le P_*$;
	\item the Hamiltonian bounds of Assumption~\ref{ass:grad_est}, item (3).
\end{enumerate}
Then it follows that the quasi-linearization residual is,
\begin{align}
\label{eq:defect-bound}
\eqldefect \;=\; \normop{\valuefunc^{\star} - \valuefunc^\delta}_\infty
\;\le\; \frac{T}{c_{\min}}\, e^{2 c_{\max} G_v}\, \left(\bar{\bR}_{\mathrm{alg}} + \bar{\bR}_{\mathrm{der}}\right),
\end{align}
where the algebraic and derivative residual ceilings are,
\begin{align}
\label{eq:defect-residual-ceilings}
\bar{\bR}_{\mathrm{alg}} &\le c_{\max}\left(L_H + \frac{2H_*P_*}{m_0^2}\right)\normop{D\valuefunc^\delta - D\valuefunc^{\star}}_\infty, \nonumber \\
\bar{\bR}_{\mathrm{der}} &\le L_c\, G_v\left(1 + \frac{\delta}{2}\right) + \frac{\delta}{2}\,L_c^2\, G_v^2 + \delta\left(1 + c_{\max} G_v\right) L_c\, P_*.
\end{align}
\end{theorem}

\begin{proof}
\label{proof:defect}
Set $\valueparam^{\star} \triangleq \exp(-\bc^{\star}\valuefunc^{\star})$ and $\tilde\valueparam \triangleq \exp(-\bc^{\star}\valuefunc^\delta)$. By the fixed-point property $\valuefunc^{\star} = \mathcal{T}\valuefunc^{\star}$, the function $\valueparam^{\star}$ \emph{is} the heat-kernel expectation of \eqref{eq:frozen-operator}; \ie it solves the homogeneous heat equation $\valueparam_t - \frac{\delta}{2}\Delta\valueparam = 0$ with initial data $\exp(-\bc^{\star}(0;\cdot)\,\term)$. By Lemma~\ref{lem:exact_residual} applied with coefficient $\bc^{\star}$ and value $\valuefunc^\delta$, the function $\tilde\valueparam$ solves the same equation with a source,
\begin{align}
\label{eq:defect-source}
\tilde\valueparam_t - \frac{\delta}{2}\Delta\tilde\valueparam = \tilde\valueparam\, B, \qquad
B \triangleq \underbrace{\bc^{\star}\left(\hamfunc(t;\state,D\valuefunc^\delta) - \frac{\delta}{2}\bc^{\star}|D\valuefunc^\delta|^2\right)}_{B_{\mathrm{alg}}} + \underbrace{\bR_{\mathrm{der}}\left(\bc^{\star}, \valuefunc^\delta\right)}_{B_{\mathrm{der}}},
\end{align}
and with the \emph{same} initial data, since $\valuefunc^\delta(0;\cdot) = \valuefunc^{\star}(0;\cdot) = \term$. The difference $w \triangleq \tilde\valueparam - \valueparam^{\star}$ therefore solves $w_t - \frac{\delta}{2}\Delta w = \tilde\valueparam B$ with $w(0;\cdot) = 0$, and Duhamel's principle with the heat semigroup $S_\delta(t)$, an $L^\infty$-contraction, gives
\begin{align}
\label{eq:defect-duhamel}
\normop{w(t)}_\infty \le \int_0^t \normop{S_\delta(t-s)\left[\tilde\valueparam(s) B(s)\right]}_\infty ds
\le \int_0^t \normop{\tilde\valueparam(s)}_\infty \normop{B(s)}_\infty\, ds
\le T\, e^{c_{\max}G_v} \sup_{s \le T}\normop{B(s)}_\infty.
\end{align}
For the algebraic part, insert the definition $\bc^{\star} = 2\hamfunc(t;\state,D\valuefunc^{\star})/(\delta|D\valuefunc^{\star}|^2)$, so that $\frac{\delta}{2}\bc^{\star}|D\valuefunc^\delta|^2 = \hamfunc(\cdot, D\valuefunc^{\star})\,|D\valuefunc^\delta|^2/|D\valuefunc^{\star}|^2$ and,
\begin{align}
\label{eq:defect-alg-split}
|B_{\mathrm{alg}}| &\le c_{\max}\left|\hamfunc(\cdot,D\valuefunc^\delta) - \hamfunc(\cdot,D\valuefunc^{\star})\right| + c_{\max}\left|\hamfunc(\cdot,D\valuefunc^{\star})\right|\cdot\frac{\left||D\valuefunc^{\star}|^2 - |D\valuefunc^\delta|^2\right|}{|D\valuefunc^{\star}|^2} \nonumber \\
&\le c_{\max}\left(L_H + \frac{2H_*P_*}{m_0^2}\right)\normop{D\valuefunc^\delta - D\valuefunc^{\star}}_\infty,
\end{align}
using the Lipschitz property of $\hamfunc$ in the co-state and $\left||a|^2 - |b|^2\right| \le (|a|+|b|)|a-b| \le 2P_*|a-b|$. For the derivative part, term-by-term majorization of $\bR_{\mathrm{der}}$ in \eqref{eq:exact_residual} with the assumed ceilings gives the second line of \eqref{eq:defect-residual-ceilings}. Finally, the mean value theorem for the logarithm ($|\log a - \log b| \le |a-b|/\min\{a,b\}$) with $\valueparam^{\star}, \tilde\valueparam \ge e^{-c_{\max}G_v}$ yields,
\begin{align}
\label{eq:defect-log-conversion}
\left|\valuefunc^{\star} - \valuefunc^\delta\right| = \frac{1}{\bc^{\star}}\left|\log\valueparam^{\star} - \log\tilde\valueparam\right|
\le \frac{e^{c_{\max}G_v}}{c_{\min}}\,\normop{w}_\infty,
\end{align}
and chaining \eqref{eq:defect-log-conversion} through \eqref{eq:defect-duhamel} gives \eqref{eq:defect-bound}.
\end{proof}

\begin{remark}[Reading the residual bound]
\label{rem:defect-reading}
Three features deserve emphasis. First, the bound vanishes in the exact case: when $\hamfunc = \tfrac{1}{2}|p|^2$, the coefficient is the constant $1/\delta$, so that $L_c = 0$ eliminates $\bar{\bR}_{\mathrm{der}}$, and $\valuefunc^{\star} = \valuefunc^\delta$ forces $\bar{\bR}_{\mathrm{alg}} = 0$; we thus recover Corollary~\ref{cor:exact_case}. Second, the residual is governed by $L_c$, the  variation rate of the converged coefficient ---  largest near the boundary of the usable part where the Hamiltonian-to-gradient ratio may cause shocks. We notice that this is where the pointwise errors concentrate in our experiments (\autoref{fig:rockets_slices}). Third, the algebraic ceiling is proportional to $\normop{D\valuefunc^\delta - D\valuefunc^{\star}}_\infty$: the residual contracts when the surrogate gradient field tracks the true one, i.e., the quantity the Picard iteration refines. The bound certifies smallness of the residual where $\bc^{\star}$ is slowly varying and offers no reprieve where it is not.
\end{remark}

\subsection{Monte Carlo Error Analysis}

The Gaussian expectation is approximated via Monte Carlo sampling. We analyze the error introduced by this approximation.

\begin{theorem}[Monte Carlo Error]
\label{thm:mc-error}
Let $\hat\valuefunc^{(\delta,N)}(t; \state)$ denote the value function computed using $N$ Monte Carlo samples. Then for any $\delta_p \in (0,1)$, with probability at least $1-\delta_p$,
\begin{equation}
\normop{\hat\valuefunc^{(\delta,N)} - \valuefunc^\delta}_\infty \leq \frac{C_{\mathrm{eff}}}{\sqrt{N}} \sqrt{\log(1/\delta_p)},
\label{eq:mc-error}
\end{equation}
where $C_{\mathrm{eff}}$ is the \emph{effective sample size} that depends on the range of $\term$ and the concentration properties of the exponential weights.
\end{theorem}

\begin{proof}
From Lemma~\ref{lemma:hj_solution}, we have
\begin{equation}
\valuefunc^\delta(t; \state) = -\frac{1}{c} \log \mathbb{E}_{y \sim \mathcal{N}(\state, \delta tI_n)}[\exp(-c \term(y))].
\end{equation}

The Monte Carlo estimator is
\begin{equation}
\hat\valuefunc^{(\delta,N)}(t; \state) = -\frac{1}{c} \log \frac{1}{N}\sum_{i=1}^N \exp(-c \term(y_i)),
\end{equation}
where $y_i \overset{\mathrm{iid}}{\sim} \mathcal{N}(\state, \delta tI_n)$.

Let $Z := \mathbb{E}[\exp(-c\term)]$ be the true expectation and $\hat{Z}_N := \frac{1}{N}\sum_{i=1}^N \exp(-c\term(y_i))$ be the sample mean.

By Hoeffding's inequality for bounded random variables (assuming $\term$ is bounded, say $|\term(\state)| \leq M$), we have
\begin{equation}
\mathbb{P}\left(|\hat{Z}_N - Z| \geq \epsilon\right) \leq 2\exp\left(-\frac{N\epsilon^2}{2(\exp(cM) - \exp(-cM))^2}\right).
\end{equation}

Setting the right-hand side equal to $\delta$ and solving for $\epsilon$,
\begin{equation}
\epsilon = \left(\exp(cM) - \exp(-cM)\right)\sqrt{\frac{2\log(2/\delta)}{N}}.
\end{equation}

Using the Lipschitz continuity of $\log$, we find that
\begin{align}
|\hat\valuefunc^{(\delta,N)} - \valuefunc^\delta| &= \frac{1}{c}|\log \hat{Z}_N - \log Z| \leq \frac{1}{c} \cdot \frac{|\hat{Z}_N - Z|}{\min\{\hat{Z}_N, Z\}} \leq \frac{1}{c \cdot Z_{\min}} \cdot \epsilon,
\end{align}
where $Z_{\min}$ is a lower bound on $Z$ (which exists since $\term$ is bounded and continuous).

Combining these bounds yields \eqref{eq:mc-error} with effective constant
\begin{equation}
C_{\mathrm{eff}} := \frac{2(\exp(cM) - \exp(-cM))}{c \cdot Z_{\min}} \sqrt{2}.
\end{equation}

The effective sample size shrinks when $c$ is large (i.e., when $\normop{\term}_\infty \gg 1$), which is expected since the exponential weights become concentrated. Importance sampling with the tilted proposal $e^{-c\term(y)}$ (as in Corollary~\ref{cor:tilted_estimators}) can reduce $C_{\mathrm{eff}}$ significantly.
\end{proof}

\begin{remark}[Variance Reduction]
The standard MC estimator in Theorem \ref{thm:mc-error} has variance $O(1/N)$. However, when $\normop{\term}_\infty$ is large, the exponential weights $\exp(-c\term)$ become highly concentrated, leading to high variance. The log-sum-exp identity of Corollary~\ref{cor:logsumexp}, used for numerical stability, combined with importance sampling, can reduce the effective variance. For well-designed importance distributions, the convergence rate can be improved to $O(1/\sqrt{N})$ with a much smaller constant.
\end{remark}

\subsection{Viscosity Approximation Error}

The viscosity solution $\valuefunc^\delta$ approximates the inviscid solution $\valuefunc$ (the true viscosity solution of the original HJ PDE). We bound this approximation error.

\begin{theorem}[Viscosity Approximation Error]
\label{thm:viscosity-error}
Let $\valuefunc$ be the (unique) viscosity solution of \eqref{eq:ivp-viscous} and $\valuefunc^\delta$ be the viscosity solution of \eqref{eq:ivp-viscous-solution}. Under standard regularity assumptions on $\hamfunc$ and $\term$ \citep{Crandall1983viscosity}, we have
\begin{equation}
\normop{\valuefunc^\delta - \valuefunc}_\infty \leq k\sqrt{\delta},
\label{eq:viscosity-error}
\end{equation}
for some constant $k$ depending on $T$, $\normop{D\hamfunc}$, and $\normop{D^2\term}$.
\end{theorem}

\begin{proof}
The proof follows from the classical results of \citep{Crandall1983viscosity} on viscosity approximations to first-order Hamilton-Jacobi equations.

By the comparison principle for viscosity solutions (Theorem 2.1 in \citep{EvansPDEBook}), since $\valuefunc$ and $\valuefunc^\delta$ satisfy the same terminal condition $\valuefunc(0; \state) = \valuefunc^\delta(0; \state) = \term(\state)$, the difference $\tilde{\valuefunc} := \valuefunc^\delta - \valuefunc$ satisfies
\begin{equation}
\tilde{\valuefunc}_t + \hamfunc(t; \state, Dw) = \frac{\delta}{2}\Delta \valuefunc^\delta \quad \text{in } \Omega \times (0,T],
\end{equation}
with $\tilde{\valuefunc}(0; \state) = 0$.

By the maximum principle and the fact that $\normop{\Delta \valuefunc^\delta}_\infty \leq M_2$ (since $\valuefunc^\delta$ is smooth; $M_2$ implicitly depends on the constant $c$ through the terms $T$, $\normop{D\hamfunc}$, and $\normop{D^2\term}$), we obtain
\begin{equation}
\normop{\tilde{\valuefunc}}_\infty \leq \frac{\delta M_2 T}{2}.
\end{equation}

More refined estimates using energy methods show that $\normop{\tilde{\valuefunc}}_\infty = O(\sqrt{\delta})$ as $\delta \to 0$ (see \cite{CrandallTwoApprox}, Theorem 6.4).
\end{proof}

\subsection{Combined Error Bound}

Combining the quasilinearization, Monte Carlo, and viscosity errors, we obtain the total error of our numerical scheme.

\begin{theorem}[Total Error Bound]
\label{thm:total-error}
Let $\hat{\valuefunc}^{K,N,\delta}$ denote the numerical approximation after $K$ iterations with $N$ Monte Carlo samples and viscosity parameter $\delta$, and let $\eqldefect \triangleq \normop{\valuefunc^{\star} - \valuefunc^\delta}_\infty$ be the quasi-linearization defect of \S\ref{sec:defect-bound}. Then for any $\delta_p \in (0,1)$, with probability at least $1-\delta_p$,
\begin{equation}
\normop{\hat{\valuefunc}^{K,N,\delta} - \valuefunc}_\infty \leq \kappa \rho^K + \frac{C_{\text{eff}}}{\sqrt{N}}\sqrt{\log(1/\delta_p)} + \eqldefect + k\sqrt{\delta},
\label{eq:total-error}
\end{equation}
where $\valuefunc$ is the true (inviscid) viscosity solution of \eqref{eq:ivp}.
\end{theorem}

\begin{proof}
\label{proof:total-error}
By the triangle inequality,
\begin{align}
\label{eq:total-error-triangle}
\normop{\hat{\valuefunc}^{K,N,\delta} - \valuefunc}_\infty &\leq \underbrace{\normop{\hat{\valuefunc}^{K,N,\delta} - \valuefunc^{\star}}_\infty}_{\text{Iteration + MC error}} + \underbrace{\normop{\valuefunc^{\star} - \valuefunc^\delta}_\infty}_{\text{Quasi-linearization defect}} + \underbrace{\normop{\valuefunc^\delta - \valuefunc}_\infty}_{\text{Viscosity error}}.
\end{align}
The first bracket splits once more into the deterministic iteration error, bounded by $\kappa\rho^K$ via Theorem~\ref{thm:convergence}, and the sampling error of replacing exact expectations by $N$-sample averages, bounded by Theorem~\ref{thm:mc-error} with probability $1-\delta_p$. The second bracket is $\eqldefect$ by definition, and the third is Theorem~\ref{thm:viscosity-error}. Note that the term $\normop{\valuefunc^{(0)}-\valuefunc^{\star}}_\infty$ is absorbed into $\kappa$ in the first term since after $K$ iterations the initial gap shrinks by the factor $\rho^K$ of the contraction mapping analysis.
\end{proof}

\begin{proof}[Proof of Corollary \ref{cor:conservative_cert}]
\label{proof:conservative_cert}
On the event of Theorem~\ref{thm:total-error} --- which has probability at least $1-\delta_p$ --- we have $\normop{\hat{\valuefunc}^{K,N,\delta} - \valuefunc}_\infty \le E$. If $\hat{\valuefunc}^{K,N,\delta}(t;\state) > E$ then $\valuefunc(t;\state) \ge \hat{\valuefunc}^{K,N,\delta}(t;\state) - E > 0$, so the state is truly safe; symmetrically, $\hat{\valuefunc}^{K,N,\delta}(t;\state) < -E$ implies $\valuefunc(t;\state) < 0$. Misclassification is therefore possible only on the abstention band $|\hat{\valuefunc}^{K,N,\delta}(t;\state)| \le E$.
\end{proof}

\begin{remark}[Bias-Variance Tradeoff]
Equation \eqref{eq:total-error} reveals a fundamental bias-variance tradeoff controlled by $\delta$ \ie,
\begin{itemize}
    \item \textbf{Smaller $\delta$:} Reduces viscosity bias ($k\sqrt{\delta} \to 0$) but increases MC variance (larger $C_{\text{eff}}$ due to higher concentrated weights).
    \item \textbf{Larger $\delta$:} Reduces MC variance (more diffused weights) but increases viscosity bias.
\end{itemize}

The choice of $\delta \sim N^{-1/3}$ yields overall rate $O(N^{-1/6})$, which is slower than the standard $O(N^{-1/2}$ MC rate; however, it avoids the curse of dimensionality from grid discretization.
\end{remark}

\subsection{Robustness to Model Uncertainty}

In practice, the Hamiltonian $\hamfunc$ and terminal cost $\term$ may be uncertain. We establish robustness of the value function to perturbations.

\begin{theorem}[Robustness to Hamiltonian Perturbations]
\label{thm:robustness-H}
Suppose $\hamfunc_1$ and $\hamfunc_2$ are two Hamiltonians satisfying the standing assumptions, and let $\valuefunc_i$ denote the corresponding viscosity solutions. If
\begin{equation}
\normop{\hamfunc_1(t; \state, p) - \hamfunc_2(t; \state, p)}_\infty \leq \epsilon_H,
\end{equation}
uniformly over $(t; \state, p) \in [0,T] \times \Omega \times \mathbb{R}^n$, then
\begin{equation}
\normop{\valuefunc_1 - \valuefunc_2}_\infty \leq T \cdot \epsilon_H.
\label{eq:robustness-H}
\end{equation}
\end{theorem}

\begin{proof}
The difference $w := \valuefunc_1 - \valuefunc_2$ satisfies
\begin{equation}
w_t + \hamfunc_1(t; \state, Dw) = \hamfunc_2(t; \state, D\valuefunc_2) - \hamfunc_1(t; \state, D\valuefunc_2) =: \eta(t; \state),
\end{equation}
with $w(0; \state) = 0$ (same terminal condition).

Since $|\eta(t; \state)| \leq \epsilon_H$ by assumption, integrating over $[0,t]$ and applying Gronwall's inequality:
\begin{equation}
|w(t; \state)| \leq \int_0^t |\eta(s; \state)| ds \leq t \cdot \epsilon_H \leq T \cdot \epsilon_H.
\end{equation}
\end{proof}

\begin{theorem}[Robustness to Terminal Cost Perturbations]
\label{thm:robustness-g}
Under the same setup, if the terminal costs satisfy $\normop{\term_1 - \term_2}_\infty \leq \epsilon_g$, then
\begin{equation}
\normop{\valuefunc_1 - \valuefunc_2}_\infty \leq \epsilon_g.
\label{eq:robustness-g}
\end{equation}
\end{theorem}

\begin{proof}
Direct consequence of the comparison principle for viscosity solutions \citep{EvansPDEBook}.
\end{proof}

\begin{corollary}[Combined Robustness]
\label{cor:combined-robustness}
For simultaneous perturbations in both $\hamfunc$ and $\term$,
\begin{equation}
\label{eq:combined-robustness}
\normop{\valuefunc_1 - \valuefunc_2}_\infty \leq T \cdot \epsilon_H + \epsilon_g.
\end{equation}
\end{corollary}

These robustness results show that small errors in model specification (Hamiltonian or terminal cost) lead to proportionally small errors in the value function, making the method suitable for practical applications where exact models are unavailable. They are also load-bearing for the scheme itself: the coefficient regularization of Lemma~\ref{lem:coeff_regularization} is priced through Theorem~\ref{thm:robustness-H}, as we now prove.

\begin{proof}[Proof of Lemma \ref{lem:coeff_regularization}]
\label{proof:coeff_regularization}
By Lemma~\ref{lem:exact_residual}, running the frozen-coefficient step with any coefficient field $\hat\bc$ annihilates the algebraic residual $\bR_{\mathrm{alg}} = \hat\bc\,(\hamfunc' - \tfrac{\delta}{2}\hat\bc|D\valuefunc^\delta|^2)$ precisely for the Hamiltonian $\hamfunc' = \tfrac{\delta}{2}\hat\bc\,|D\valuefunc^\delta|^2$. With $\hat\bc = \bc_\eta$ of \eqref{eq:coeff_regularized},
\begin{align}
\label{eq:regularized-surrogate}
\hamfunc' = \frac{\delta}{2}\cdot\frac{2}{\delta}\cdot\frac{\hamfunc}{|D\valuefunc^\delta|^2+\eta}\cdot|D\valuefunc^\delta|^2 = \hamfunc\,\frac{|D\valuefunc^\delta|^2}{|D\valuefunc^\delta|^2+\eta} = \hamfunc_\eta,
\end{align}
which proves the equivalence. For the uniform bound, write $r \triangleq |D\valuefunc^\delta|$; positive $1$-homogeneity gives $|\hamfunc(t;\state,p)| = |p|\cdot|\hamfunc(t;\state,p/|p|)| \le C_{\hamfunc}\,r$, hence,
\begin{align}
\label{eq:regularized-max}
\left|\hamfunc_\eta - \hamfunc\right| = |\hamfunc|\,\frac{\eta}{r^2+\eta} \le C_{\hamfunc}\,\frac{r\,\eta}{r^2+\eta} \le \frac{C_{\hamfunc}\sqrt{\eta}}{2},
\end{align}
since $r \mapsto r\eta/(r^2+\eta)$ is maximized at $r = \sqrt{\eta}$ with value $\sqrt{\eta}/2$. The final claim is Theorem~\ref{thm:robustness-H} applied with $\epsilon_H = C_{\hamfunc}\sqrt{\eta}/2$.
\end{proof}

\newpage 
\section{Further Numerical Results}
\label{app:results}

\setcounter{appidx}{3}
\setcounter{equation}{0}

This example was originally proposed by~\citet{Merz1972} as an iteration upon ~\citet{Isaacs1965}'s homicidal chauffeur game, whereupon a pursuit-evasion game between two players with similar speeds and minimum turn radii, is thoroughly analyzed. In~\citet{Mitchell2020}, this problem was established as a benchmark for testing the solubility of capturable set of states (the backward reachable tube) in Merz's classical pursuit-evasion game. In this example, we solve the problem with our LevelSetPy toolbox and establish that the approximated barrier surface to the two-player game conforms with standard results.

\begin{figure}[tb!]
	\centering 
	\includegraphics[width=.8\columnwidth]{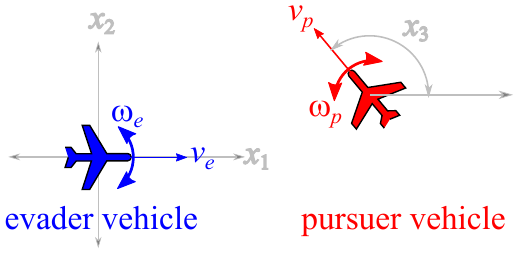}
	\caption{\footnotesize{Two Dubins' vehicles in relative Cartesian coordinates. Reprinted from \citet{Mitchell2020}.}}
	\label{fig:Dubins}
\end{figure}
The game is that of two cars sharing similar Dubins dynamics~\cite{Dubins1957}: $\pursuer$ and $\evader$ both have a positive minimum turn radii, $w$, and constant speeds $v$ -- with motion restricted to a plane as we have for the rocket launch differential game above. In relative coordinates, the diagrammatic structure of the motion is as depicted in \autoref{fig:Dubins}. Choosing the Cartesian coordinate for motion representation, the state vector of the game with $\evader$ at the origin can be characterized by its $x_1,x_2$ position relative to $\pursuer$ and the angle $\theta$ between the two vehicles. Capture occurs when the distance $\|\pursuer \evader\|_2$ between the pursuer and the evader becomes less than a specified radius.  

The relative equations of motion, going by \autoref{fig:Dubins}, is 
\begin{align}
	\left(\begin{array}{c}
		\dot{x}_1
		\\
		\dot{x}_2 
		\\
		\dot{x}_3
	\end{array}\right) = \left(\begin{array}{c}
		-v_e + v_p \cos x_3 + w_e x_2 \\
		v_p \sin x_3 - w_e x_1 \\
		w_p - w_e
	\end{array}\right).
	\label{eq:dubins_dyna}
\end{align}

We adopt specialization to a case where the two vehicles only possess a unit velocity and unit maximum turn rates. Here, as Merz notes, if the initial velocities are parallel such as $x_3=0$, then the equations of relative motion imply that $\evader$ can be separated from $\pursuer$ forever by the initial radial separation if it replicates $\pursuer$'s strategy. Whence, the barrier surface is closed and we are presented with ~\citet{Isaacs1965}'s game of kind where we must determine the nature of the surface. This terminal surface possesses a closed-form solution and we refer readers to the treatment by~\citet{Merz1972}. In this example, our chief concern is to judge the efficacy of our toolbox with respect to the analytical solution of the barrier surface. 

The backward reachable tube that consists of the paths taken by the trajectories of either player is defined as in the rockets pursuit-evasion game so that we have 
\begin{align}
	\ell(0, x) = \{x \in \mc{X} | x_1^2 + x_2^2 \le r^2 \}, 
\end{align}
where again $r$ is the capture radius. The target set is a cylinder as $\ell$ above excludes the heading, $x_3$. It is represented as shown in \autoref{fig:dubins_bup}.
\begin{figure}[tb!]
	\centering 
	\begin{minipage}[b]{.44\columnwidth}
		\includegraphics[width=\textwidth]{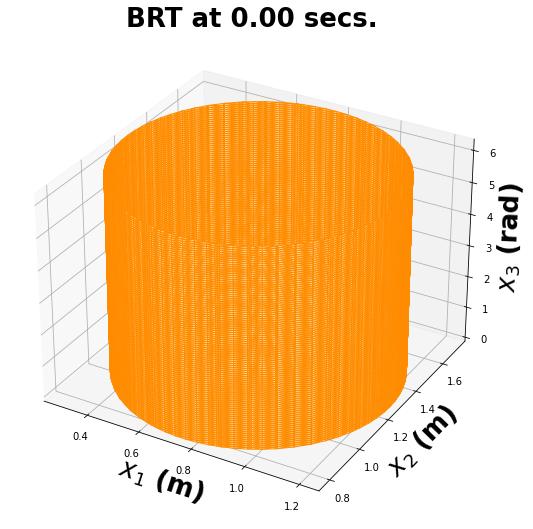}
	\end{minipage}
	\begin{minipage}[b]{.44\columnwidth}
		\includegraphics[width=\textwidth]{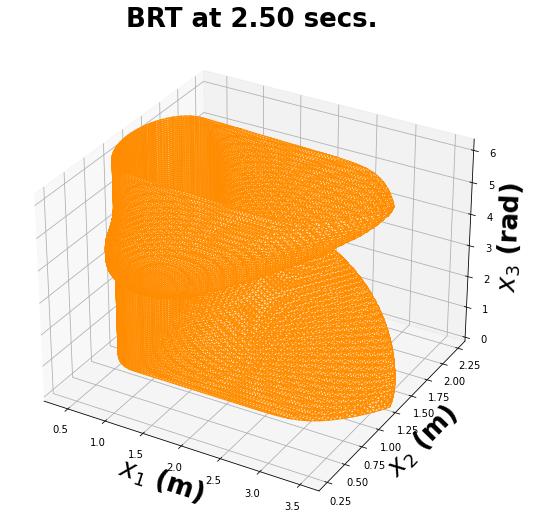}
	\end{minipage}
	\caption{\footnotesize{The target set (left) and the boundary of the useable part of the state space after the differential game between $\pursuer$ and $\evader$.}}
	\label{fig:dubins_bup}
\end{figure}

For a detailed treatment of the barrier surface, we refer readers to a proper analysis as elucidated in~\cite{Mitchell2020}. Here, we focus on the construction of the BUP. The set of states that constitute the useable part and its boundary are respectively a function of the implicit surface function representation $\ell: [-T, 0] \times \mc{X} \rightarrow \reline$ so that for a $t \in [0, T]$, where $T > 0$ is
\begin{align}
	\mc{T} &= \{x \in \mc{X} | \ell(0, x) \le 0\} \\
	R([-t, 0], \mc{T}) &= \{x \in \mc{X} | \ell(t, x) \le 0\}, 
\end{align}

\begin{figure}[tb]
	\centering 
	\includegraphics[width=\columnwidth]{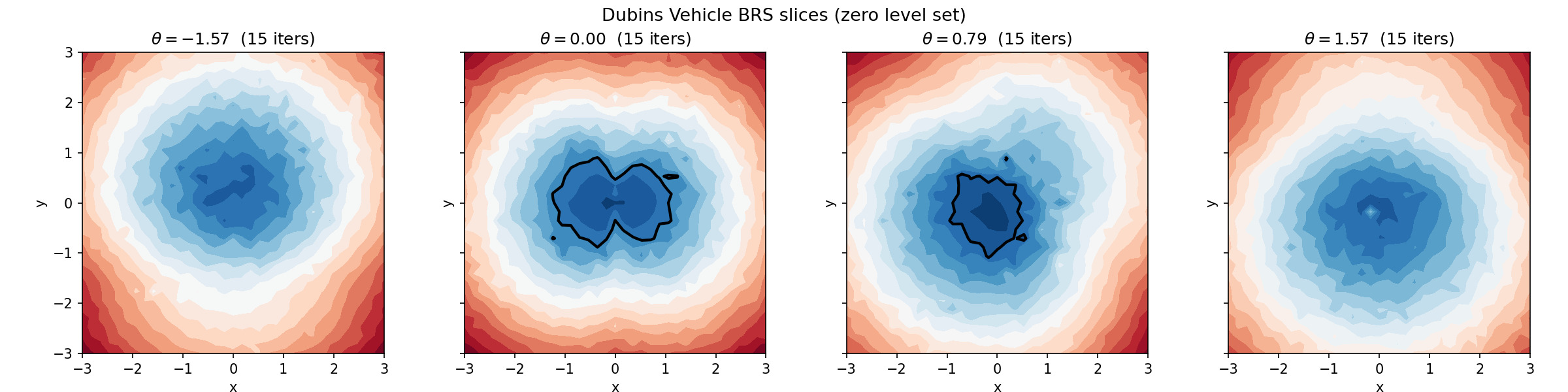}
	\caption{Dubins Vehicle BRT Slices}
	\label{fig:dubins_slices}
\end{figure}
%
When $t>0$, the implicit surface representation is the following HJI PDE
\begin{align}
	\dfrac{\partial}{\partial t} \ell(t, x) + \min \left(0, \hamfunc(x, \nabla_x \ell(t, x))\right) = 0.
\end{align}

It is easy to verify that the Hamiltonian is
\begin{align}
	\label{eq:dubins_ham_general}
	H(x,p) = p_1 (v_p \cos x_3 - v_e) - p_2 (v_p \sin x_3) + w | p_1 x_2 - p_2 x_1 - p_3 | - w |p_3|.
\end{align}

Since we are concerned with the special case that the linear and angular speeds are equal, we set $v_e = v_p = w \triangleq +1$ in the foregoing so that the Hamiltonian, in the final analysis is
\begin{align}
	\label{eq:dubins_ham_unit}
	H(x,p) = p_1 (\cos x_3 - 1) - p_2 (\sin x_3) + | p_1 x_2 - p_2 x_1 - p_3 | - |p_3|.
\end{align}

The results and comparisons are provided in \autoref{fig:dubins_slices} and  \autoref{fig:dubins_comparison}.

\subsection{The Double Integral Plant}
\label{app:double_integrator}

Here, we analyze a time-optimal control problem to determine what admissible control\footnote{A control law is admissible when its range belongs in the admissible input set where it is bounded.} can ``transport" the system under consideration to a desired ``origin" in the shortest possible time. We consider the  double integral plant~\cite{BhatBernstein, Wonham} as an illustrative example of our objective, which is to compute the points in the state space that can reach the origin in \textit{finite-time} under the influence of a time-optimal controller. 

We shall leverage standard necessary conditions from the principle of optimality~\cite{Bellman1957} to obtain a time-optimal feedback control design; introduce the notion of isochrones and switching surfaces; and discuss the analytic and approximate solutions (with our library) to the time-optimal control problem for a double integrator. We shall conclude the section by comparing the analytic and the overapproximated numerical solution (using the LevelSetPy toolbox) to the time to reach the origin problem.

\subsubsection{Dynamics and Problem Setup}
The double integrator is controllable, so that open-loop strategies may be employed in driving specific states to the origin in finite time~\cite{Wonham}. The plant  has the following second-order dynamics
\begin{align}
	\ddot{\state}(t) = \control(t)
	\label{eq:double_integ}
\end{align}
and admits bounded control signals $\mid \control(t) \mid \le 1$ for all time $t$. After a change of variables,we have the following system of first-order differential equations
\begin{align}
	\dot{\state}_1(t) &= \state_2(t), \nonumber \\
	\dot{\state}_2(t) &= \control(t), \quad \mid \control(t) \mid \le 1. 
	\label{eq:double_integ_first_ord}
\end{align}
The \textit{reachability problem that we consider is to address the question of what states can reach a certain point (here, the origin) in a transient manner}. That is, we would like to find point sets on the state space, at a particular time step, such that we can bring the system to the equilibrium, $\left(0, 0\right)$. 

\subsubsection{Time-optimal control scheme}
This is an $\hamfunc$-minimal control problem whereupon we must find the control law that minimizes the Hamiltonian
\begin{figure}[tb!]
	\centering
	\includegraphics[width=\columnwidth]{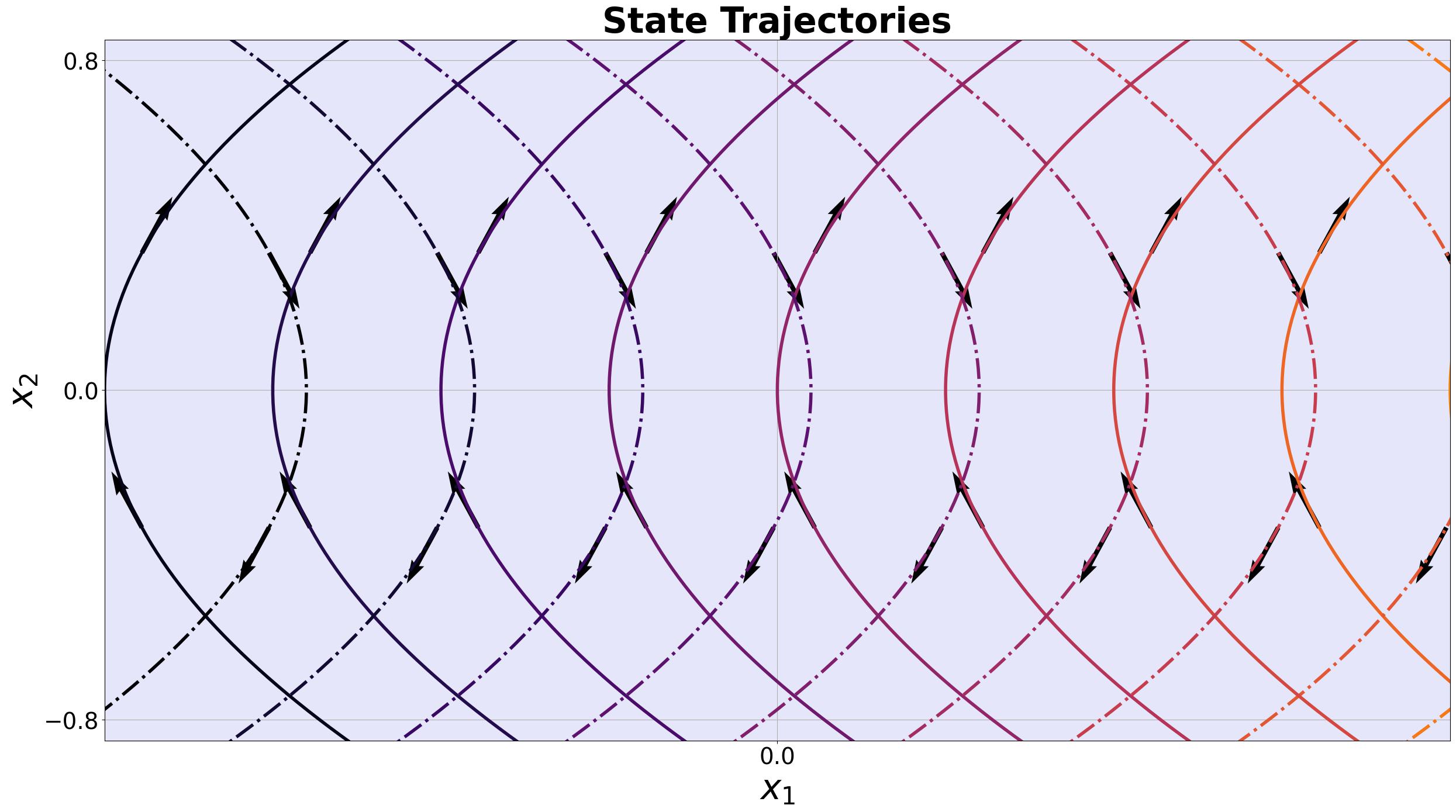}
	\caption{State trajectories of the double integral plant. The solid curves with upward-pointing arrows are trajectories generated under $\control=+1$ (switching curve $\switchcurve_+$); dashed curves with downward-pointing arrows are trajectories under $\control=-1$ (switching curve $\switchcurve_-$).}
	\label{fig:ttr_trajectories}
\end{figure}
\begin{align}
	\hamfunc(\state, p) = p_1 \dot{\state}_1 + p_2 \dot{\state}_2.
	\label{eq:dint_ham}
\end{align}
The necessary optimality condition stipulates that the minimizing control law be 
\begin{align}
	\control(t) = -\text{ sign }(p_2(t)) \triangleq \pm 1.
\end{align} 

For the co-states in question, suppose that their initial values (for constants $k_1$ and $k_2$) are $p_1(t_0) = k_1$ and $p_2(t_0) = k_2$,  only four candidates can serve as time-optimal control sequences \ie $\{[+1], [-1], [+1, -1], [-1, +1]\}$.  On a finite time interval, $t \in [t_0, t_f]$, the time-optimal $\control(t)$ is a constant $k \equiv \pm 1$ so that for initial conditions $\state_1(t_0) = \bm{\xi}_1$ and $\state_2(t_0) = \bm{\xi}_2$, it can be  verified that the state trajectories obey the relation
\begin{align}
	\state_1(t) = \bm{\xi}_1 + \frac{1}{2} k  \left(\state_2^2 - \bm{\xi}_2^2\right), \,\text{for } t = k \left(\state_2(t) - \bm{\xi}_2\right).
	\label{eq:trajectories}
\end{align}

\begin{figure}[tb!]
	\centering 
	\includegraphics[width=\columnwidth]{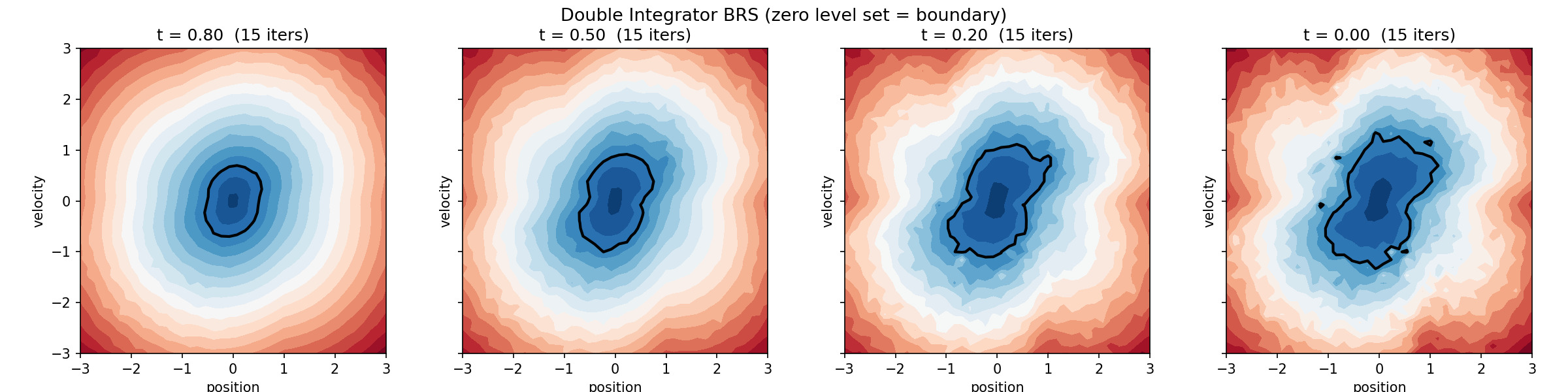}
	\caption{The time to reach the origin as 2D slice comparison for the double integratral plant}
	\label{fig:double_integrator}
\end{figure}
The  trajectories of \eqref{eq:trajectories} traced out over a finite time horizon $t=[-1, 1]$ with \textit{piecewise constant control laws}, $u = \pm 1$ on a state space and under the control laws $\control(t) = \pm 1$ is depicted in \autoref{fig:ttr_trajectories}. Curves with arrows that point upwards denote trajectories under the control law  $\control=+1$; call these trajectories $\switchcurve_+$; while the  trajectories marked by dashed arrows pointing downward on the curves were executed under $\control=-1$; call these trajectories $\switchcurve_-$. \\

 Table~\ref{tab:double_integrator_convergence} reports convergence at four backward times; all cases achieve $15$ iterations with residuals well below the  $O(\sqrt{\delta}) \approx 0.28$ bound.
\begin{table}[tb]
	\centering
	\caption{Double integrator BRS; Algorithm~\ref{alg:quasi_lin} convergence at backward times $t \in \{0.8, 0.5, 0.2, 0.0\}$ with $N = 8{,}000$ and $\delta = 0.08$.}
	\label{tab:double_integrator_convergence}
	\begin{tabular}{l c c}
		\toprule
		$t$ & Iterations & Final Residual \\
		\midrule
		0.80 & 15 & $0.0143$ \\
		0.50 & 15 & $0.0334$ \\
		0.20 & 15 & $0.0596$ \\
		0.00 & 15 & $0.0828$ \\
		\bottomrule
	\end{tabular}
\end{table}
\subsection{The Game of Two rockets on a Plane}

We adopt the rocket launch problem of Dreyfus~\cite{Dreyfus1966} which is to launch a rocket in fixed time to a desired altitude, given a final vertical velocity component and a maximum final horizontal component as constraints. The  rocket's motion is dictated by the following differential equations (under Dreyfus' assumptions)
\begin{subequations}
	\begin{align}
		\dot{x}_{1} &= x_{3}; \,\, &x_{1}(t_0) = 0; 
		\\ 
		\dot{x}_{2} &= x_{4},\,\, &x_{2}(t_0)= 0; \label{eq:dreyfus_mitter_ii}  
		\\ 
		\dot{x}_{3} &= a \cos u,\, &x_{3}(t_0)= 0; 
		\\ 
		\dot{x}_{4} &= a \sin u - g,\,\, &x_{4}(t_0)= 0; \label{eq:dreyfus_mitter_iv} 
	\end{align}
	\label{eq:dreyfus_mitter}
\end{subequations}
where, $(x_1, x_2)$ are respectively the horizontal and vertical range of the rockets (in feet), $(x_3, x_4)$ are respectively the horizontal and vertical velocities of the rockets (in feet per second), while $a$ and $g$ are respectively the acceleration and gravitational accelerations (in feet per square second). 
\begin{figure}[tb!]
	\centering 
	\includegraphics[width=0.48\textwidth]{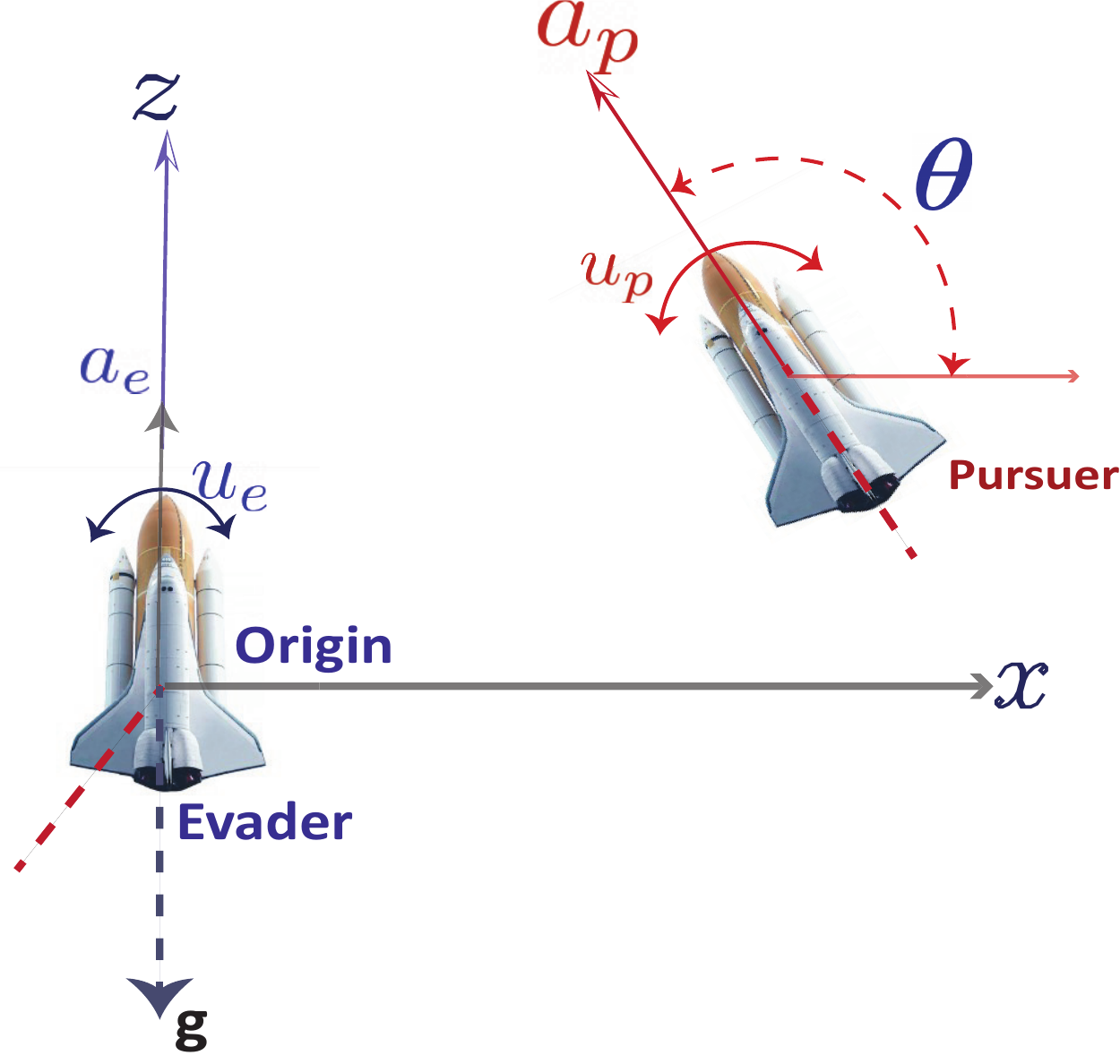} 
	\hfill 
	\includegraphics[width=0.48\textwidth]{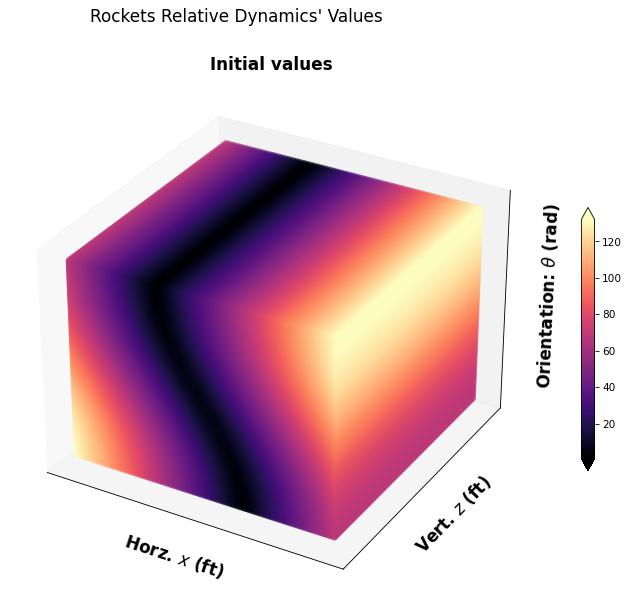} 
	\caption{\textit{Left:} Motion of two rockets on a Cartesian $\bm{xz}$-plane with a thrust inclination in relative coordinates given by $\theta:=u_p- u_e$. \textit{Right:} Representation of the initial values as a heat map for the two rockets described in \eqref{eq:rocket_spatial_bounds}.}	\label{fig:rocket_relative}
\end{figure}
Being a free endpoint problem, we transform it into a game between two players \eqref{eq:dreyfus_mitter} without the terminal time constraints as defined in ~\cite{JacobsonMayne}. The states of $\pursuer$ and $\evader$ are now denoted as $(x_p, x_e)$ respectively which are driven by their thrusts $(u_p, u_e)$ respectively in the $xz$-plane (see Figure \ref{fig:rocket_relative}). The relevant kinematic equations are \eqref{eq:dreyfus_mitter_ii} and \eqref{eq:dreyfus_mitter_iv}. 
 \begin{subequations}
	 	\begin{align} 
		 		\dot{x}_{2e} &= x_{4e};   &\dot{x}_{2p} &= x_{4p},  \\ 
		 		\dot{x}_{4e} &= a \sin u_e - g, &\dot{x}_{4p} &= a \sin u_p - g
		 	\end{align}
	 \label{eq:dreyfus_mitter_relevant_eq}
	 \end{subequations}
\noindent where $a$ and $g$ are respectively the acceleration and gravitational accelerations (in feet per square second) \ie $a=64 ft/sec^2$ and $g=32 ft/sec^2$. 
We reformulate the problem as a two-player differential game where the optimization operations remain in the interior, avoiding discontinuous switches in control. This formulation implicitly bounds each rocket's speed through the gravitational dynamics; the natural speed limit is $a\sin u / g$, which corresponds to each rocket's asymptotic speed when launched vertically.


Therefore, we rewrite \eqref{eq:dreyfus_mitter} with $\pursuer$'s motion relative to $\evader$'s  along  the $(x,z)$ plane so that the relative orientation as shown in \autoref{fig:rocket_relative} is $\theta=u_p- u_e$. The coordinates of $\pursuer$ are freely chosen; however, the coordinates of $\evader$ are chosen a distance $r$ away from $(x,z)$ so that the $\evader \pursuer$ vector's inclination measured counterclockwise from the $x-$axis is $\theta$. Following the conventions in \autoref{fig:rocket_relative}, the game's relative equations of motion in reduced space 
is $\mc{X} = (x, z, \theta)$ where $\theta \in \left[-\frac{\pi}{2}, \frac{\pi}{2}\right)$ and $(x,z) \in \bb{R}^2$ are 
\begin{subequations}
	\begin{align}
		\dot{x} &= a_p \cos \theta + u_e x, \\
		\dot{z} &=a_p \sin \theta + a_e + u_e x - g, \\
		\dot{\theta} &= u_p -u_e.
	\end{align}
\end{subequations}

The payoff, $\payoff$, is the distance of $\pursuer$ from $\evader$ when capture occurs denoted as $\|\pursuer \evader\|_2$. Capture occurs when $\| \pursuer \evader \|_2 \le r$ for a pre-specified capture radius, $r>0$. In \eqref{eq:rocket_me},  we say $\pursuer$ controls $u_p$ and is minimizing $\payoff$, and $\evader$ controls $u_e$ and is maximizing $P$. The boundary of the \textit{usable part} of the origin-centered circle of radius $r$\footnote{We set $r=1.5ft$ in our evaluations.} is $\|\pursuer \evader\|_2 $ so that
\begin{subequations}
	\begin{align}
		r^2 &=  x^2 + z^2,
	\end{align}
\end{subequations}
and all capture points are specified by 
\begin{align}
	\dot{r}(x,t) + \min \left[0, \hamfunc(\state, \frac{\partial r(x, t)}{\partial x})\right] \le 0,
	\label{eq:rocket_me}
\end{align}
with the corresponding Hamiltonian
\begin{align}
	\hamfunc(\state, p) = -\max_{u_e \in \mc{U}_e} \min_{u_p \in \mc{U}_p
	} \begin{bmatrix}
		p_1 \\ p_2 \\ p_3
	\end{bmatrix}^T
	\begin{bmatrix}
		a_p \cos \theta + u_e x \\
		a_p \sin \theta + a_e + u_p x - g \\
		u_p -u_e
	\end{bmatrix}.
	\label{eq:ham_def}
\end{align}
Suppose that the maximizing $u_e$ is $\bar{u}_e$ and the minimizing $u_p$ is $\bar{u}_p$. We have at the point of slowest-quickest descent on the capture surface, that
\begin{subequations}
	\begin{align}
		\bar{u}_e &= p_1 x - p_3, \\
		\bar{u}_p &= p_3 - p_2 x.
	\end{align}
\end{subequations}
We set the linear velocities and accelerations equal to one another \ie $u_e = u_p$ and $a_e = a_p$. Thus, the Hamiltonian takes the form
\begin{align}
	\hamfunc(\state, p) &= -\cos(u) |a p_1| + \cos(u) |a p_1| -\sin (u) |a p_2| - \nonumber \\ 
	& \qquad \qquad \sin (u) | ap_2 | + u | p_3| - u |p_3|.
\end{align}
%
Using a distributed version of the levelset toolbox~\cite{LevelSetPy}, the backward reachable tube of the game is depicted in \autoref{fig:rockets_results}. A game between the two players was run over 11 global optimization time steps. The initial value function (left inset of \autoref{fig:rockets_results}) is represented as a dynamic implicit surface over all point sets in the state space with a signed distance function. We made the third coordinate axis of the state space (here the common heading of the two rockets) to align with the third cylinder axis. The final BRT at the end of the optimization run is shown in the right inset of \autoref{fig:rockets_results}. 
\begin{figure}[tb!]
	\centering
	\begin{minipage}[tb]{.45\textwidth}
		\includegraphics[width=\textwidth]{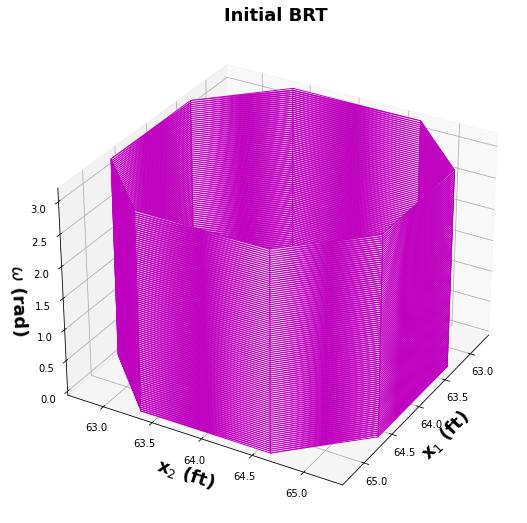}
	\end{minipage}
	\hfill
	\begin{minipage}[tb]{.45\textwidth}
		\includegraphics[width=\textwidth]{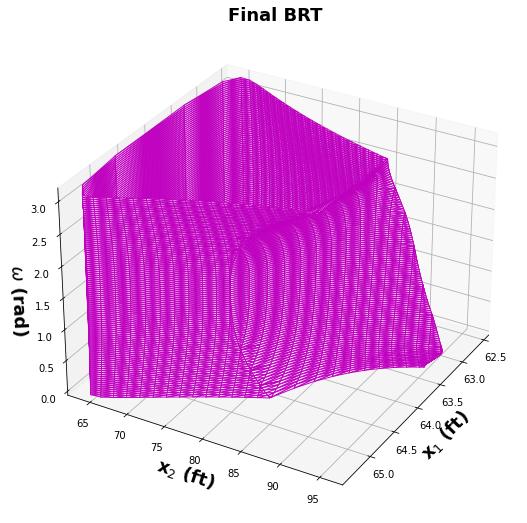}
	\end{minipage}
	\caption{Initial and final backward reachable tubes for the rocket system (\cf \autoref{fig:rocket_relative}) computed using the method outlined in \cite{Evans1984, OsherFronts, MitchellLSToolbox}. We set $a_e = a_p = 64ft/sec^2$ and $g=32 ft/sec^2$ as in Dreyfus' original example. 
		We compute the reachable set by optimizing for the paths of slowest-quickest descent in equation \eqref{eq:ham_def}.
	}
	\label{fig:rockets_results}
\end{figure}
%

\subsubsection{Detailed Hamiltonian Derivation for Two-Rockets Game}
\label{app:ham_derivation}

For the two-rockets pursuit-evasion game in relative coordinates, we derive the simplified Hamiltonian from first principles. The general Hamiltonian for a two-player zero-sum game is
\begin{align}
	\hamfunc(t; \state, p) &= -\max_{u_e \in [\underline{u}_e, \bar{u}_e]} \min_{u_p \in [\underline{u}_p, \bar{u}_p]}
	\begin{bmatrix}
		p_1 & p_2 & p_3
	\end{bmatrix}
	\begin{bmatrix}
		a_p \cos \theta + u_e x \\
		a_p \sin \theta + a_e + u_p x - g \\
		u_p - u_e
	\end{bmatrix},
	\label{eq:ham_general}
\end{align}
where $p_1, p_2, p_3$ are the co-states (spatial derivatives of the value function), and $[\underline{u}_e, \bar{u}_e]$ and $[\underline{u}_p, \bar{u}_p]$ are the evader's and pursuer's control bounds, respectively. The game structure reflects the evader's ability to choose controls first (outer max over $u_e$) followed by the pursuer's response (inner min over $u_p$).

Setting symmetric control bounds $\underline{u}_e = \underline{u}_p = -1$ and $\bar{u}_e = \bar{u}_p = +1$, and using the symmetry assumption $a_e = a_p = a$, we optimize over the control variables:
\begin{align}
	\hamfunc &= -\max_{u_e \in [-1, 1]} \min_{u_p \in [-1, 1]}
	\left[ p_1 (a \cos \theta + u_e x) + p_2 (a \sin \theta + a + u_p x - g) + p_3(u_p - u_e) \right].
\end{align}
Rearranging by control terms:
\begin{align}
	\hamfunc &= -\max_{u_e \in [-1, 1]} \min_{u_p \in [-1, 1]}
	\left[ p_1 a \cos \theta + p_2 a(1 + \sin \theta) - p_2 g + (p_1 x + p_3) u_p + (p_2 x - p_3) u_e \right].
\end{align}
Since the pursuer minimizes over $u_p$, it chooses $u_p = -\text{sign}(p_1 x + p_3)$ to minimize the linear term. Similarly, the evader maximizes over $u_e$, choosing $u_e = \text{sign}(p_2 x - p_3)$. For bang-bang optimal controls, these become:
\begin{align}
	\min_{u_p} (p_1 x + p_3) u_p &= -|p_1 x + p_3|, \\
	\max_{u_e} (p_2 x - p_3) u_e &= |p_2 x - p_3|.
\end{align}
Thus the simplified Hamiltonian becomes:
\begin{align}
	\hamfunc(t; \state, p) &= -\left[ a p_1 \cos \theta + p_2 (a + a \sin \theta - g) - |p_1 x + p_3| - |p_2 x - p_3| \right],
	\label{eq:rocket_ham_simplified}
\end{align}
which, when substituted into the HJ PDE yields the simplified rockets HJI equation used in the numerical experiments (cf.\ Section~\ref{subsec:quant_results}).

\begin{figure}[tb]
	\centering 
	\includegraphics[width=\columnwidth]{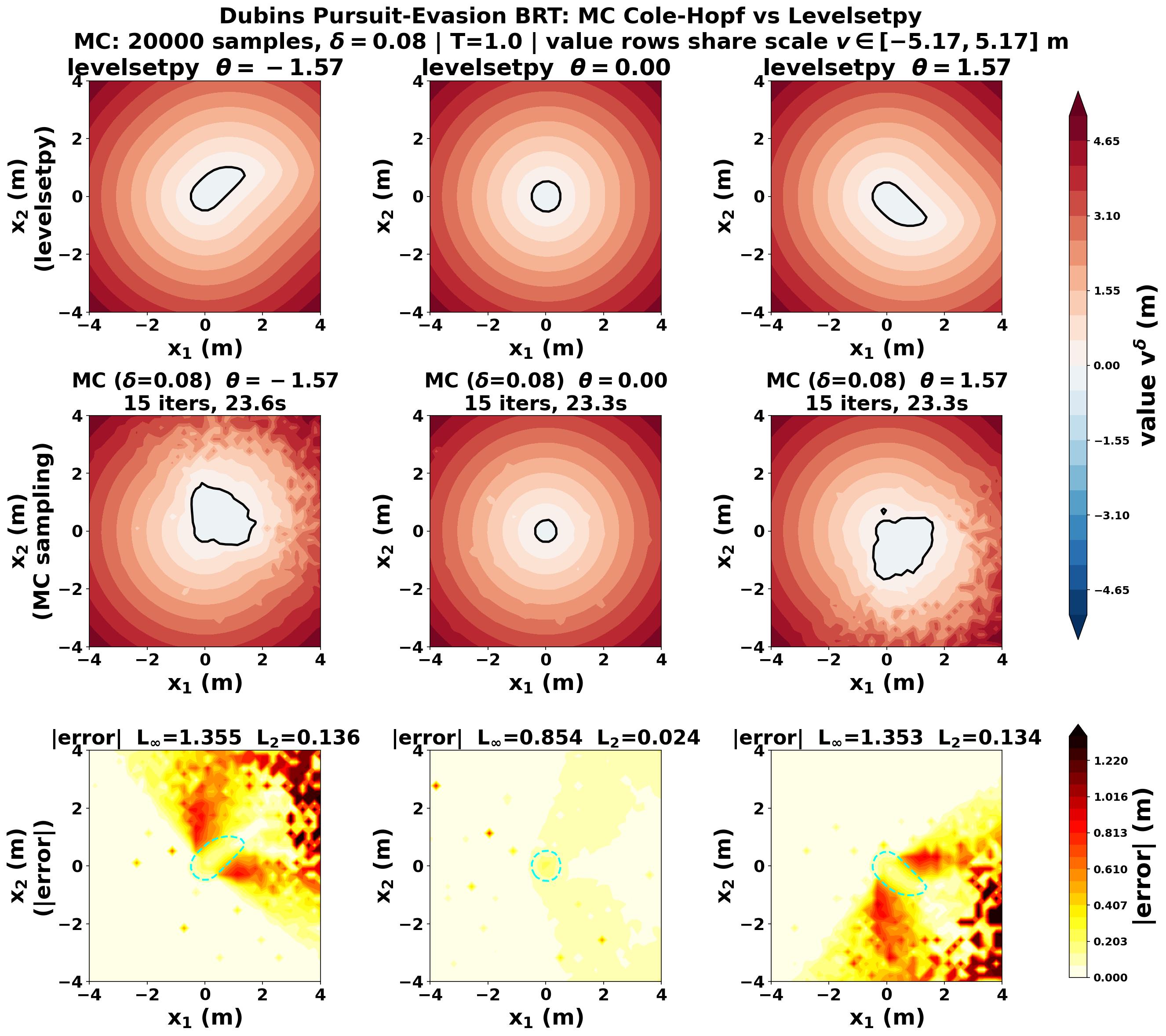}
	\caption{Dubins Vehicles' BRT Slices: Levelsetpy vs Monte Carlo (Ours). The bottom row shows a heat map of the errors at various relative vehicle orientations (the control input) between the Monte Carlo sampling scheme and the \texttt{LevelSetPy} scheme. Layout, color scale, and colorbars match \autoref{fig:rockets_slices}. A single representative Monte Carlo seed is shown; \autoref{tab:error_comparison} reports mean $\pm$ std over 30 seeds.}
	\label{fig:dubins_comparison}
\end{figure}

\subsection{Starlings murmurations safety analysis in a high dimensions}
\label{app:starlings}

We take inspiration from natural swarms (see \autoref{fig:murmurations}), particularly the murmuration of European starlings (\textit{sturnus vulgaris}). In these settings, local flocks within large murmurations  maintain an anisotropic formation based on a topological interaction, regardless of sparsity of birds on a phase space~\citep{Cavagna2010Scale}. Thus, intra- and inter-flock collisions are avoided and attacks are fended off~\citep{Ballerini1232}.  Approximating the viscosity solutions of nonconvex Hamilton-Jacobi partial differential equations with our quasilinearization scheme and importance sampling, the Hamiltonian, control laws, and strategies that govern the transient behaviors of many systems that possess structural subsystems with unique nearest neighbor properties may be computed.

\begin{figure}[tb!]
	\centering
	\begin{tabular}{ccc} 
		\includegraphics[height=8em,width=12.5em]{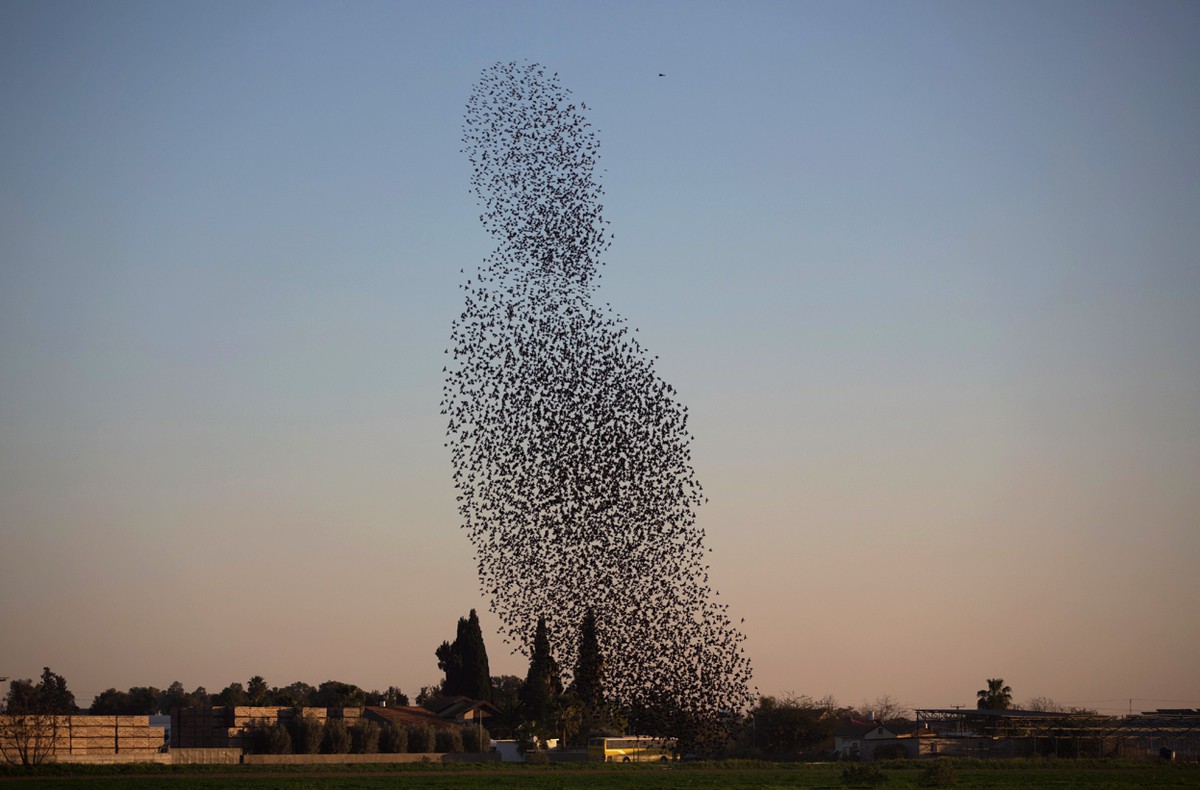} 
		&
		\includegraphics[height=8em,width=12.5em]{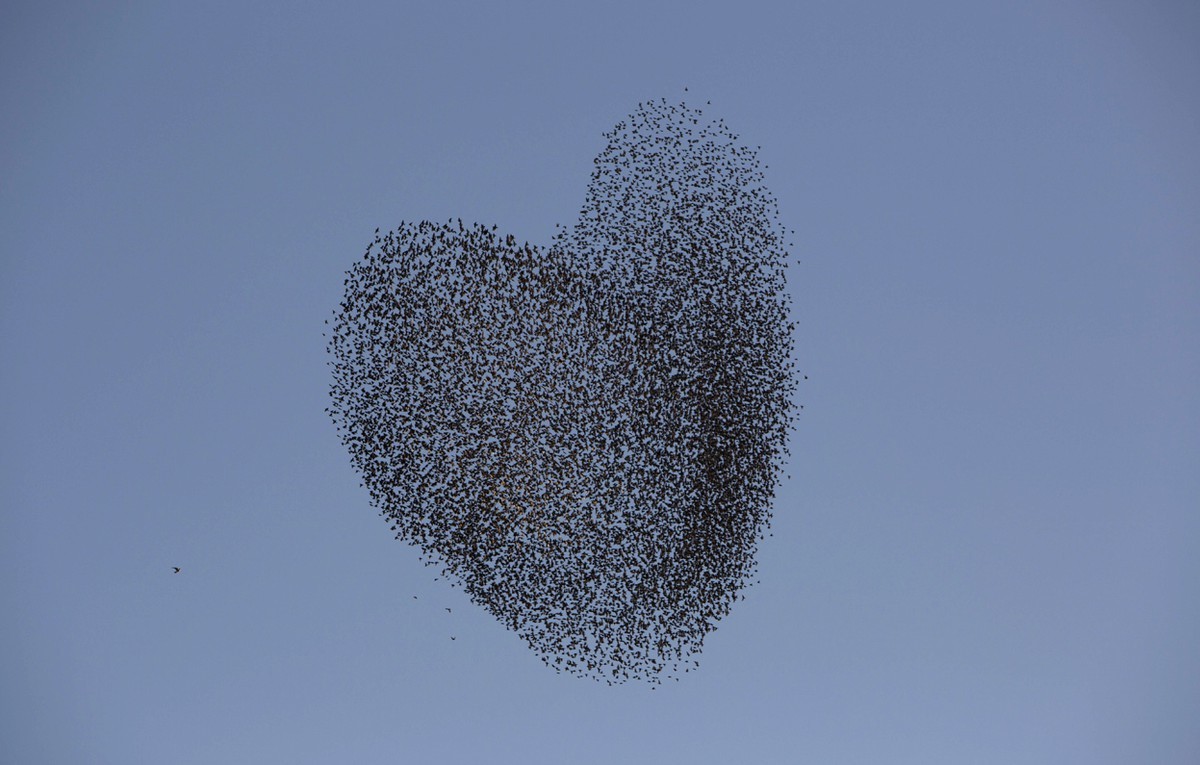} %
		& 
		\includegraphics[height=8em,width=12.5em]{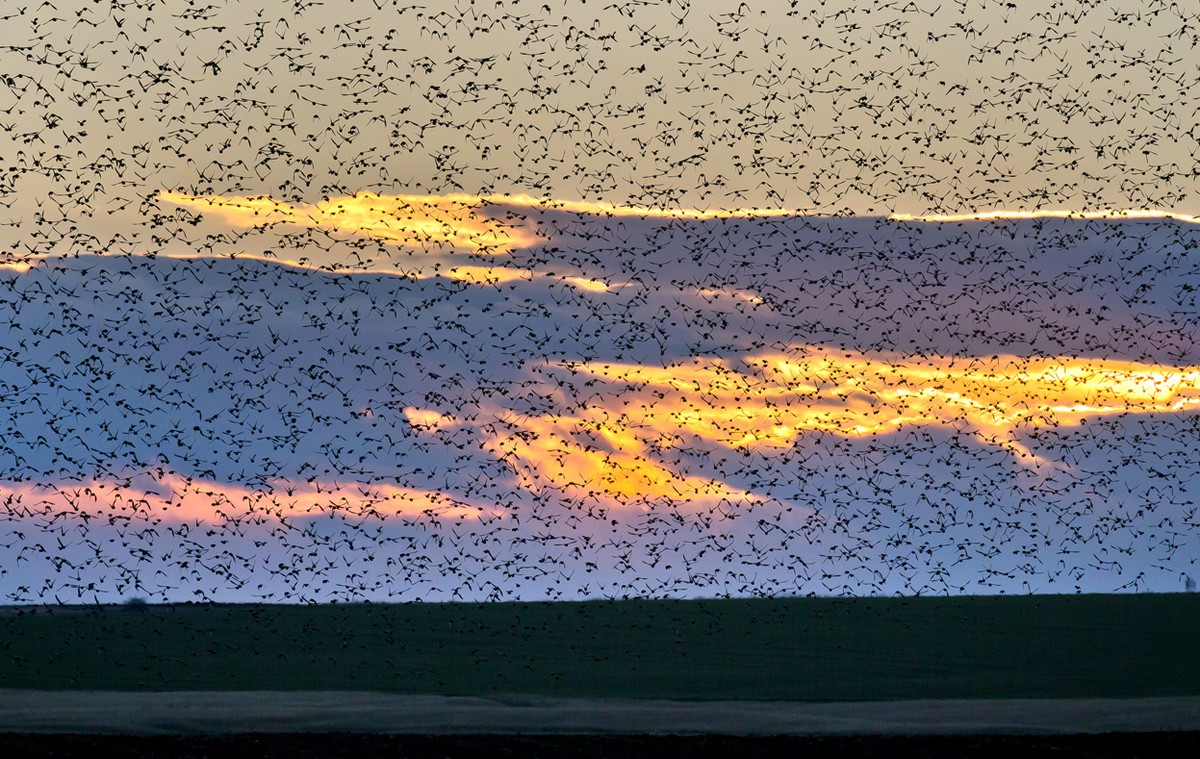} 
		\\	
		\includegraphics[height=8em,width=12.5em]{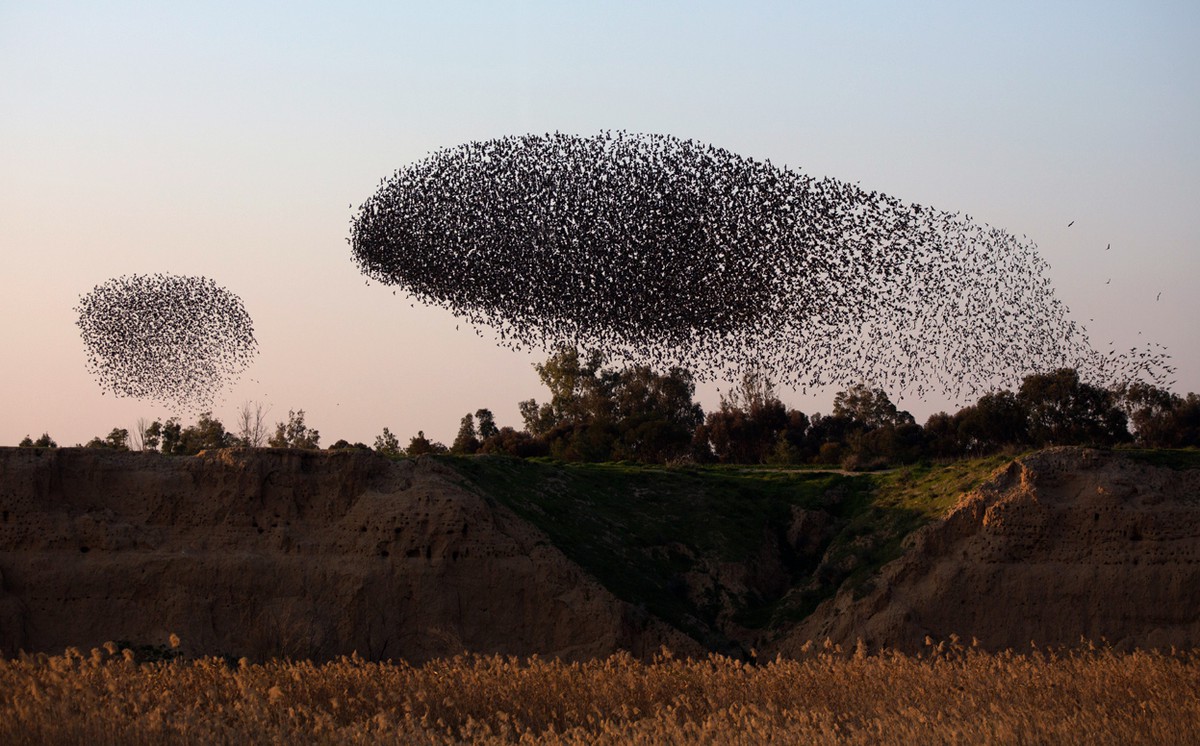}
		&
		\includegraphics[height=8em,width=12.5em]{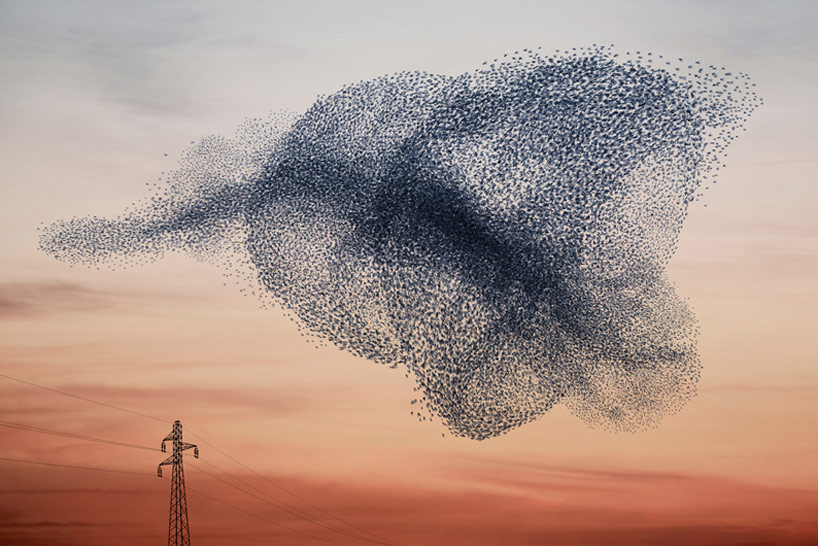} 
		&
		\includegraphics[height=8em,width=12.5em]{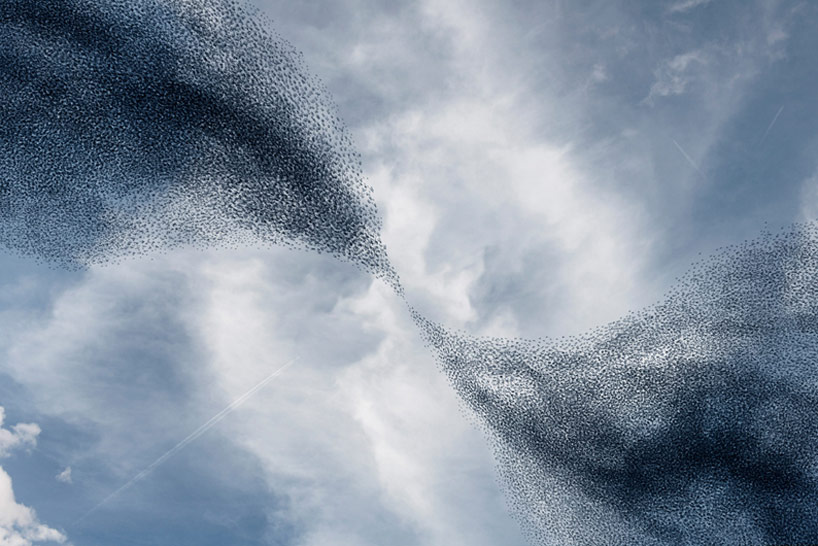}
	\end{tabular}
	\caption{The scalability veracity of our formalism is evaluated on $100,000$ European starlings (\textit{sturnus vulgaris}) in murmurations. From the top-left and clockwise. (i) A starlings flock rises into the air, in a dense structure (Reuters/Amir Cohen).  (ii) Starlings migrating over an Israeli village (AP Photo/Oded Balilty). (iii) Starlings feeding on laid seeds  in the ground in Romania. (iv) Two flocks of migrating starlings (Menahem Kahana/AFP/Getty Images). (v) A concentric conical formation of starlings (Courtesy of \href{http://www.thegatheringsite.net/qcgems/2014/1/24/murmuration}{The Gathering Site.}). (vi)  Splitting and joining of a flock of starlings.} 
	\label{fig:murmurations}
\end{figure}
Through empirical~\citep{Ballerini1232, Cavagna2010Scale, Helbing20, VicsekPhaseNovel, Bialek2012Statistical} and theoretical findings~\citep{JadbabaieCoord}, evidence now abounds that in certain natural species that exhibit collective behavior, convergence and group cohesion is based on simple topological interaction rules that they employ to keep a tab on one another in \textit{local flocks} for collision avoidance, preserving density and structure in an anisotropic formation, and exhibiting flock splitting, vacuole, cordon, and flash expansion isotropically~\citep{NatGeo}. This aids these animals in emerging an eye-pleasing local anisotropic synchrony, which taken together among possibly hundreds of thousands of local interactions\footnote{It has been reported that no birds fly together with greater coordination and complexity than European starlings,  with murmurations counting upwards of 750,000 individual birds!}~\citep{NatGeo}, keep these animals whirling, swooping, and flying in isotropic formations~\citep{Ballerini1232}. Thus, individual agents aggregate into substructures within the overall system, and overall group motion is synergized via local topological interactions so that  a stable global heading and cohesion~\citep{JadbabaieCoord} is preserved.

\subsubsection{Notations and Background}
A few mathematical notations, conventions, and taxonomy used throughout this section are in order at this juncture.
%
Capital and lower-case Roman letters are matrices and vectors respectively. Exceptions: time variables such as $t, t_0, t_f, T$ are real numbers throughout. Calligraphic letters are sets. Exception: the HJ equation's solution (shortly introduced) is the set $\openset \in \ren$ in which agents move. We work in a multi-agent system context where individual agents self-organize into phases or regions $\subgroup$ which are in turn members of a union of multiple regions $\group$. Note that every $\subgroup \subseteq \group$ and all members of $\group$ are disjoint from one another \ie $\subgroup_i \cup \subgroup_j = \emptyset$ for any $i\neq j$. The total number of elements in $\mc{S}$ is denoted $\left[\mc{S}\right]$, and we denote by $\text{ int } \openset$ the  interior of $\openset$.  The closure of $\openset$ is $\bar{\openset}$. We let $\delta \openset \,(:= \bar{\openset} \backslash \text{int } \openset)$ be the boundary of $\openset$. 
%

Structural homogeneity of starlings' motion in every region  $\subgroup_i \in \group$ for $i=1, \cdots, \left[\group\right]$ applies; this is enforced by introducing an external disturbance on the zeroth-index agent (this aids compactness of the zero levelset of an $\subgroup_i$ as we will introduce shortly). Hence, each safety verification episode can be characterized as a pursuit \textit{game}, $\Gamma$.  And by a game, we do not necessarily refer to a single game, but rather a \textit{collection of games}, $\Upsilon = \{\Gamma_1, \cdots, \Gamma_g\}$. Such a game terminates when \textit{capture} occurs, that is the distance between players falls below a predetermined threshold. As in previous examples, each player in a game shall constitute either a pursuer ($\pursuer$) or an evader ($\evader$). Let the cursory reader not interpret $\pursuer$ or $\evader$ as controlling a single agent. In our setup, we are poised with several pursuers (e.g. Falcons) or evaders (starlings). However, when $\pursuer$ or $\evader$ governs the behavior of but one bird, these symbols will denote the bird itself. An evading bird in a region $\subgroup_i$ has a state notation $\state_a^i$ (read: the state of bird $a$ in region $i$).  A state $\state_a^i$ has linear velocity components, $\state_{a_1}^i, \state_{a_2}^i$, and heading $\state_{a_3}^i:=w_a^i$. When we must distinguish an bird $\state_a^i \in \group_x$ from some other bird \eg in another multiphase $\group_y$, we shall write $^x\state_a^i$ and $^y\state_a^i$ respectively. 
Given the various possibilities of outcomes, the ``best" outcome is resolved by a \textit{payoff}, $\Phi$, whose extremal over a time interval will constitute a \textit{value}, $\valuefunc$\footnote{The functional $\bm{\Phi}$ may be considered a functional mapping from an infinite-dimensional space to the space of real numbers.}. We adopt~\citep{Isaacs1965}'s language so that if the payoff for a game is finite we shall have a \textit{game of kind} (a qualitative game); and for a game with a continuum of payoffs we shall have \textit{a game of degree} (quantitative games). The \textit{strategy} executed by $\pursuer$ or $\evader$ during a game shall be denoted by $\alpha \in \mc{A}$ (resp. $\beta \in \mc{B}$). 

The many interacting subsystems under consideration employ
\begin{inparaenum}[(i)]
	\item natural units of measurements that are the same for all birds; 
	\item kinematics with same linear speeds but with a capacity for orientation changes;
	\item inter-region interaction occurs within unique and distinct state space manifolds; and by birds maneuvering their direction, a kinematic alignment is obtained with other regions;
	\item region-to-region interaction occurs when a pursuer is within a threshold of capturing any bird in a region;
	\item the interaction among respective regions is described by the time-evolution of an interface, which is the zero-level set of the objective functional of the respective local subgroups.
\end{inparaenum} 

\subsubsection{Murmuration as a Levelsets Fronts  Equation}

The key idea is that when the multi-bird problem has separable structure, one can resolve the associated HJ equations in a numerically consistent manner by considering the interface of these separable structures, $i$, as evolving dynamic interfaces, $\interface_i$, via the modified levelset equation \eqref{eq:ivp}
\begin{align}
	\phi_t + \Gamma_i H(x, t, \phi, D\phi) = 0
\end{align}
At each time step, we advance each interface $\phi_i$ for birds in a region $i$ by solving the levelset equation for a small time step $\Delta t$; afterwards, we reconstruct the {unsigned} distance function that captures the safe set, which informs the controller  for the agents $\mc{S}_i \in \mc{C}$. Contrary to distributed consensus-type algorithms such as~\citep{MDDDP}, this framework guarantees safe exploration by directly incubating safety into the control optimization problem.

The motion of interfaces between subgroups are cast as Gaussian kernel approximations of the zero level set of implicitly defined {unsigned} distance functions (UDFs). Interfaces evolve by solving time-dependent Eulerian initial value partial differential equations of these SDFs in the form of viscosity solutions~\citep{Crandall1983viscosity} to hyperbolic conservation laws~\citep{CrandallMethodFracSteps}, which are essentially the original HJ equations~\citep{Evans1984}. Hence, the front's position gets updated by this means and the interface velocity is derived from physics on and off the interface.

Popular numerical tools for simulating the evolution of interfaces in applied physics, mathematics, and computational sciences include fast marching front tracking  methods~\citep{Sethian87Numerical, SethianFast}, Voronoi implicit interface methods (or VIIMs)~\citep{SethianVIIM,ZaitzeffVoronoi}, discretization schemes for Hamilton-Jacobi equations~\citep{Tsitsiklis95}, or multilayer volume-of-fluid methods~\citep{FoamingKarnakov} for foaming across scales.  While each of these schemes has its own advantages, we will resort to levelset methods~\citep{SethianLSBook} in developing algorithmic efficiency for the trajectory optimization among a \textit{collection} ($\mc{C}$) of \textit{disjoint} multiple starlings \textit{sets}, $\mc{S} \subseteq \group$. The entire collection of starlings moves over an open set $\openset \subset \ren$. 
When subgroups $\{\subgroup_i\}_{i=1}^{i=k} \subset \group, k<n$ of aerial starlings must traverse a narrow opening, or temporarily break apart to avoid collisions, for example, we expect topological changes to occur without explicit surgery.  Furthermore, group geometric formation and cohesion are achieved by local changes to subgroups' collective heading and speed. This is a challenging problem numerically and we look to natural behaviors in animals and multiphase simulations for inspiration. In our problem construction, a \textit{finite set} of individual starlings self-organize into local structural groups (which we refer to as subgroups, $\mc{S}$); subgroups interact based on nearest neighbor rules so that effective group consensus is dictated by interactions among separate subgroups~\citep{JadbabaieCoord}.

Throughout, our theater of operations involve multiple aerial Dubins vehicles~\citep{Dubins1957} with an extra dimension along the vertical position to describe the 4-dimensional motion of birds. We treat bird motions as kinematic models that possess linear and angular speeds as state variables. The safety verification optimization problem is played as a game of multiple vehicles in a four-dimensional space. As we are dealing with multiple starlings, objectives such as collision avoidance, starlings' spatial separation, overall group coherence become paramount. At issue is continually computing the points set that belong on the \textit{reachable set} boundary as the game advances forward in time. These reachable sets are those state space subsets where starlings may collide into one another, hit obstacles, lose topological integrity or miss group coherence behavior. 

\subsubsection{Related Works}
As all starlings move in $\openset \subset \ren$, we want a \text{stable and safe} numerical algorithm that correctly represents dynamic interface boundary conditions, accurately represents kinematics, whilst sensitive enough to rapidly interpret subgroup topological and structural changes. \citet{Tsitsiklis95} worked under the restrictive assumption that HJ's running cost is independent of the control laws' update -- providing an $O(n/p)$ parallel algorithm that resolved the discretized HJ-Hamiltonian for $p$ processors on $n$ grid points provided that $p=O(\sqrt{n}/\log n)$. In front tracking methods, a Lagrangian geometric representation techniques track the surfaces between separate structures with mechanisms such as triple point junctions which are endowed with shared nodes so that members of the front set are updated as time evolves~\citep{YongsamFrontTracking}. In applied mechanics and microfluidics, methods such as volume-of-fluid~\citep{VOF} resolve the interaction among multiple phases that are separated by thin boundaries with volume fraction fields for each region or unique functions in levelset methods -- leading to coalescence prevention among multiple regions. However, they come at a computational cost of $\mc{O}(N_{regions}N_{cells})$.  To improve its scalability, ~\citep{FoamingKarnakov} compactly stored many fields thereby keeping needed scalar fields constant and independent of regions to simulate.

\subsubsection{Problem Formulation}

We resolve local payoff extremals, $\{\payoff_1, \cdots, \payoff_{n_f}\}$. as a state space partition induced by an aggregation of desired collective behavior from local flocks' values $\{\valuefunc_1, \cdots, \valuefunc_{n_f}\}$.  Suppose that the local control laws are properly coordinated, the region of the state space across which their coordinated influence might be exerted constitute a larger \eg \textit{manipulability volume} for a dexterous kinematic task. We now formalize definitions that will aid the modularization of the problem into manageable forms.
\begin{definition}[Neighbors of an Agent]
	We define the neighbors $\mc{N}_i(t)$ of agent $i$ at time $t$ as the set of all agents that lie within a predefined radius, $r_i$.
\end{definition}
\begin{definition}
	We define a \textit{flock}, $\flock$, consisting of agents labeled $\{1, 2, \cdots, n_a\}$  as a collection of agents within a phase space $(\mc{X}, T)$ such that all agents within the flock interact with their nearest neighbors in a topological sense.
	\label{def:flock}
\end{definition}
\begin{remark}
	
	Every agent within a flock has similar dynamics to that of its neighbor(s). Furthermore, agents travel at the same linear speed, $v$; climb/dive rate (the control) $\climb \in [-\gamma_{\max}, \gamma_{\max}]$, where $\gamma_{\max} \in \reline$; the angular headings, $w$, however, may be different between agents, seeing we are dealing with a many-bodied system.  Each agent's continuous-time dynamics, $\dot{\state}^{(i)}(t)$, evolves as
	\begin{align}
		\begin{bmatrix}
			\dot{\state}^{(i)}_1(t) \\ \dot{\state}^{(i)}_2 (t) \\ \dot{\state}^{(i)}_3 (t) \\ \dot{\state}^{(i)}_4 (t)
		\end{bmatrix} 
		&= \begin{bmatrix}
			v(t) \cos \state_4^{(i)}(t) \\ v(t)\sin \state_4^{(i)}(t) 
			\\
			u_z(t)
			\\ \langle w^{(i)}(t) \rangle_r
		\end{bmatrix}, \,\, \langle w^{(i)}(t) \rangle_r &= \dfrac{1}{1+n_i(t)}\left(w^{(i)}(t) + \sum_{j \in \mc{N}_i(t)}^{} w_j(t)\right) 
		%
		\label{eq:DubinsJadbabaie}
	\end{align}
	\text{ for agents } $i = \{1, 2, 3, ..., n_a\}$, where $t$ is the continuous-time index, $n_i(t)$ is the number of agent $i$'s neighbors at time $t$, $\mc{N}_i(t)$ denotes the sets of labels of  agent $i$'s neighbors at time $t$, and $\langle w^{(i)}(t) \rangle_r$ is the average orientation of agent $i$ and its neighbors at time $t$. Note that for a game where all agents share the same constant linear speed and heading, and $u_z = 0$, \eqref{eq:DubinsJadbabaie} reduces to the dynamics of a Dubins' vehicle in absolute coordinates with $-\pi \le w^{(i)}(t) < \pi$. The averaging over the degrees of freedom of other agents in \eqref{eq:DubinsJadbabaie} is consistent with \textit{mean field theory}, where the effect of all other agents on any one agent is an approximation of a single averaged influence.
\end{remark}


\begin{definition}[Payoff of a Flock]
	To every flock $F_j$ (with a finite number of agents $n_{a}$) within a murmuration, $j = \{1, 2, \cdots, n_f\}$ , we associate a payoff, $\bm{\Phi}_j$, that is the union of all respective agent's payoffs for expressing the outcome of a desired kinematic behavior. 
	\label{def:payoff}
\end{definition}
%
%
%
%

Flock $\flock_j$ and $\flock_k$ within a murmuration, $\flock_j \cup \flock_k \cup \flock_l \cdots $ are separated by \textit{partitions}, or \textit{interfaces}, $\interface_{jk}, \interface_{kl}, \cdots$. This interface may be  implicitly represented as a signed distance function $\payoff(\state)$ which is negative on the interior of each flock, and zero on the edges. The zero-level set (\ie $\payoff(\state)=0$) corresponds to the  interface $\valuefunc$~\citep{Sethian87Numerical}. As the system evolves over time, $\flock_j$'s interface (zero-level set) motion can be parameterized by time, so that the flow field $\valuefunc(\state, t)$ is equivalent to the solution of the 
Cauchy-type Hamilton Jacobi partial differential equation~\citep{Evans1984, Crandall1983viscosity}:
\begin{align}
	\valuefunc_t + v_j | \nabla \valuefunc_j | = 0, \quad j = 1, \cdots, n_f,
	\label{eq:level_set}
\end{align} 
where $v_j$ is the flow speed for $\flock_j$. Equation \eqref{eq:ivp} is the level set equation~\citep{SethianLSBook}.

In the sentiment of \citep{Mitchell2005}, we say the zero sublevel set of $g(\cdot)$ \ie $\mathcal{L}_0$ in \eqref{eq:target_set} 
%
%
is the \textit{target set} in the phase space $\openset \times \mathbb{R}$ for a backward reachability problem~\citep{Mitchell2001}. 
This target set\footnote{Note that the target set, $\mathcal{L}_0$, is a closed subset of $\ren$ and is in the closure of $\openset$.} can represent the failure set, regions of danger, or obstacles to be avoided e.t.c. in the vectogram.  And the \textit{robustly controlled backward reachable tube} for $\tau \in [-T, 0]$\footnote{The (backward) horizon, $-T$ is negative for $T>0$.} is the closure of the open set $\mathcal{L}([\tau, 0], \mathcal{L}_0) $ in \eqref{eq:rcbrt}.

Each agent within a flock interacts with a fixed number of neighbors, $n_c$, within a fixed topological range, $r_c$. This topological range 
is consistent with findings in collective swarm behaviors and it reinforces \textit{group cohesion}~\citep{Ballerini1232}. However, we are interested in \textit{robust group cohesion} in reachability analysis. Therefore, we let a pursuer, $\pursuer$, with a worst-possible disturbance attack the flock, and we take it that flocks of agents constitute an evading player, $\evader$. 
Returning to \eqref{eq:DubinsJadbabaie}, for a single flock, we now provide a sketch for the HJI formulation for a heading consensus problem. 

\subsubsection{Framework for Separated Payoffs.}
\label{subsec:framework}
Suppose that a murmuration's global heading is predetermined and each agent $i$ within each flock, $\flock_j,\, (j=\{1, \cdots, n_f\})$ in the murmuration has a constant linear velocity, $v^i$. An agent's orientation is its control input, given by the average of its own orientation and that of its neighbors. Instead of metric distance interaction rules that make agents very vulnerable to predators~\citep{Ballerini1232}, we resort to a topological interaction rule. With metric distance rules, we will have to formulate the breaking apart of value functions that encode a consensus heading problem in order to resolve the extrema of multiple payoffs; which is typically what we want to prevent in real-world autonomous tasks.

%
What constitutes an agent's neighbors are computed based on empirical findings and studies from the lateral vision of birds and fishes~\citep{Ballerini1232, JadbabaieCoord, Helbing20} that provide insights into their anisotropic kinematic density and structure. Importantly, starlings' lateral visual axes and their lack of a rear sector reinforces their lack of nearest neighbors in the front-rear direction. As such, this enables them to maintain a tight density and robust heading during formation and flight.

Each agent within a flock $\flock_j$ interacts with a fixed number of neighbors, $n_c$, within a fixed topological range, $r_c$. The topological range can be set as the distance between the labels of agents in a flock. This topological range 
is consistent with findings in collective swarm behaviors and it reinforces \textit{group cohesion}~\citep{Ballerini1232}. However, we are interested in \textit{robust group cohesion} in reachability analysis. Therefore, we let a pursuer, $\pursuer$, with a worst-possible disturbance attack the flock, and we take it that flocks of agents constitute an evading player, $\evader$. 

\subsubsection{Global Isotropy via Local Anisotropy.}
\label{subsec:aniso}
%
Structural anisotropy is not merely an effect of a preferential velocity in animal flocking kinematics but rather an explicit effect of the anisotropic interaction character itself: agents choose a mutual position on the state space in order to maximize the sensitivity to changes in heading and speed of neighbors as the neighbors' anisotropy is optimized via vision-based collision avoidance  characteristically unrelated to the eye's structure~\citep{Ballerini1232}.

To reinforce robust group cohesion in local flocks, we randomly simulate a pursuer $\pursuer_j$ against an evading agent in every flock $\flock_j$ so that one agent is always relative coordinates with $\pursuer^j$. 
In this specialized case,   the $\evader$ and $\pursuer$'s speeds and maximum turn radii are equal: if both players start the game with the same initial velocity and orientation, the relative equations of motion show that $\evader$ can mimic $\pursuer$'s strategy by forever keeping the starting radial separation. As such, the \textit{barrier} is closed and \textit{the central theme in this \textit{game of kind} is to determine the surface}~\citep{Merz1972}. We defer a thorough analysis of the nature of the surface to a future work. 

Owing to the high-dimensionality of the state space, we cannot resolve this barrier analytically, hence we resort to numerical approximation methods -- in particular, we leverage a parallel Lax-Friedrichs integration scheme~\citep{Crandall1984} which we implement in Cupy~\citep{CuPy} in order to provide a \textit{consistent} and \textit{monotone} solution to the Hamiltonians of the flocks. The assembly in the large of these respective Hamiltonians, and hence numerically robust solutions to the variational backward reeachability problem is resolved with a Voronoi tesselation of the zero-level sets of the boundaries of the flocks. 

Therefore, for an agent $i$ within a flock with index $j$ in a murmuration, the equations of motion under attack from a predator $p$ in relative coordinates is 
\begin{align}
	\left[\begin{array}{c}
		\dot{\state}_1^{(i)_j}(t) \\ \dot{\state}_2^{(i)_j} (t) \\ \dot{\state}_3^{(i)_j} (t) \\ \dot{\state}_4^{(i)_j} (t)
	\end{array}\right] = \begin{bmatrix}
		-v_e^{(i)_j}(t) + v_p^{(j)} \cos \state_4^{(i)_j}(t) + \langle w_e^{(i)_j} \rangle_r \state_2^{(i)_j}(t)
		\\
		v_p^{(i)_j}(t)\sin \state_4^{(i)_j}(t) - \langle w_e^{(i)_j} \rangle_r \state_1^{(i)_j}(t)
		\\ 
		u_p^{(j)}(t) -  u_e^{(i)_j}(t) 
		\\ 
		w_p^{(j)}(t) - \langle w_e^{(i)_j}(t) \rangle_r
	\end{bmatrix} 
	\label{eq:DubinsRelative}
\end{align}
\text{for} $i=1,\cdots, n_a$ where $n_a$ is the number of agents within a flock, $\left(\state_1^{(i)_j}(t), {\state}_2^{(i)_j} (t)\right) \in \bb{R}^2$, $\state_3 \in [-\gamma_{\max}, \gamma_{\max}]$, and  we have $\state_4^{(i)_j} (t) \in \left[-\pi, +\pi\right)$\footnote{We have multiplied the dynamics by $-1$ so that the extremal's resolution evolves backwards in time.}. Read ${\state}_1^{(i)_j}(t)$: the first component of the state of an agent $i$ at time $t$ which belongs to the flock $j$ in the murmuration at time $t$. In absolute coordinates, the equation of motion for \textit{free agents} is 
\begin{align}\left[\begin{array}{c}
		\dot{\state}_1^{(i)_j}(t) \\ \dot{\state}_2^{(i)_j} (t) \\ \dot{\state}_3^{(i)_j} (t) \\ \dot{\state}_4^{(i)_j} (t)
	\end{array}\right] = \begin{bmatrix}
		v_e^{(i)_j}(t) \cos \state_4^{(i)_j}(t) 
		\\
		v_e^{(i)_j}(t)\sin \state_4^{(i)_j}(t) 
		\\ 
		u_z^{(i)_j}(t)
		\\
		\langle w_e^{(i)_j}(t) \rangle_r
	\end{bmatrix}.
\end{align}

\subsubsection{Flock Motion from Aggregated Value Functions.}
\label{subsec:flockbrat}
We introduce the union operator \ie $\cup $ below as an aggregation symbol since the respective payoffs of each agent in a flock may be implicitly or explicitly constructed\footnote{In resolving the zero-level sets of HJ value functions, it is typical to represent the payoff's surface as the isocontour of some function (usually a signed distance function).} -- when it is implicitly represented, say from a signed distance function, we shall aggregate the  payoff of agents 1 and 2 as
\begin{align}
	\cup &\left\{\payoff_1(\state, t), \payoff_2(\state, t)\right\} \equiv 	 \payoff_1(\state, t)\cup \payoff_2(\state, t) = 	\min(\payoff_1(\state, t), \payoff_2(\state, t))
	\label{eq:sdf_union}
\end{align}
otherwise, other appropriate arithmetic or logical operation shall apply.

We assume that the \textit{value} of a flock heading control (differential game) exists.  And by an extension of Hamilton's principle of least action, the terminal motion of a flock coincide with the extremal of the payoff functional 
\begin{align}
	\valuefunc(\state, t) &= \inf_{\beta^{(1)} \in \mathcal{B}^{(1)}} \sup_{\bm{u^{(1)}} \in \mathcal{U}^{(1)}} 	 g^{(1)}(\state(T)) \cup  \cdots
	\cup \inf_{\beta^{(n_f)} \in \mathcal{B}^{(n_f)}} \sup_{\bm{u}^{(n_f)} \in \mathcal{U}^{(n_f)}} 
	g^{(n_f)}(\state(T))
\end{align}
%
%
where $n_f$ is the total number of distinct flocks in a murmuration. The resolution of this equation  admits a viscosity solution to the following variational terminal HJI PDE \citep{Mitchell2005}
\begin{align}
	\cup _{j=1}^{n_f}&\left[
	\cup _{i=1}^{n_a}\left(\dfrac{\partial \valuefunc_i}{\partial t}(\state, t) + \min \left[0, \hamfunc^{(i)}(\state^{(i)}, \valuefunc_x(\state, t))\right]\right)\right]
	\nonumber \\
	&= 0.
	\label{eq:value_aggregate}
\end{align}
with Hamiltonian, 
\begin{align}
	\hamfunc^{(i)} &(t; \state^{(i)}, \bm{u}^{(i)}, \bm{v}^{(i)}, p^{(i)}) = \max_{u^{(i)} \in \mathcal{U}^{(i)}} \min_{v^{(i)} \in \mathcal{V}^{(i)}}  \langle f^{(i)}(t; \state, \bm{u}^{(i)}, \bm{v}^{(i)}), p^{(i)}  \rangle. 
	\label{eq:lower_visc_ham_surfaces}
\end{align}

In swarms' collective motion, when \eg a Peregrine Falcon attacks, immediate nearest agents change direction almost instantaneously. And because of the interdependence of the orientations of individual agents with respect to one another, all other agents respond instantaneously. Thus, we only simulate a single attack against a flock within the murmuration to realize robust cohesion.  

A pursuer can attack any flock within the murmuration from a distinct surface: a $\pursuer$ direction: this side of the surface reached after penetration in the $\pursuer-[\evader-]$ direction is the $\pursuer-[\evader]$ \textit{side}\citep{Isaacs1965}. We  attribute the term \textit{in the small} to determine the smooth parts of the singular surface solution when a pursuer attacks, and when they are stitched together into the total solution, we shall describe them as \textit{in the large}. There exists at least one value $\bar{\alpha}$ of $\alpha$ such that if $\alpha = \bar{\alpha}$, no vector in the $\beta$-vectogram\footnote{A $\beta-$vectogram is the resulting state space when a the strategy $\beta$ is applied in computing the optimal control law for an agent.} penetrates the surface in the $E$-direction. Similar arguments can be made for $\bar{\beta}$ which prevents penetration in the $P$-direction. We adopt ~\citep{Isaacs1965}'s terminology and call these surfaces semi-permeable surfaces (SPS).

Throughout the game, we assume that the roles of $\pursuer$  and $\evader$ do not change, so that when capture can occur, a necessary condition to be satisfied by the saddle-point controls of the players is the Hamiltonian, $\hamfunc^i(\state, p)$.  

\begin{tcolorbox}[title=\textbf{Benchmark Toolkit}, colback=yellow!3, colframe=blue!70, fonttitle=\small, fontlower=\small] 
	\small
	In comparison with benchmarks, the zero-level set is constructed implicitly from the  isocontour of a signed distance function as described in \cite[Chapter II]{LevelSetsBook}. 
	We introduce the union operator, $\cup $, below in lieu of the arithmetic summation symbol since the respective payoffs of each agent in a flock are implicitly initialized as signed distance functions on the state space. It is trivial to extend these results to other arithmetic or Boolean operators depending on different task domains. 
	Recall that if $\Phi_1(\state, t)$ and $\Phi_2(\state, t)$ are respectively the payoffs for agents $1$ and $2$, constructed from signed distance functions, then the \textit{union of their interiors}\footnote{We are only concerned with states that are within or on the boundary of the zero-level set or tube in backward reachability analysis. This follows from \eqref{eq:reachables}.} is given by \eqref{eq:sdf_union}.
\end{tcolorbox}

\begin{theorem}
	For  a flock, $\flock_j$, the Hamiltonian is the total energy given by a summation of the exerted energy by each agent $i$ so that we can write the \textit{main equation} or total Hamiltonian of a murmuration as 
	\begin{align}
		&\hamfunc(\state, p)= \max_{w_e^{(k)_j} \in [\underline{w}_e^j, \bar{w}_e^j]}  \min_{w_p^{(k)_j}  \in [\underline{w}_p^{j}, \bar{w}_p^j]} \cup _{j=1}^{n_f} \left[ H^{(k)_j}_a(\state, p) \cup  \left( \cup _{i=1}^{n_a-1} H^{(i)_j}_f(\state, p) \right) \right] 
		\\
		& = \cup _{j=1}^{n_f} \left( \cup _{i=1}^{n_a-1} 
		\left[p_1^{(i)_j} \, v^{(i)_j} \cos \state_4   + p_2^{(i)_j} \, v^{(i)_j} \sin \state_4  + p_3^{(i)_j} \, u_z^{(i)_j} + p_4^{(i)_j}\, \langle w_e^{(i)_j}\rangle_r\right] \right.\nonumber \\
		& \quad \left. \cup  \left[
		p_1^{(k)_j} \left(v^{(k)_j}  - v^{(k)_j} \cos \state_4^{(k)_j}\right) -   p_2^{(k)_j} v^{(k)_j} \sin \state_4^{(k)_j} - \underline{w}_p^j |p_4^{(k)_j}|   \right. \right. \nonumber \\
		&  \quad \left.  \left. 
		+ \bar{w}_e^j \bigg|p_2^{(k)_j} \state_1^{(k)_j} - p_1^{(k)_j}\state_2^{(k)_j} + p_3^{(k)_j}\bigg|
		\right] \right).
		\label{eq:HamiltonianOverall}
	\end{align}
	where $\hamfunc^{(k)_j}_a(\state, p)$ is the  Hamiltonian of the individual under attack by a pursuing agent, $\pursuer$, and $H^{(i)_j}_f(\state, p)$ are the respective Hamiltonians of the free agents, $i=\{1, \cdots, n_f\}$, within an evading flock in a murmuration, and not under the direct influence of capture or attack by $\pursuer$; we denote by  $w_e^{(i)_j}$ the heading of an evader $i$ within a flock $j$ and $w_p^{(j)}$  the heading of a pursuer aimed at flock $j$; $\underline{w}_e^{(k)_j}$ is the orientation that corresponds to  the orientation of the agent with minimum turn radius among all the neighbors of agent $k$, inclusive of agent $k$ at time $t$; similarly, $\bar{w}_e^{(k)_j}$ is  the maximum orientation among all of the orientation of agent $k$'s neighbors. 
	\label{th:ham_sum}
\end{theorem}

\begin{proof}
	From \eqref{eq:HamiltonianOverall}, the total Hamiltonian of a flock is a union of the mechanical energy of the free agents in a flock and the individual under attack. At each instant, we are extremizing over the headings ${w_e^{(k)_j} \in [\underline{w}_e^j, \bar{w}_e^j]},  {w_p^{(k)_j}  \in [\underline{w}_p^{j}, \bar{w}_p^j]}$ and climb rates $u_e^{(k)_j}, u_p^{(k)_j} \in [-\gamma_{\max}, \gamma_{\max}]$ respectively \ie
	\begin{align}
		\hamfunc(\state, p) &= \max_{[w_e^{(k)_j},u_e^{(k)_j}]}  \min_{[w_p^{(k)_j}, u_p^{(k)_j}]}   \cup _{j=1}^{n_f} \left[ H^{(k)_j}_a(\state, p) \cup  \left( \cup _{i=1}^{n_a-1} H^{(i)_j}_f(\state, p) \right) \right] 
	\end{align}
	We now set $\underline{\gamma}=-\gamma_{\max}$ and $\bar{\gamma} = \gamma_{\max}$ for ease of readability. We write the Hamiltonian of the free agents in absolute coordinates and the Hamiltonian of the agent under attack in relative coordinates with respect to the pursuer. A flock's Hamiltonian is  
	the aggregation of all the mechanical energy in the system in absolute coordinates \ie 
	%
	%
	\begin{align}
		&\cup _{i=1}^{n_a-1} \hamfunc^{(i)_j}_f(\state, p) =  \cup _{i=1}^{n_a-1} 
		\left[p_1^{(i)_j} \, v^{(i)_j} \cos \state_4 + p_2^{(i)_j} \, v^{(i)_j} \sin \state_4 + p_3^{(i)_j} u_{z_e}^{(i)_j} + p_4^{(i)_j}\, \langle w_e^{(i)_j}\rangle_r\right],
	\end{align}
	where we have again dropped the time arguments for convenience. And,
	\begin{align}
		\hamfunc^{(k)_j}_a(\state, p)&= - \left(\max_{[w_e^{(k)_j},u_e^{(k)_j}]}  \min_{[w_p^{(k)_j}, u_p^{(k)_j}]}   \begin{bmatrix}p_1^{(k)_j}(t) & p_2^{(k)_j}(t) & p_3^{(k)_j}(t) & p_4^{(k)_j}(t) \end{bmatrix}\right. \nonumber \\ 
		& \qquad\qquad  \left. 
		\begin{bmatrix}
			-v_e^{(k)_j}(t) + v_p^{(j)} \cos \state_4^{(k)_j}(t) + \langle w_e^{(k)_j} \rangle_r (t) \state_2^{(k)_j}(t)
			\\ 
			v_p^{j}(t)\sin \state_4^{(k)_j}(t) -\langle w_e^{(k)_j} \rangle_r (t) \state_1^{(k)_j}(t) 
			\\
			u_p^{(i)_j}(t)-u_e^{(i)_j}(t)
			\\ 
			w_p^j(t) - \langle w_e^{(k)_j}(t) \rangle_r
		\end{bmatrix}\right),
		\label{eq:Hamiltonian}
	\end{align}
	where $p_l^{(k)_j}(t)\mid_{l=1,2,3}$ are the adjoint vectors~\citep{Merz1972}. For the pursuer, its minimum and maximum turn and climb rates are fixed so that we have $(\underline{w}_p^{j}, \underline{\gamma})$ as the minimum turn and climb bounds of the pursuing vehicle, and $(\bar{w}_p^j, \bar{\gamma})$ are the maximum turn, and climb bounds of the pursuing vehicle, respectively. Henceforth, we drop the templated time arguments for ease of readability. From \eqref{eq:Hamiltonian}, we have 
	\begin{align}
		\begin{split}
			&\hamfunc^{(k)_j}_a(\state, p)=- \left(\max_{[w_e^{(k)_j},u_e^{(k)_j}]}  \min_{[w_p^{(k)_j}, u_p^{(k)_j}]}  
			\left[
			-p_1^{(k)_j} v_e^{(k)_j} + p_1^{(k)_j} v_p^{j} \cos \state_4^{(k)_j} +  p_1^{(k)_j}\langle w_e^{(k)_j} \rangle_r \state_2^{(k)_j} 
			\right. \right.  \\ 
			& \left. \left. \qquad
			+ p_2^{(k)_j} v_p^{j} \sin \state_4^{(k)_j} - p_2^{(k)_j} \langle w_e^{(k)_j} \rangle_r  \state_1^{(i)_j} 
			+ p_3^{(k)_j} \left(u_p^{(i)_j}-u_e^{(i)_j}\right)
			+ p_4^{(k)_j}\left(w_p^j - \langle w_e^{(k)} \rangle_r\right)
			\right] 
			\right),
			\\
			&\triangleq p_1^{(k)_j} \left(v_e^{(k)_j} - v_p^{j} \cos \state_4^{(k)_j}\right) -  p_2^{(k)_j} v_p^{j} \sin \state_4^{(k)_j} + \left(		
			\max_{\langle w_e^{(k)_j}\rangle_r \in [\underline{w}_e^j, \bar{w}_e^j], u_e^{(k)_j}\in [\underline{\gamma}, \bar{\gamma}]} \min_{w_p^{j}  \in [\underline{w}_p^{j}, \bar{w}_p^j], u_p^{j}  \in [\underline{\gamma}, \bar{\gamma}]} \right.\nonumber \\
			& \left. \qquad \qquad
			\begin{bmatrix}
				\langle w_e^{(k)_j}\rangle_r  
				 &\left(p_2^{(k)_j} \state_1^{(k)_j} - p_1^{(k)_j}  \state_2^{(k)_j} + p_4^{(k)_j}\right) 
				+ p_3^{(k)_j} \left(u_p^{(i)_j}-u_e^{(i)_j}\right) 
				\nonumber \\
				& 
				\qquad \qquad- p_4^{(k)_j} w_p^j 
			\end{bmatrix}
			\right).
		\end{split}
		\label{eq:ham_flock_interm}
	\end{align}
	So that,
	\begin{align}
		\hamfunc^{(k)_j}_a(\state, p) &= p_1^{(k)_j} \left(v_e^{(k)_j} - v_p^{j} \cos \state_4^{(k)_j}\right) - p_2^{(k)_j} v_p^{j}  \sin \state_4^{(k)_j} - \underline{w}_p^j |p_4^{(k)_j}| +  \bar{w}_e^j \bigg|p_2^{(k)_j} \state_1^{(k)_j} \nonumber 	\\
		& \quad 
		 - p_1^{(k)_j}\state_2^{(k)_j} + p_4^{(k)_j}\bigg|+ \underline{\gamma}|p_3^{(k)_j} |-\bar{\gamma}| p_3^{(k)_j} |
	\end{align}
	and 
	\begin{align}
		\hamfunc^{(i)_j}_f(\state, p) &= \left[p_1^{(i)_j} \, v^{(i)_j} \cos \state_4 + p_2^{(i)_j} \, v^{(i)_j} \sin \state_4 + p_3^{(i)_j} u_{z_e}^{(i)_j} + p_4^{(i)_j}\, \langle w_e^{(i)_j}\rangle_r\right].
	\end{align}
	Therefore, the main equation \eqref{eq:HamiltonianOverall} becomes 
	\begin{align}
		&\hamfunc(\state, p)  = \cup _{j=1}^{n_f} \left( \cup _{i=1}^{n_a-1} 
		\left[p_1^{(i)_j} \, v^{(i)_j} \cos \state_4 + p_2^{(i)_j} \, v^{(i)_j} \sin \state_4 + p_3^{(i)_j} u_{z_e}^{(i)_j} + p_4^{(i)_j}\, \langle w_e^{(i)_j}\rangle_r\right] \cup_{k_j} \right. \nonumber \\ 
		& \qquad \left.
		\left[
		p_1^{(k)_j} \left(v_e^{(k)_j} - v_p^{j} \cos \state_4^{(k)_j}\right) - p_2^{(k)_j} v_p^{j}  \sin \state_4^{(k)_j} - \underline{w}_p^j |p_4^{(k)_j}| +  \bar{w}_e^j \bigg|p_2^{(k)_j} \state_1^{(k)_j} 	- p_1^{(k)_j}\state_2^{(k)_j}
		\right. \right. \nonumber \\ 
		& \qquad\qquad \left. \left.
	 + p_4^{(k)_j}\bigg|-{\gamma}|p_3^{(k)_j} |-{\gamma}| p_3^{(k)_j} |
		\right] \right).
	\end{align}
\end{proof}

\begin{remark}
For the special case where the linear speeds of the evading agents and pursuer are equal \ie $v_e^{(i)_j}(t) = v_p(t) = +1 m/s$, a murmuration's Hamiltonian reduces to 
\begin{align}
	&\hamfunc(\state, p)  = \cup _{j=1}^{n_f} \left( \cup _{i=1}^{n_a-1} 
	\left[p_1^{(i)_j} \cos \state_4 + p_2^{(i)_j} 	\, \sin \state_4 + p_3^{(i)_j} u_{z_e}^{(i)_j} + p_4^{(i)_j}\, \langle w_e^{(i)_j}\rangle_r\right] \cup_{k_j} \right. \nonumber \\ 
	& \qquad \left.
	\left[
	p_1^{(k)_j} \left(1 - \cos \state_4^{(k)_j}\right) - p_2^{(k)_j}  \sin \state_4^{(k)_j} - \underline{w}_p^j |p_4^{(k)_j}| +  \bar{w}_e^j \bigg|p_2^{(k)_j} \state_1^{(k)_j} 	- p_1^{(k)_j}\state_2^{(k)_j}
	\right. \right. \nonumber \\ 
	& \qquad\qquad \left. \left.
	+ p_4^{(k)_j}\bigg| -{\gamma}|p_3^{(k)_j} | - {\gamma}| p_3^{(k)_j} |
	\right] \right).
	\label{eq:ham_murmur}
\end{align}
\end{remark}

\subsubsection{Vacuole Nucleation Topology}
\label{app:vacuole}

When an attacking predator penetrates (e.g. a Peregrine Falcon~\citep{Ballerini1232}) the interior of a flock's backward
reachable tube~(BRT), a topological hole or a \emph{vacuole}~\citep{vanleeuwenhoek1800select} nucleates in
the zero-sublevel set $\mc{L}(t) \triangleq \{x : v(t,x) \leq 0\}$,  an event characterized by the Euler characteristic.

\begin{definition}[Euler Characteristic of the BRT]
	\label{def:euler_char}
	Let $\mc{L}(t)$ be the BRS at time~$t$, discretised on a uniform
	$m\times m$ grid with vertex, edge, and face counts $V$, $E$, $F$, respectively.  The
	\emph{Euler characteristic} ~\citep{Hatcher2002EulerChar} is
	\begin{align}
		\label{eq:euler_char}
		\chi(t) \;\triangleq\; V - E + F \;\triangleq\; \beta_0(t) - \beta_1(t),
	\end{align}
	where $\beta_0$ is the number of connected components and $\beta_1$ is the
	first Betti number (number of independent loops/holes).
\end{definition}

\begin{theorem}[Vacuole Nucleation Topology]
	\label{thm:vacuole_nucleation}
	 Suppose that the evader agent's  Hamiltonian
	$\hamfunc_f(x, p)$  satisfies $\hamfunc_f(x^*, p^*) = 0$
	for some interior state $x^* \in \Omega$ with costate $p^* = Dv(t^*, x^*)$. 	Then at time $t^*$ when the zero-crossing occurs, if the predator penetrates
	the interior of the flock's reachable set, the topology of $\mc{L}(t^\star)$
	undergoes a transition characterized by,
	\begin{align}
		\label{eq:vacuole_threshold}
		\chi(t^{*+}) \;=\; \chi(t^{*-}) - 1,
	\end{align}
	where $\chi$ is the Euler characteristic~\eqref{eq:euler_char}. This jump
	indicates the nucleation of a vacuole (topological hole) in the flock's safe set,
	marking the transition from simply connected ($\beta_1 = 0$) to multiply connected
	($\beta_1 \geq 1$) topology.
\end{theorem}

\begin{proof}
	\label{proof:vacuole_nucleation}
	By Morse theory for sublevel sets of smooth functions, the topology of $\Omega^{-}(t)$
	changes only at critical values of $v(\cdot, t)$.  Specifically, when the parameter
	$t$ crosses a value $t^*$ at which a saddle point of $v(\cdot, t^*)$ lies on the
	boundary of $\Omega^{-}(t^*)$, the sublevel set attaches a 1-handle.
	
	The attachment of a 1-handle increases the first Betti number by one: $\beta_1(t^{*+}) = \beta_1(t^{*-}) + 1$.
	Since the Euler characteristic is defined as $\chi = \beta_0 - \beta_1$ (Definition~\ref{def:euler_char}),
	this increase in $\beta_1$ directly reduces $\chi$ by one:
	\begin{align*}
		\chi(t^{*+}) &= \beta_0(t^{*+}) - \beta_1(t^{*+}) \\
		&= \beta_0(t^{*-}) - (\beta_1(t^{*-}) + 1) \\
		&= \chi(t^{*-}) - 1.
	\end{align*}
	The condition $H_{\mathrm{att}}(x^*, p^*) = 0$ identifies the critical state-costate pair
	where the attacked agent's optimal trajectory encounters the reachable set boundary.
	At this instant, the saddle-point singularity of $v$ causes the topological change described above.
	The geometric interpretation is that the predator has penetrated sufficiently deep into the flock
	to create an enclosed region, manifested as a topological hole in the backward reachable set.
\end{proof}

\subsubsection{Flock Splitting Dynamics}
\label{sec:d39_splitting}

When a predator applies sufficient control to reduce the attacked agent's
optimal escape velocity, the BRT bifurcates into disjoint components.
This bifurcation corresponds to the moment when the attacking predator can
successfully isolate one or more agents from the main group.

\begin{proposition}[Bifurcation Condition for Flock Splitting]
	\label{prop:bifurcation}
	Let $\Omega^{-}(t)$ be the zero-sublevel set of the BRT at time~$t$.
	The reachable set undergoes a bifurcation into $n_{\mathrm{comp}}(t) > 1$
	connected components if and only if the attacked-agent Hamiltonian satisfies
	the control-constrained minimax condition,
	\begin{align}
		\label{eq:bifurcation_condition}
		H_{\mathrm{att}}(x^*, p^*) \;=\; \min_{u_p \in \mathcal{U}_p} \max_{u_e \in \mathcal{U}_e}
		\left[
		p_1(u_p - v_e) + p_2 u_e + p_3(u_p^{z} - u_e^z) + p_4(\omega_p - \omega_e)
		\right] \;<\; -\delta_c,
	\end{align}
	where $\delta_c > 0$ is a separation margin, $\mathcal{U}_p, \mathcal{U}_e$ are
	admissible control sets, and $(u_p, \omega_p)$, $(u_e, \omega_e)$ are the
	pursuer and evader acceleration and angular velocity controls respectively.
\end{proposition}

\begin{remark}[Temporal Bifurcation Pattern]
	\label{rem:bifurcation_time}
	The bifurcation occurs at a critical time $t_{\mathrm{split}}$ determined by
	when the predator's control authority first permits the pursuer to violate
	the costate constraint in the evader's optimal strategy. The sequence of
	bifurcation times is monotone increasing: $t_{\mathrm{split}}^{(1)} < t_{\mathrm{split}}^{(2)} < \cdots$,
	with the $m$-th bifurcation occurring when isolation of the $m$-th agent
	becomes achievable under the game-theoretic equilibrium. The flock's resilience
	to splitting depends on the coupling strength in the risk-reduction term
	(the $p_4$ coefficient in the Hamiltonian), which represents the predation
	risk shared across connected components.
\end{remark}

\subsubsection{Cordon Formation and Capture Prevention}
\label{sec:d310_cordon}

A flock executing a defensive \emph{cordon formation} creates an annular
perimeter around threatened agents, preventing isolated predator incursions into
the interior. This configuration is geometrically characterized by a multiply
connected reachable set with one independent loop ($\beta_1 = 1$).

\begin{proposition}[Annular BRT and Barrier Certificate]
	\label{prop:cordon}
	Let $\mathcal{C}(r_{\mathrm{in}}, r_{\mathrm{out}}) \triangleq \{x \in \Omega : r_{\mathrm{in}}
	\leq \|x_{1:2}\| \leq r_{\mathrm{out}}\}$ denote the annular barrier region (``cordon'')
	with inner radius $r_{\mathrm{in}}$ and outer radius $r_{\mathrm{out}}$.
	If the BRT restricted to $\mathcal{C}$ satisfies,
	\begin{align}
		\label{eq:cordon_euler}
		v(t, x) \;>\; 0 \quad \text{for all } x \in \{y \in \mathcal{C} : \|y_{1:2}\| < r_{\mathrm{in}}\},
	\end{align}
	then the inner region is a \emph{barrier-certificate-protected domain}, and any trajectory of
	the attacked agent passing through $\mathcal{C}(r_{\mathrm{in}}, r_{\mathrm{out}})$ remains captured
	(i.e., reaches the terminal set) within finite time. The topology of $\Omega^{-}(t) \cap \mathcal{C}$
	is annular, with Betti number $\beta_1 = 1$ and Euler characteristic,
	\begin{align}
		\label{eq:cordon_euler_char}
		\chi\bigl(\Omega^{-}(t) \cap \mathcal{C}\bigr) \;=\; 0.
	\end{align}
\end{proposition}

\begin{remark}[Barrier Certificate Interpretation]
	\label{rem:barrier_cert}
	The annular topology (Eq.~\eqref{eq:cordon_euler_char}) physically represents
	a mutual-defense configuration where the safe set (valued region, $v \leq 0$)
	wraps around the predator or protected interior, creating a topological
	obstruction to escape. From the reachability perspective, any agent trapped
	inside the cordon cannot escape; from outside, no agent can cross the barrier
	without incurring catastrophic cost. This geometric property is independent
	of the detailed dynamics and depends only on the level-set structure of $v$.
\end{remark}

The numerical validation of the cordon configuration is visualised through
the reachability level sets of our simulations, which show the barrier topology explicitly.

\subsubsection{Flash Expansion Events}
\label{sec:d311_flash}

The \emph{flash expansion} is characterised by a rapid,
isotropic outward motion of the flock boundary in response to a predator
approach.  In the BRT framework, this corresponds to a linear growth of the
zero-level-set radius.

\begin{lemma}[Bounded Flash Expansion Rate]
	\label{lem:flash_rate}
	Under the 4D aerial dynamics~\eqref{eq:DubinsRelative} with evader linear speed $v_e$, the maximum radial
	extent of the BRT zero-level-set in the horizontal plane grows at most linearly,
	\begin{align}
		\label{eq:flash_rate}
		\max_i \lVert x_i(t+1)_{1:2} \rVert
		- \max_i \lVert x_i(t)_{1:2} \rVert
		\;\le\; v_e\, \Delta t,
	\end{align}
	per backward time step of size $\Delta t$.  The altitude component $x_3$ does not contribute to the radial
	growth rate because the capture set is a cylinder --- the terminal cost excludes the heading and altitude coordinates from the capture metric.
\end{lemma}

\begin{proof}
	\label{proof:flash_rate}
	By the comparison principle for viscosity solutions, $v(t,x) \leq 0$ implies
	$v(t+\tau, x + \tau f(x,u)) \leq 0$ for any admissible trajectory.  The
	maximum horizontal speed of the free-agent dynamics is $v_e$ (linear speed
	bound), giving a radius growth of at most $v_e \Delta t$ per step.
\end{proof}

The 4D flash-expansion rate converges to the 3D rate as altitude variance
$\sigma^2_{x_3} \to 0$, consistent with the altitude-decoupling property of
the cylinder terminal cost.

\subsubsection{Phase Transition Markers in Numerical Solutions}
\label{sec:d312_markers}

We detect phase transitions computationally by monitoring the topological
invariants $\chi(t)$, $\beta_1(t)$, and $n_{\mathrm{comp}}(t)$ of the
zero-sublevel set across backward time steps.  The detection algorithm is:

\begin{enumerate}
	\item At each step~$t$, compute the 2D slice $v(t, x_1, x_2;\bar{x}_3, 0)$
	on a grid of resolution $m \times m$.
	\item Label connected components of $\{v \leq 0\}$ via
	\textsc{scipy.ndimage.label}.
	\item Estimate $\beta_1$ by counting encircled complementary components.
	\item Compute $\chi(t) = n_{\mathrm{comp}}(t) - \beta_1(t)$.
\end{enumerate}

The threshold-crossing rules that trigger an event label are:
\begin{align}
	\label{eq:vacuole_marker}
	\Delta\chi_t \;\triangleq\; \chi_{t+1} - \chi_t \;<\; -0.5
	&\quad\Longrightarrow\quad \text{vacuole nucleation,} \\
	\label{eq:cordon_marker}
	\beta_1(t) \;\geq\; 1 \;\text{ and }\; \beta_1(t-1) = 0
	&\quad\Longrightarrow\quad \text{cordon formation,} \\
	\label{eq:frag_marker}
	n_{\mathrm{comp}}(t) \;>\; n_{\mathrm{comp}}(t-1)
	&\quad\Longrightarrow\quad \text{flock fragmentation.}
\end{align}

The phase-transition classification partitions the state space into five canonical murmuration phases: (i)
cohesion, (ii) evasion, (iii) cordon, (iv) expansion, and (v) fragmentation.
Each detected event is logged with its
time index, event type, and $\chi$, $\beta_1$, $n_{\mathrm{comp}}$ values, so that the topology metrics are traced across the full backward horizon.

\subsubsection{Simulation Setup and Observed Transitions}
\label{sec:d313_setup}

This case tests scale along two axes at once. Each flock carries a four-dimensional relative state, so the value solve sits above the dimension where a dense grid stays affordable; and the population is large, so a certificate whose cost grows with the number of birds would be of no use in the field. The scheme separates the two concerns. A value function is solved once per flock from the relative dynamics \eqref{eq:DubinsRelative} with the flock Hamiltonian \eqref{eq:ham_murmur}, and no grid is stored. The population enters afterwards, when the sign of that value is read at each bird's own state.

\paragraph{Setup.} We certify $99{,}996$ starlings, organised into six flocks of $16{,}666$ birds. Seven pursuers are placed on a ring of radius $1.15$ about the flock centres, and their capture radius $r_c = 0.6$ enters the datum as the union of capture cylinders in \eqref{eq:sdf_union}. We take viscosity $\delta = 0.18$ and draw $N = 1600$ Gaussian samples at each evaluation state. The frozen-coefficient iteration is capped at seven passes with a relative tolerance of $10^{-5}$, which suffices here because the min-max Hamiltonian of \eqref{eq:ham_murmur} evaluates in closed form under box-constrained climb rate and turn rate. Values are read on a $128 \times 128$ window over $[-3.5, 3.5]^2$ at fixed altitude and heading, at $18$ points along the backward horizon $\tau \in [0, 2]$. Sampling speckle is smoothed with a Gaussian of $1.5$ grid cells before the topology is measured, so that the invariants of \eqref{eq:euler_char} are not read off single-cell noise. The whole study, comprising the eighteen value solves and the certification of every bird, runs in $429$ seconds on eight CPU cores with no GPU.

\paragraph{Certifying the population.} A bird is certified by evaluating the cached value at its own four-dimensional state, so the cost of certifying a flock is linear in the number of birds, with no joint state space anywhere in the loop, and the queries are independent of one another. At the end of the horizon, $87.7\%$ of the $99{,}996$ birds carry a positive value and are therefore certified outside the capture set, while the remainder lie inside it or within the band where the sign is not resolved. The solve itself never represents the population, which is what keeps the memory at $O(N n)$ while the bird count stays free to grow. This is the operational reading of the scalability claim: the binding limit is evaluation throughput, not the reachability computation.

\paragraph{What the run shows.} \autoref{fig:murmur_numerics} traces the zero-sublevel set backwards in time. A hole is present in the set from the start of the horizon through $\tau = 1.18$, so that $\beta_1 = 1$ and $\chi = 0$; by the rule in \eqref{eq:cordon_marker} the set is annular and these steps are labelled a cordon. The hole closes at $\tau = 1.29$, where $\beta_1$ falls to zero and $\chi$ steps up to one. Two invariants therefore witness one event, which is a useful check that the transition is topological rather than an artefact of the sampling. The component count holds at $n_{\mathrm{comp}} = 1$ over the whole horizon, so this population does not fragment and the rule in \eqref{eq:frag_marker} never fires. The closure time is localised only to within the sampling interval of the backward horizon: re-running the same solve on a finer grid of $24$ steps places the same event one interval later, and that run also shows a two-component reading near $\tau \approx 0.8$ that reverts within one step, which we read as sampling noise rather than a split. Reading the three integers on their own reports both that the safety margin is changing and the manner of the change, which is what makes them usable as an operator-facing summary rather than a diagnostic to be inspected offline.

\begin{figure}[tb!]
	\centering
	\begin{tabular}{c}
		\includegraphics[width=\linewidth]{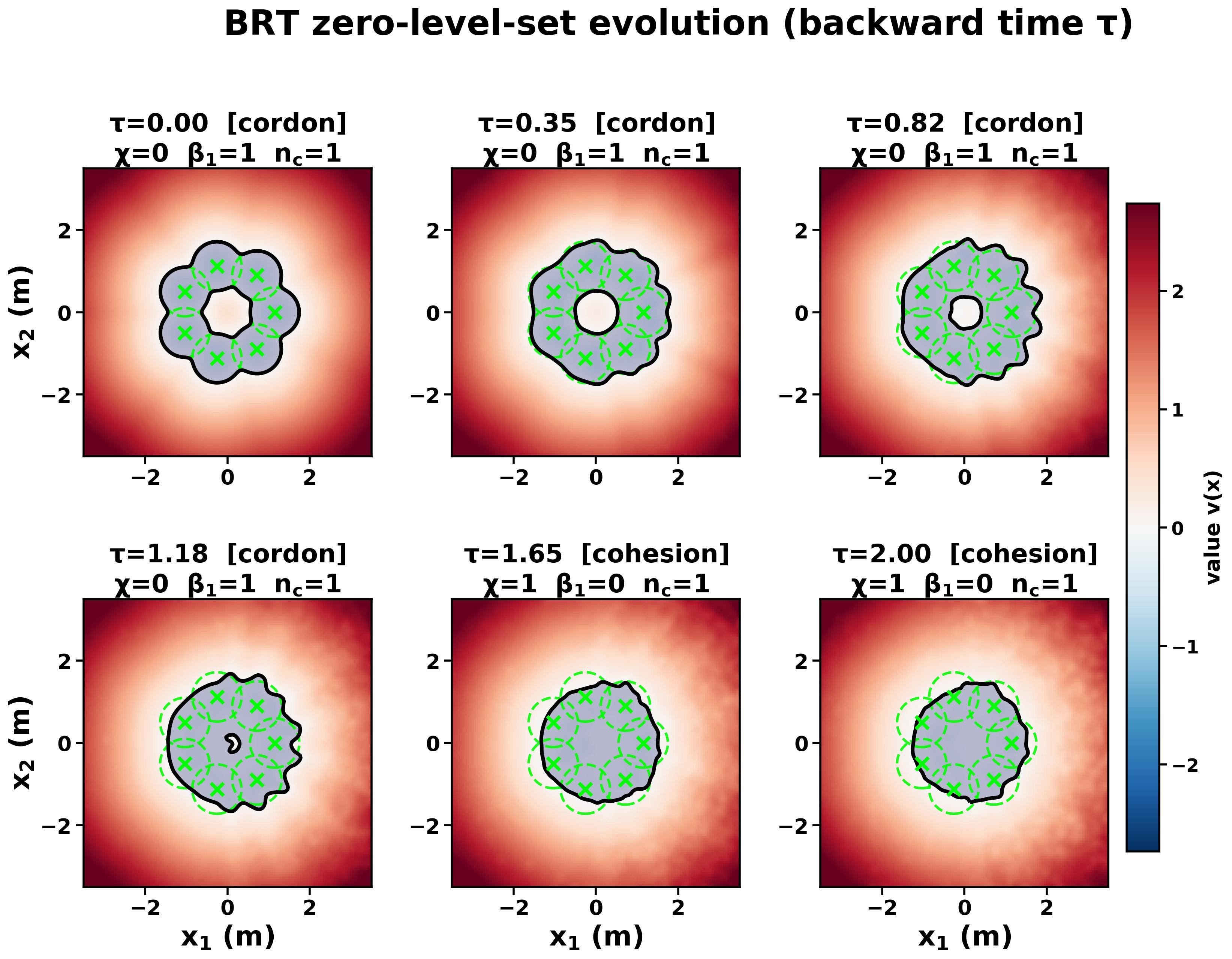}
		\\[0.3em]
		\includegraphics[width=\linewidth]{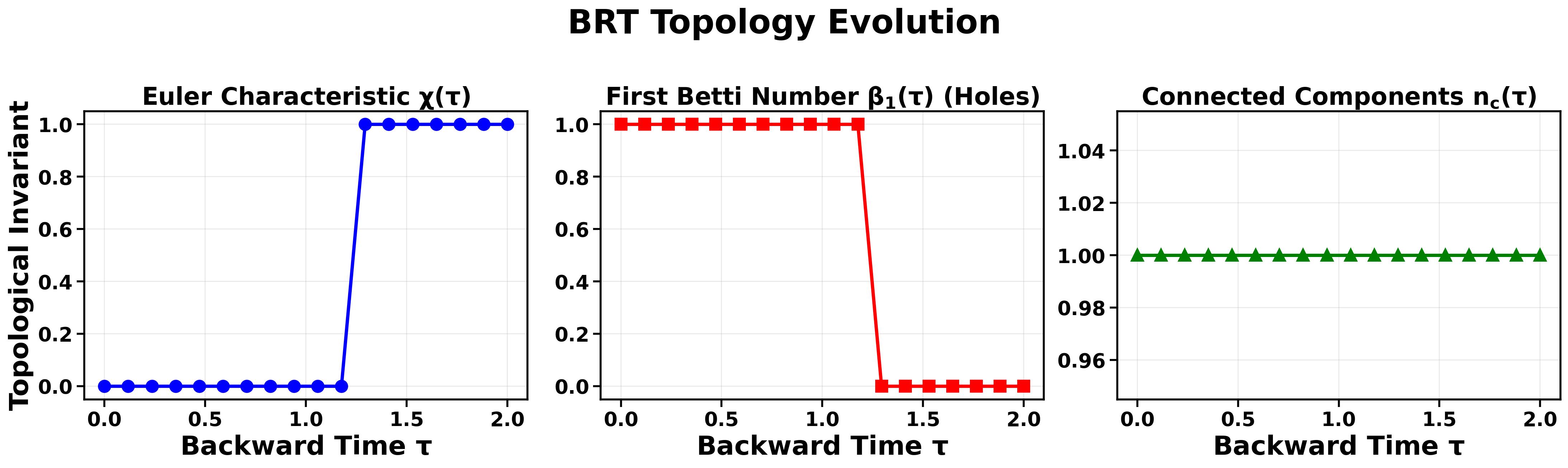}
	\end{tabular}
	\caption{Safety certificate for $99{,}996$ starlings in six flocks under attack by seven pursuers. \textit{Top:} six instants of the zero-sublevel set for an attacked flock, annotated with the phase label and the triple $(\chi, \beta_1, n_{\mathrm{comp}})$. Green crosses and dashed circles mark the pursuers and their capture radii, and the black curve is the $v = 0$ boundary. \textit{Bottom:} the same run read through its topological invariants. The step in $\chi$ and the drop in $\beta_1$ occur together at $\tau = 1.29$, while $n_{\mathrm{comp}}$ holds at one, so the flock closes its annulus without splitting.}
	\label{fig:murmur_numerics}
\end{figure}

\begin{figure}[tb!]
	\centering
	\includegraphics[width=\linewidth]{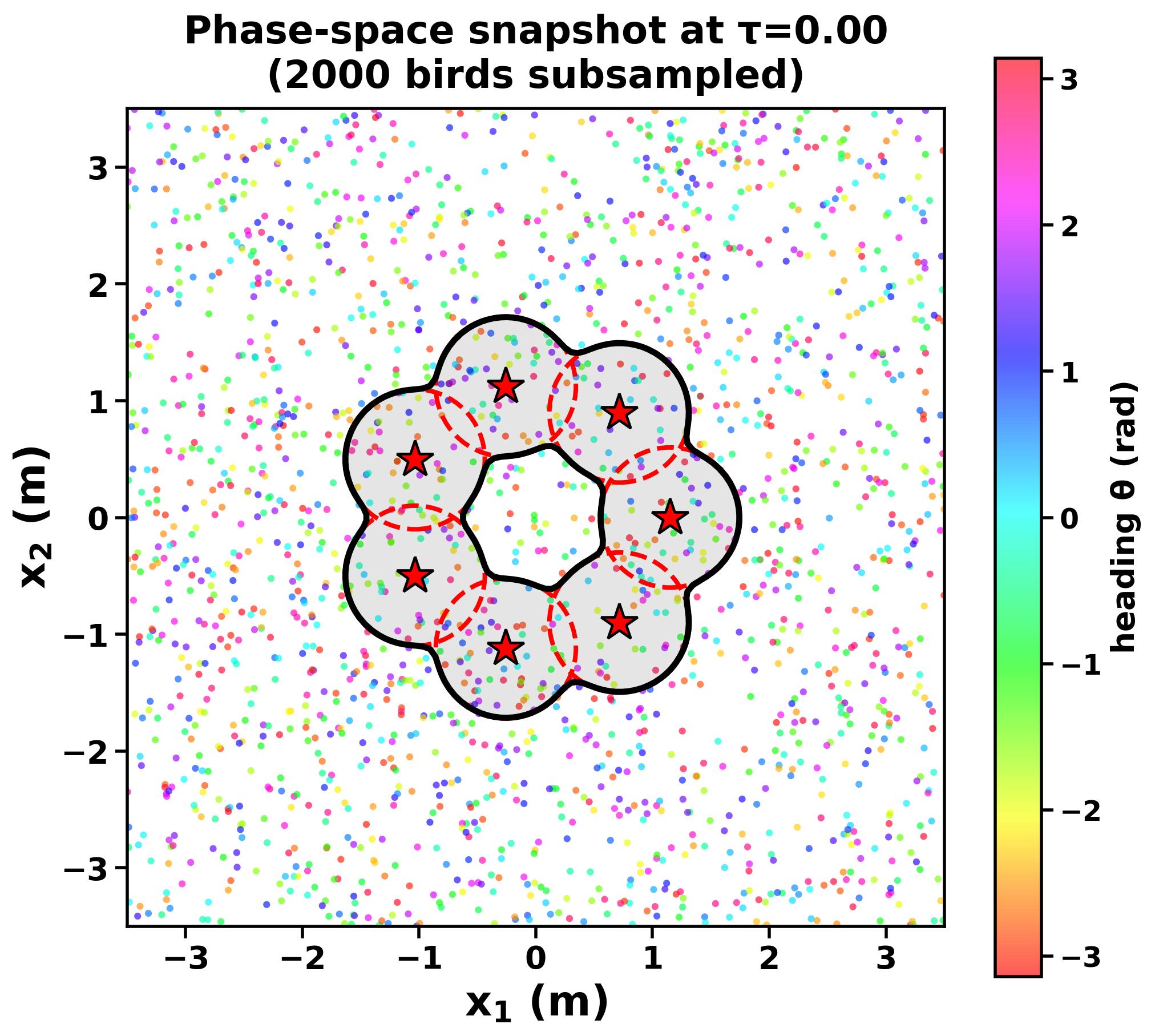}
	\caption{Phase-space view of the certified population at $\tau = 0$, with $2{,}000$ of the $99{,}996$ birds drawn and coloured by heading. Stars mark the flock centres and dashed circles their capture radii; the black curve is the boundary of the zero-sublevel set, which is annular here and so records a cordon. The colour spread shows that the population is a heading distribution rather than a rigid formation, and each bird is certified by one value evaluation at its own state.}
	\label{fig:murmur_phase_space}
\end{figure}
\begin{figure}[tb!]
	\centering
	\includegraphics[width=\linewidth]{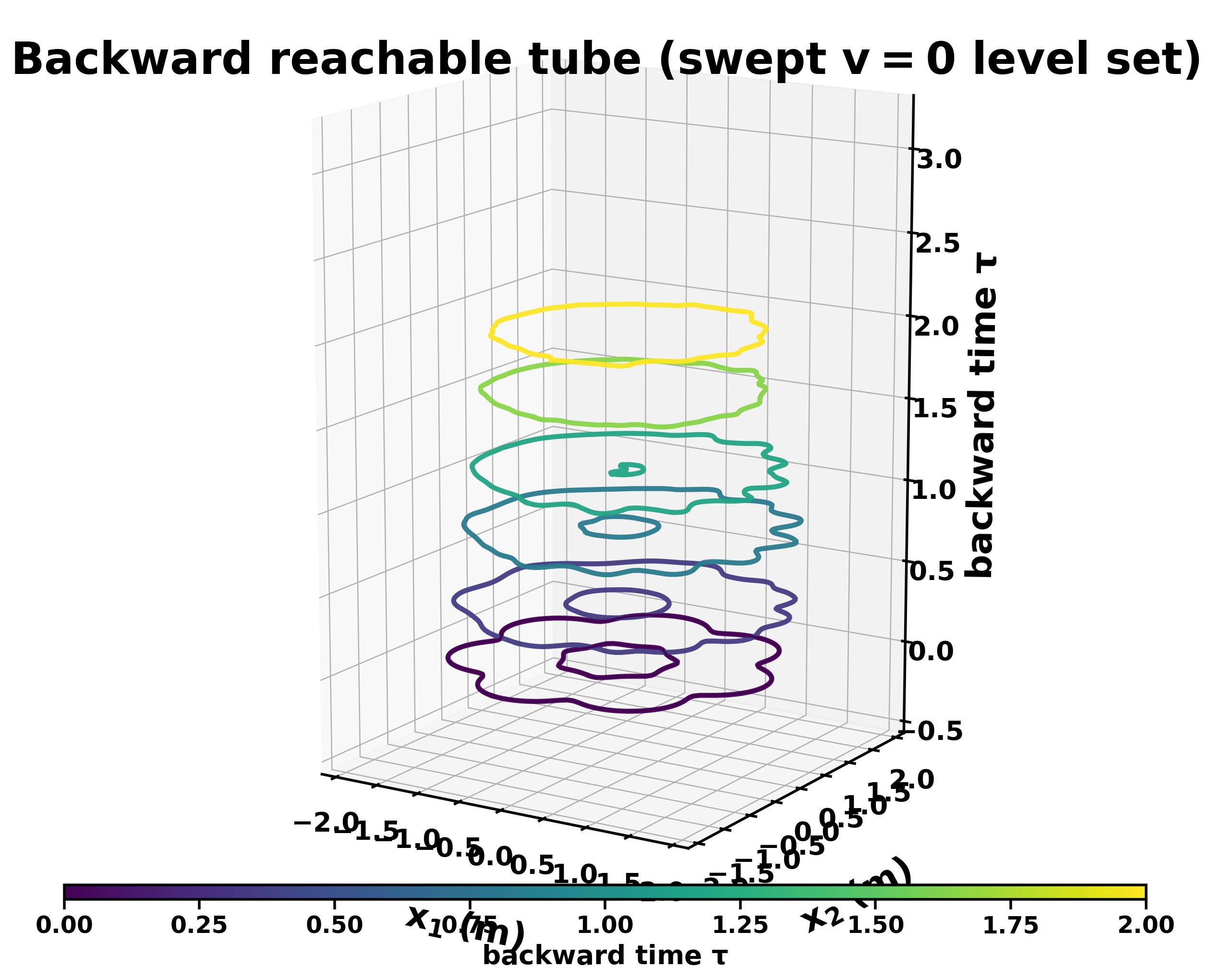}
	\caption{The same certificate viewed as a swept tube. Each contour is the $v = 0$ boundary of the attacked flock at one backward time, stacked along the vertical axis, so the annulus of the cordon appears as a hollow column that closes near the top of the horizon. This is the object the slices of \autoref{fig:murmur_numerics} cut through, and it is why the certified set is a reachable \emph{tube} rather than a set: a state that enters the capture region stays inside it for the remainder of the horizon.}
	\label{fig:murmur_tube3d}
\end{figure}

\section{Limitations}
\label{sec:limitations}
We state the limitations of the presented theory in this section.

\begin{inparaenum}[(i)]
	\item \emph{Sampling overhead in moderate dimensions}: Although the $O(N \cdot n)$ scaling becomes favorable as $n$ increases, the absolute cost of $N = 14{,}000$ samples per iteration can be higher than structured grid methods for $n \le 4$; practitioners should use grid-based solvers for low-dimensional problems where memory is not the bottleneck.

	\item \emph{Viscosity parameter sensitivity}: The choice of $\delta$ controls the smoothing level and approximation error $O(\sqrt{\delta})$; suboptimal $\delta$ selection can degrade accuracy or require excessive samples. Grid-search or Bayesian optimization is recommended during tuning.

	\item \emph{Convergence dependent on regularity}: The Picard iteration convergence guarantee (Theorem~\ref{thm:algorithm_convergence}) requires Lipschitz Hamiltonians and bounded second derivatives; nonconvex or highly nonlinear Hamiltonians may converge slowly or not at all. The scheme is at its most comfortable when the coefficient $\bc = 2\hamfunc/(\delta|D\valuefunc|^2)$ varies slowly (small $L_c$), and is exact for the quadratic Hamiltonian $\hamfunc=\tfrac12|p|^2$; conversely, we expect it to struggle where $\bc$ varies violently.

	\item \emph{Sample complexity is exponential in $1/\delta$ in the worst case} (Remark~\ref{rem:sample-bound-exponential}): with $\bc=1/\delta$, the explicit bound of Corollary~\ref{cor:mc_complexity} scales as $N \gtrsim \exp\big(2(g_{\max}-g_{\min})/\delta\big)$, since $\beta/\alpha = \exp((g_{\max}-g_{\min})/\delta)$ enters the bound squared. This is in direct tension with the Crandall-Lions viscosity error $O(\sqrt\delta)$~\citep{CrandallTwoApprox}: driving $\delta\to0$ for better viscosity accuracy drives the worst-case required sample count up exponentially, and Corollary~\ref{cor:mc_complexity} alone does not certify a favorable trade-off at small $\delta$. Concretely, for the Rockets/Dubins geometry ($g_{\max}-g_{\min}\approx 7$, $\delta=0.08$) the bound's exponent is in the hundreds --- utterly impractical if taken literally. Two things keep the method usable in practice despite this. First, the bound is a worst-case Hoeffding-type estimate that ignores variance (Remark~\ref{rem:tightness-bernstein-jensen}): a Bernstein bound, or a sharper (Jensen) estimate of $\mu$, can reduce the required $N$ substantially when $\mathrm{Var}(Z)$ is well-controlled, which we observe empirically to be the case (the 30-seed statistics of \autoref{tab:error_comparison} reject $L^2_{\text{rel}}\ge\sqrt\delta$ at $p_{\text{holm}}<10^{-50}$ using only $N=14$-$20$k samples, far below what the worst-case bound would suggest is needed). Second, we do not have a general class of Hamiltonians or cost functions for which the bound's constants ($g_{\max}-g_{\min}$ in particular) are controlled uniformly in $\delta$; establishing one is open. Practitioners should treat Corollary~\ref{cor:mc_complexity} as a qualitative warning against driving $\delta$ too small for a fixed sample budget, not as a tight prescription for $N$.

	\item \emph{Gradient estimation noise}: The spatial gradient estimate \eqref{eq:spatial_grad} relies on Monte Carlo variance in the log-sum-exp; in high-viscosity regimes or when samples become concentrated in narrow regions, gradient estimation can become unstable.

	\item \emph{Limited applicability to learned approximators}: The theoretical guarantees (Theorems~\ref{thm:mc_complexity} and~\ref{thm:algorithm_convergence}) assume exact Hamiltonian evaluation. Extension to neural-network-approximated Hamiltonians (e.g., in learned dynamics models) introduces additional approximation error not quantified by our theory.

	\item \emph{Single-core CPU evaluation}: All reported wall-clock times use a single CPU core of a 20-core Intel i7-14700K, leaving the machine's other 19 cores and its NVIDIA RTX A2000 GPU unused. The per-point Monte Carlo estimator (Algorithm~\ref{alg:quasi_lin}) is embarrassingly parallel across evaluation states and samples, so multi-core and GPU execution should reduce wall-clock time substantially; we have not yet benchmarked that scaling.

	\item \emph{Hamiltonians that are well-suited to our formulation}: Our presented solution is well-suited to  smooth, possibly nonconvex two-player game Hamiltonians, i.e., pursuit-evasion games, including Dubins-type, unicycle dynamics, 4D/3D dynamical vehicles where the co-state direction varies gradually and $L_c$ stays small. It may struggle on systems with Hamiltonians that possess rapid or discontinuous co-state dependence, such as bang-bang optimal control problems, where $\mathbf{c}$ varies sharply; this failure mode is inherited from viscous HJ theory rather than being specific to our frozen-coefficient scheme. This is because such controls fall outside the regularity conditions where the underlying viscous HJ solution applies at all. 
	
	\item \emph{Extensions to Stochastic Dynamics}: 	The Feynman-Kac expectation representation that underlies our value-expectation formula is the cost-to-go of a controlled diffusion $d\mathbf{x} = f\,dt + \sqrt{\Delta}\,d\mathbf{W}$ (read $f$: drift, ${\Delta}$: covariance). Here,  $\Delta$ \textit{is} the additive process noise's covariance, not just a numerical regularizer. Thus, the method already solves a stochastic reachability problem with additive Gaussian state noise; sampling $\mathbf{y}\sim\mathcal{N}(\mathbf{x},\Delta t I)$  simulates that diffusion's terminal state.
	
	\item \emph{Limitations in multiplicative noise dynamics}: What does not follow for free is control- or state-dependent (multiplicative) noise. There, the Cole-Hopf coefficient becomes a function of the diffusion matrix, and removing the resulting residual would need the same Picard quasi-linearization we already use for general Hamiltonians (\autoref{sec:methods}), rather than an exact substitution. Thus, additive-noise stochastic dynamics are already covered; multiplicative or control-dependent noise is a natural but nontrivial extension along the lines of our existing quasi-linearization machinery, not something the Cole-Hopf transform fundamentally precludes.
\end{inparaenum}

%
\end{document}